\documentclass[12pt]{article}
\usepackage[utf8]{inputenc}
\pdfoutput=1

\usepackage[margin=1.25in]{geometry}
\usepackage[onehalfspacing]{setspace}

\usepackage{fullpage}
\usepackage{booktabs} 
\usepackage[ruled]{algorithm2e} 

\SetAlFnt{\small}
\SetAlCapFnt{\small}
\SetAlCapNameFnt{\small}
\SetAlCapHSkip{0pt}
\IncMargin{-\parindent}

\usepackage{amsfonts,amsmath,amssymb,amsthm}
\usepackage{verbatim}
\usepackage{graphicx}
\usepackage{xcolor}
\usepackage{bbm}
\usepackage{hyperref}
\usepackage{cleveref}
\usepackage{nicefrac}       
\usepackage{xfrac}
\usepackage{algorithmic}
\newcommand{\Comment}[1]{{\footnotesize\ttfamily\textcolor{blue}{/* #1 */}}}

\usepackage{tikz}
\usetikzlibrary{positioning,arrows.meta,calc,fit,backgrounds}

\usepackage[suppress]{color-edits}  
\addauthor{wt}{blue}
\addauthor{yl}{red}
\addauthor{yd}{cyan}

\usepackage{enumitem}
\setlist{nosep} 
\usepackage[font=small,skip=2pt]{caption}

\newtheorem{proposition}{Proposition}
\newtheorem{lemma}{Lemma}
\newtheorem{theorem}{Theorem}
\newtheorem{corollary}{Corollary}

\theoremstyle{definition}
\newtheorem{definition}{Definition}

\newtheorem{remark}{Remark}

\newcommand{\ValSet}{\mathcal{V}}
\newcommand{\val}{v}
\newcommand{\valnum}{n}
\newcommand{\CostSet}{\mathcal{C}}
\newcommand{\cost}{c}
\newcommand{\costnum}{m}
\newcommand{\costpmf}{g}
\newcommand{\ValDist}{F}
\newcommand{\valpmf}{f}
\newcommand{\MarketSegment}{\ValDist_{d}}
\newcommand{\marketSegmentPmf}{\valpmf_{d}}
\newcommand{\segWeight}{\alpha_d}
\newcommand{\segi}{d}
\newcommand{\segI}{D}
\newcommand{\AggMarket}{F^*}

\newcommand{\MarketSegSet}{\mathcal{F}}

\newcommand{\FragSuppSet}{T}

\newcommand{\sw}{w}
\newcommand{\esw}{w}
\newcommand{\rsdz}{z}
\newcommand{\rsdw}{x}
\newcommand{\rsdf}{y}
\newcommand{\df}{d}
\newcommand{\mscdf}{p}
\newcommand{\reals}{\mathbb{R}}

\newcommand{\CostDist}{G}
\newcommand{\SuppSet}{S}
\newcommand{\supp}{s}
\newcommand{\ExtMarketDecomp}{\mathcal{D}}

\newcommand{\Frag}{{\bar F}}
\newcommand{\frag}{{\bar f}}

\newcommand{\op}{q}
\newcommand{\opm}{r}
\newcommand{\Op}{Q}
\newcommand{\argmax}{\operatorname*{argmax}}

\newcommand{\ba}{\boldsymbol{\alpha}}
\newcommand{\bw}{\boldsymbol{w}}
\newcommand{\bx}{\boldsymbol{x}}
\newcommand{\by}{\boldsymbol{y}}
\newcommand{\bz}{\boldsymbol{z}}
\newcommand{\ct}{\Rightarrow}
\newcommand{\SellerSurplus}{\textsf{SS}}
\newcommand{\BuyerSurplus}{\textsf{BS}}
\newcommand{\MinBuyerSurplus}{\BuyerSurplus_{\text{min}}}
\newcommand{\MaxBuyerSurplus}{\BuyerSurplus_{\text{max}}}
\newcommand{\SocialWelfare}{\textsf{SW}}
\newcommand{\MinSocialWelfare}{\SocialWelfare_{\text{min}}}
\newcommand{\MaxSocialWelfare}{\SocialWelfare_{\text{max}}}

\usepackage{natbib}
\bibpunct{(}{)}{;}{a}{,}{,}

\title{The Limits of Price Discrimination with a Bayesian Seller\thanks{A one-page abstract appeared in ACM EC'26.  We thank anonymous referees for valuable feedback.}}

\author{
Yuan Deng\thanks{Google Research, \texttt{dengyuan@google.com}.}
\and
Yilin Li\thanks{Chinese University of Hong Kong, \texttt{ylli25@cse.cuhk.edu.hk}.}
\and
Wei Tang\thanks{Chinese University of Hong Kong, \texttt{weitang@cuhk.edu.hk}.}
\and
Hanrui Zhang\thanks{Chinese University of Hong Kong, \texttt{hanrui@cse.cuhk.edu.hk}.}
}
\date{}
\begin{document}
\maketitle

\begin{abstract}
    
We study the limits of third-degree price discrimination when the production cost is Bayesian and private to the seller, generalizing the seminal work of \citet{BBM-15}.  The rough setup is the following: A monopoly seller sets different prices for buyers in different ``segments'' of the market so as to maximize seller surplus.  Different ways in which the aggregate market is decomposed into segments lead to different welfare outcomes, i.e., (seller surplus, buyer surplus) pairs.  When the production cost is Bayesian, the region of achievable welfare outcomes can exhibit complex shapes beyond the clean characterization by Bergemann, Brooks and Morris for the case with a fixed cost.  We show that with a Bayesian cost, this region coincides with a proper projection of a polytope defined by a polynomial number of linear constraints, the essential ones of which correspond to flow conservation in a ``discounted'' flow network.  As a result, we give a polynomial-time algorithm that computes optimal market segmentations in terms of any linear combination of the seller surplus and the buyer surplus.  En route, we establish the following structural property: Any market can be written as a convex combination of ``extremal markets'' in a way preserving the seller surplus and the buyer surplus.  These extremal markets are piecewise equal-surplus with respect to different possible costs, generalizing a similar notion introduced by Bergemann, Brooks and Morris when the cost is fixed.

\medskip

\noindent\textsc{Jel Classification: }D47, D82, D83.

\noindent\textsc{Keywords: }Price discrimination, Bayesian Persuasion, Extremal markets.
\end{abstract}



\newpage
\section{Introduction}
\label{sec:intro}

{\em Price discrimination} is the practice where a monopoly seller sets different prices for buyers in different ``segments'' of the market so as to maximize seller surplus.  The idea is the following: The aggregate market can be summarized as a distribution $F$ of all buyers' willingness to pay, i.e., $F$ is the distribution of a uniformly random buyer's value.  Suppose the production cost is $c$.  Without segmentation, the seller's optimal strategy is to set a price $p^* \in \argmax_p (p - c) \cdot F(p)$, where $F(p)$ is the {\em tail} probability at $p$, i.e., the probability that a uniformly random buyer's value is {\em at least} $p$.\footnote{As a general convention, throughout the paper, we use $F$ to denote the tail probability of a distribution, instead of the CDF.  This simplifies the presentation both syntactically and semantically.}  Nonetheless, if the seller knows that the aggregate market consists of two {\em distinguishable} groups of buyers (i.e., two segments), characterized by distributions $F_1$ and $F_2$ respectively, then the optimal pricing strategy is to offer two generally different prices $p_1^* \in \argmax_p (p - c) \cdot F_1(p)$ and $p_2^* \in \argmax_p (p - c) \cdot F_2(p)$ to the two segments respectively.  The latter strategy is always no worse, and almost always strictly better, than the former segment-oblivious strategy in terms of seller surplus.  This is but natural --- in the latter scenario, the seller makes use of more information, and extracts higher surplus in return.

Perhaps more curious is the fact that {\em buyers} can sometimes also benefit from price discrimination.  Consider the following example: In the aggregate market, $49\%$ of all buyers have a value of $1$, and the other $51\%$ have a value of $2$ (denoted by $\{(0.49, 1), (0.51, 2)\}$ for brevity).  The production cost is $0$.  Without segmentation, the seller sets a price of $2$, leading to a total buyer surplus of $0$ in the entire market.  On the other hand, suppose the market is decomposed into two segments: Segment 1: $\{(0.49, 1), (0.48, 2)\}$, and Segment 2: $\{(0.03, 2)\}$.  Now the optimal strategy for the seller is to set a price of $1$ on Segment 1 and a price of $2$ on Segment 2, leading to a total buyer surplus of $0.48$, contributed by Segment 1.  In other words, the example shows that price discrimination can lead to {\em strong Pareto improvements}.  In fact, an alternative perspective is to view (optimal) market segmentation as an {\em information design} \citep{kamenica2011bayesian} problem, where a mediator selectively discloses information about buyers, effectively creating segments, so as to optimize a certain objective (e.g., a function of the seller surplus and / or the buyer surplus).  This perspective provides further motivation for research on the {\em limits of price discrimination}, where the goal is to understand (in ways elaborated below) all pairs of $(\text{seller surplus}, \text{buyer surplus})$ achievable through segmentation in a given aggregate market.  This is the general context of our investigation in the current paper.

\paragraph{The model: from fixed production costs to Bayesian ones.}
The celebrated result of \citet{BBM-15} gives a surprisingly clean characterization (henceforth the BBM characterization) of the limits of price discrimination in the basic setting, where the production cost is fixed and public: Every welfare outcome, except ``clearly impossible'' ones, is achievable through segmentation.\footnote{We will use the word ``segmentation'' in two senses: (1) the general practice of creating segments in a market, and (2) a specific way of segmentation, i.e., a decomposition of a market into a convex combination of markets.  When used in the second, more technical sense, a segmentation is mathematically similar to an information structure or a signaling scheme.}  More precisely, assuming (without loss of generality) that the production cost is $0$, the region of feasible welfare outcomes is the triangle induced by the following $3$ extremal points:
\begin{itemize}
    \item Welfare-maximizing, seller-optimal: Full efficiency is achieved (i.e., the buyer buys with probability $1$), and the seller extracts the full welfare.  This is trivially achieved by a segmentation that always reveals the buyer's value to the seller.
    \item Welfare-maximizing, buyer-optimal: Full efficiency is achieved, and the seller extracts the minimum surplus possible, which is the optimal seller surplus without any segmentation.
    \item Welfare-minimizing: The social welfare is minimized, which is equal to the minimum seller surplus possible.  The buyer surplus is $0$.
\end{itemize} 
In particular, the achievability of the latter two points (the nontrivial ones) is established constructively based on a powerful structural property: Any aggregate market can be written as a convex combination of ``extremal markets'', each of which ensures that the seller is indifferent among all prices in the support of the extremal market.  In the context of price discrimination, each of these extremal markets can be viewed as a segment of the aggregate market.  As long as the supports of these extremal markets all cover the optimal price(s) without segmentation (which is always possible), the seller surplus must be the minimum possible.  Given such a decomposition, optimal price discrimination essentially becomes a matter of tiebreaking: If the seller favors the minimum price in each extremal market, then full efficiency is achieved and the buyer surplus is maximized (point 2 above); alternatively, if the seller favors the maximum price in each extremal market, then the buyer surplus is $0$, and the social welfare is minimized (point 3 above).

In this paper, we investigate the very question asked by BBM in the more general model where the production cost --- or in other words, the type of the seller --- is Bayesian and private to the seller.  Conceptual interpretations of this model are manifold: When the segmentation is exogenous, one can imagine that the mediator who designs the information structure has only distributional knowledge of the seller's production cost; when the segmentation is endogenous, one can imagine that the production cost realizes independently of the segmentation, or even refreshes every once in a while.  As we will see below, this simple and natural generalization leads to a dramatically richer technical problem.

\subsection{Results and Techniques}

\begin{figure}[h!]
    \centering
    \includegraphics[width=0.5\linewidth]{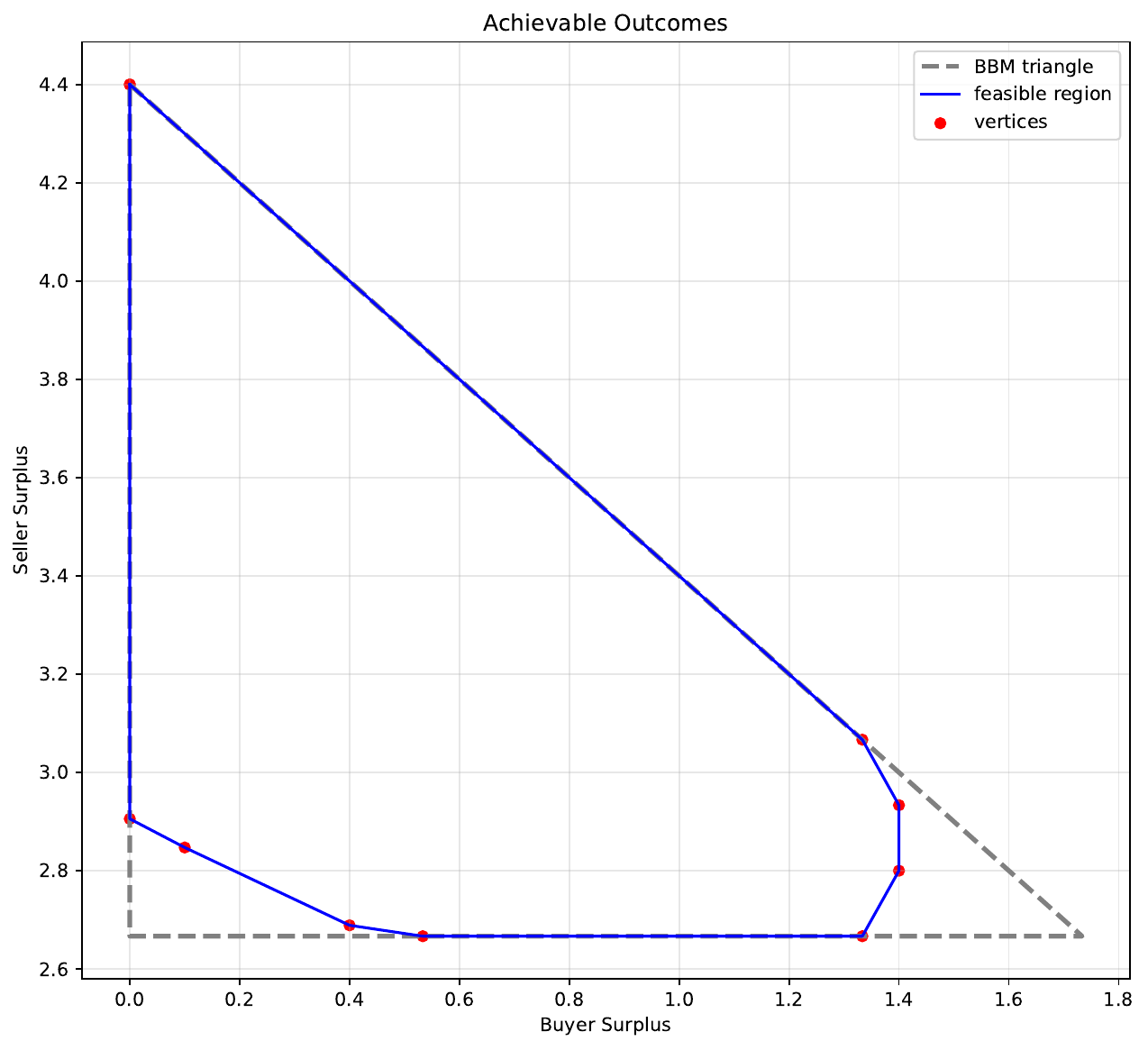}
    \caption{Landscape of achievable welfare outcomes with a Bayesian seller ($5$ possible values, $3$ possible costs).}
    \label{fig:landscape}
\end{figure}

\paragraph{The BBM characterization no longer applies.}
We first seek to answer the following natural question: {\em Does the BBM characterization generalize to the setting with a Bayesian seller?}  The answer turns out to be negative: Among the three extremal points in the BBM characterization, the two nontrivial ones are not always achievable with a Bayesian seller, and the corresponding ``blunt tips'' of the achievable region can exhibit complex shapes (Figure~\ref{fig:landscape} shows such an example).  This complex landscape motivates, and in some sense demands us to take a more sophisticated and {\em algorithmic} approach to the problem.  In particular, we will refrain from trying to explicitly characterize all extremal points in the achievable region (which appears futile given how complex the region can be).  Instead, our full characterization is implicit through an efficient algorithm that computes an optimal segmentation in terms of any linear combination of the seller surplus and the buyer surplus.  En route, we also establish and utilize strong structural properties, which we discuss next.

\paragraph{Extremal markets, generalized.}
Our first structural result, which lays the foundation of all subsequent results in the paper, is the full generality of {\em generalized} extremal markets.  We show that any market can be written as a convex combination of certain highly structured markets, which we also term extremal markets, following BBM --- in fact, BBM's extremal markets are a special case of ours when the distribution of the production cost degenerates into a single point mass.  The decomposition of a market into extremal markets can be done algorithmically in polynomial time, in a way preserving both the seller surplus and the buyer surplus.  As a result:

\begin{theorem}[Informal version of \Cref{corollary: full generality of extremal market}]
Any achievable outcome through segmentation is also achievable through segmentation {\em into extremal markets}, and there is an efficient algorithm that transforms the former into the latter.
\end{theorem}

Technically, to establish this result, we need to overcome a number of obstacles that are unique to the more general setting with a Bayesian seller, compared to the classic one considered by BBM.  To begin with, the right definition of extremal markets is no longer straightforward, since there are multiple possible costs, and it is impossible to equalize the seller surplus generated by different prices with respect to all costs simultaneously.  Regardless of the form of extremal markets, the iterative decomposition must preserve, in each step, the seller-optimal prices with respect to {\em all possible costs simultaneously}.  At the same time, throughout the decomposition procedure, we must always make sure that the ``remainder'' of the market is ``essentially similar'' to the aggregate market, again, with respect to all possible costs simultaneously.  On top of all this, we must be able to bound the number of iterations after which the procedure terminates, which, as we will see, can no longer be done simply by arguing about the size of the support of the remainder market.  As it turns out, the right notion of extremal markets involves dividing the value space into pieces according to the seller-optimal prices with respect to all possible costs (so the set of all legitimate extremal markets depends on both the value space and the cost space).  Each extremal market is piecewise equal-surplus, with respect to different costs on different pieces, respectively.  Our decomposition procedure in each iteration constructs a new extremal market to be split off from higher values to lower ones, where it enters a new phase every time it crosses a seller-optimal price with respect to some possible cost.  We argue that the sets of seller-optimal prices with respect to all possible costs can only expand throughout the procedure, which implies all the properties discussed above.

\paragraph{(Reduced) fragment decompositions.}
One particular implication of the above structural result is that in order to characterize all achievable welfare outcomes, we only need to characterize those achievable by segmentation into extremal markets.  We call such a segmentation an {\em extremal market decomposition}.  Note that there is a straightforward LP that allows us to optimize an extremal market decomposition, where the main decision variables are the weights of all possible extremal markets in the decomposition.  The issue of this LP is that it has exponentially many variables (corresponding to exponentially many possible extremal markets), and is therefore algorithmically infeasible.  This gives us the following intuition that drives our subsequent results: We need to find a more succinct representation to capture the essence of an extremal market decomposition, which still allows us to determine the seller surplus and the buyer surplus induced by this extremal market decomposition.  This motivates the idea of (reduced) fragment decompositions.

Recall that each extremal market consists of multiple pieces, on each of which the market is equal-surplus.  The {\em fragment decomposition} of an extremal market is simply the unique way of writing an extremal market as a weighted sum of these pieces, each normalized so the ``running tail probability'' is $1$ (the normalized version of each possible piece is called an (indifference) fragment).  Then, one can naturally generalize this definition to extremal market decompositions: To obtain the fragment decomposition of an extremal market decomposition, we simply sum up the fragment decomposition of every extremal market, scaled by its weight in the extremal market decomposition.  The rationale behind fragment decompositions is that if multiple extremal markets in an extremal market decomposition share the same fragment, then we only care about how much probability mass {\em in total} is put on this fragment --- in fact, through a charging argument, we show that the seller surplus and the buyer surplus induced by an extremal market decomposition are both linear functions of the weights in its fragment decomposition.

The above result suggests that the welfare outcome induced by an extremal market decomposition is fully captured by its fragment decomposition, which is much more succinct.  Nonetheless, there are still exponentially many possible fragments, in particular because each fragment is determined by its entire support, which in general can be any subset of the value space.  To this end, we further consider {\em undominated} fragments, whose supports are contiguous subsets of the value space.  Each undominated fragment is determined by the leftmost and rightmost element in its support, which means there are only polynomially many undominated fragments.  We then group all fragments that are dominated by the same undominated fragment together, where we focus on the total probability mass contained in these fragments, as well as how the probability is distributed to the elements in the support of the undominated fragment.  We show that these two quantities are enough to capture the contribution of the entire group to the seller surplus and the buyer surplus.  Such a grouped fragment decomposition is called a {\em reduced} fragment decomposition.  As a result, we finally obtain a polynomial-sized representation of extremal market decompositions.

\begin{lemma}[Informal version of \Cref{lem: surplus of reduced fragment decompositions}]
    The seller surplus and the buyer surplus of an extremal market decomposition are fully captured by, and in fact, linear in, its reduced fragment decomposition, which consists of a polynomial number of parameters.
\end{lemma}

\paragraph{Discounted flow networks.}
Ideally, we would like to optimize an extremal market decomposition based on its reduced fragment decomposition.  However, there is yet another key issue to be handled: The above chain of arguments establish that every extremal market decomposition can be summarized by a reduced fragment decomposition, but the converse may not hold.  That is, not every ``formal'' reduced fragment decomposition (as defined by probabilities on reduced fragments and how the probability on each reduced fragment is further distributed) corresponds to an extremal market decomposition.  In order to optimize {\em only} on extremal market decompositions, we need to tell which formal reduced fragment decompositions are ``feasible''.  Below we identify key constraints that are each necessary for feasibility, and then we argue constructively that these constraints together are also sufficient.

First, recall that a reduced fragment decomposition consists of two components: the probability on each reduced fragment, and how it is distributed into each element in the support of the reduced fragment.  Because of the way in which fragments are grouped, the latter ``distribution'' should also be stochastically dominated by the undominated fragment (properly scaled), because it is the weighted average of markets each dominated by the same undominated fragment.  This is a necessary condition of feasibility primarily concerning the second component of reduced fragment decompositions, which we call {\em proper domination}.  Next, we derive another condition concerning the first component of reduced fragment decompositions.

We observe that neighboring pieces in any extremal market are closely related in terms of the ``running tail probability'', i.e., the total probability on values larger than or equal to the smallest element in the support of a piece.  Intuitively, one could view an extremal market as a flow of probability from lower values / pieces to higher ones.  Whenever the flow passes through a piece, it leaves some probability within the piece, and carries the remaining probability onward.  We call the fraction of the probability carried onward the {\em discount factor}.  We show that the discount factor of a piece is fully determined by the undominated fragment that captures the piece in the reduced fragment decomposition.  Accordingly, we must be able to fit any feasible reduced fragment decomposition into a {\em discounted flow network}, where two reduced fragments are adjacent iff there exists an extremal markets in which the two reduced fragments are neighboring.  Moreover, each edge going out of a reduced fragment is associated with the discount factor determined by this fragment.  The fact that a feasible reduced fragment decomposition comes from an extremal market decomposition now translates into {\em flow conservation}, i.e., the total flow into a reduced fragment after discounting must be equal to the total flow out of the same fragment, before discounting.  In other words, a formal reduced fragment decomposition is feasible, only if it can be written as a flow satisfying proper domination, flow conservation, and other regularity constraints.

We then prove constructively that these condition are sufficient for feasibility.  We provide an iterative depth-first search procedure (similar to the Ford-Fulkerson algorithm for max-flow) that decomposes any flow satisfying these conditions into a segmentation (not necessarily an extremal market decomposition), which achieves precisely the same seller surplus and buyer surplus.  Crucially, proper domination ensures that each segment that we construct gives the right seller surplus and the right buyer surplus, and flow conservation ensures that when the algorithm terminates, we must have already exhausted the probability on every reduced fragment.

\begin{lemma}[Informal version of \Cref{lem: the existance of y} and \Cref{lem: market segmentation - x y z feasible}]
    A formal reduced fragment decomposition is feasible (i.e., corresponds to some segmentation inducing the same seller surplus and the buyer surplus) iff it fits into a reduced flow network constrained by proper domination and flow conservation.
\end{lemma}

Now we can close the entire loop: If a welfare outcome is achievable, then there is an extremal market decomposition that induces it, which means the corresponding reduced fragment decomposition fits into the discounted flow network corresponding to the market instance.  In other words, there is a feasible flow that induces this welfare outcome.  Conversely, if there is a feasible flow that induces a certain welfare outcome, then it must be achievable through segmentation, since we can construct algorithmically a segmentation that induces precisely this welfare outcome.  This gives the main result of the paper.

\begin{theorem}[Informal version of \Cref{thm: the equvilence of market segmentation and reduced fragment composition} and \Cref{corollary: polynomial-time algorithm for achievable seller buyer surplus}]
    The set of achievable outcomes of price discrimination with a Bayesian seller is precisely characterized by (a proper projection of) the polynomial-sized polytope induced by the discounted flow network constrained by proper domination and flow conservation.  As a corollary, there is a polynomial-time algorithm that computes an optimal segmentation in terms of any linear combination of the seller surplus and the buyer surplus.
\end{theorem}

\subsection{Further Related Work}
Along the line of work initiated by \citet{BBM-15}, several results are related to ours to different extents.
Most closely related is the recent result on robust price discrimination by \citet{arieli2025robust}.
They generalize the work by \citet{BBM-15} in another direction: The seller's cost is uncertain in a worst-case sense, and the designer's goal is to minimize the regret, i.e., the suboptimality against the ideal objective value when the cost is known, in terms of the buyer surplus.
They present an upper bound where the regret is $1/e$ of the optimal buyer surplus when the cost is $0$, and show that this is tight for a binary buyer.
While conceptually related, their result is not directly comparable to ours, because they consider a worst-case problem, while we consider a Bayesian one.\footnote{
    In fact, \citet{arieli2025robust} also acknowledge that ``another natural model that might be considered is the Bayesian one.''
}
\citet{cummings2020algorithmic} study a model where the designer has only a noisy signal about the buyer's realized value, which generalizes the work by \citet{BBM-15} in yet another, mostly orthogonal direction.
Several other results also build on the work of \citet{BBM-15} and concern optimal segmentation in various settings, including: \citep{shen2018closed,cai2020third,mao2022interactive,bergemann2022calibrated,alijani2022limits,FJMM-24,ko2024optimal,bergemann2024unified,haghpanah2022limits,MSX-24,SY-25}.
The fairness aspect has also received considerable attention~\citep{kallus2021fairness,banerjee2024fair}.
These results are further away from ours, both conceptually and technically.

Market segmentation, when induced exogenously, can be viewed as a form of information design or Bayesian persuasion~\citep{BM-19, kamenica2011bayesian,kamenica2019bayesian}.
Most relevant to our study is the algorithmic aspect thereof~\citep{dughmi2016algorithmic}, which in particular concerns how (and whether it is possible) to efficiently compute optimal information structures in various settings.
Subsequent work has investigated variants of the problem, including settings with multiple agents~\citep{babichenko2017algorithmic,castiglioni2023public,haghtalab2025leakage}, approximate best response~\citep{feng2024rationality,lin_generalized_2025,yang2024computational}, multiple channels~\citep{babichenko2022multi}, combinatorial actions~\citep{fujii2022algorithmic}, etc.
Our results differ from the above in that we focus on a highly structured problem, where certain strong structural properties play a vital role and drastically simplify the algorithmic task.
Several technical ingredients of ours, e.g., the notion of (reduced) fragment decompositions and the use of discounted flow networks, are also novel in the context of Bayesian persuasion, to the best of our knowledge.

\section{Preliminary}
\label{sec:prelim}

\newcommand{\price}{p}
\newcommand{\optprice}{\price^*}

There is a monopoly seller selling a product to a continuum population of consumers, or buyers.
Buyers have $\valnum$ possible values $\ValSet \triangleq (\val_i)_{i\in[\valnum]}$ with each $\val_i \in \reals_+$, where $[\valnum] = \{1,2,\ldots,\valnum\}$ indexes the
values, and
the seller has $\costnum$ possible 
costs $\CostSet = (\cost_j)_{j\in[\costnum]}$ with $\cost_j \in \reals_+$ for the product. 
Without loss of generality, we assume that both costs and values are strictly increasing in their index: $0 < \val_1 < \cdots < \val_i < \cdots < \val_\valnum$ and $0 \le \cost_1 < \cdots < \cost_j < \cdots < \cost_{\costnum}$.
We consider a setting with a Bayesian seller in the sense that the seller's production cost follows a publicly known distribution, denoted by $\CostDist\in\Delta(\CostSet)$.
We use $\costpmf(\cost)\in [0, 1]$ to denote the probability for realizing a cost $\cost\in\CostSet$.

A market $\ValDist$ is a subdistribution over the $\valnum$ possible values $\ValSet$, with the set of all market being:\footnote{We use ``subdistribution'' here as we allow the total masses of the market $\ValDist$ to be less than $1$.} 
\[
    \Delta(\ValSet) = \left\{\valpmf \in\reals_+^{\ValSet} 
    \;\middle|\; 
    \sum\nolimits_{i\in[\valnum]} \valpmf(\val_i) \le 1 \right\}~.
\]
Thus, a market $\ValDist \in \Delta(\ValSet)$ corresponds to a demand function where the tail mass, denoted $\ValDist(\val) = \sum\nolimits_{\val_i \ge \val} \valpmf(\val_{i})$ is the demand for the product at the price $\val$.
\wtdelete{Define a cost distribution $\CostDist$ over $\costnum$ costs, with the set of all distribution being:
\[
    \Delta(\CostSet) = \left\{\costpmf \in\reals_+^{\CostSet} 
    \;\middle|\; 
    \sum\nolimits_{j\in[\costnum]} \costpmf(\cost_j)= 1 \right\}~.
\]}

\paragraph{Optimal price sets.}
Fix a market $\ValDist$.
Suppose the seller with the cost $\cost_j$, $j\in [\costnum]$ determines a price $p \ge \cost_j$ for the market $\ValDist$.
Then the seller's revenue induced by the price $p$ is $(p - \cost_j) \ValDist(p)$ and the buyer surplus induced by $p$ is $\sum\nolimits_{\val \ge p}(\val - p) \valpmf(\val)$.  
\wtedit{We say a price $\price$ is optimal for cost $\cost_j$ if the seller's revenue satisfies $(\price - \cost_j) \ValDist(\price) \ge  (\val - \cost_j) \ValDist(\val)$ for any $\val\in\ValSet$.}
\wtdelete{We assume $p^{*} \in \ValSet$; otherwise, we can always find a price in the value set $\ValSet$ such that the seller surplus is at least $(p^* - \cost_j) \ValDist(p^*)$.\wtcomment{I feel we don't need to mention this?}}
Under the market $\ValDist$, since the optimal price might not be unique, we further define the optimal price set $\Op_j(\ValDist)$ for each possible cost $\cost_j$:
\[
    \Op_j(\ValDist) 
    \triangleq 
    \argmax\nolimits_{\val\in \ValSet} \; (\val - \cost_j) \cdot \ValDist(\val)~,
\]
and the minimum optimal price for $\cost_j$
\[
    \op_j(\ValDist)
    \triangleq 
    \begin{cases}
        \min \Op_j(\ValDist) & \text{if}~ \max_\val \; (v - \cost_j) \cdot \ValDist(\val) > 0 ~;\\
        \infty & \text{otherwise}~.
    \end{cases}
\]
Let $\op(\ValDist) = (\op_1(\ValDist), \dots, \op_\costnum(\ValDist))$ denote the vector of minimum optimal prices for all seller types in $\ValDist$.\footnote{When $\max_{\val}(\val-\cost_j)\ValDist(\val)=0$, every optimal price yields zero seller surplus and buyer surplus. Then we set $\op_j(\ValDist)=\infty$.}
We here collect two properties of optimal price set that will be useful for our subsequent analysis.
\begin{lemma}[Monotonicity of Optimal Price Sets]
\label{lem:monotonicity of optimal price set -v2}
Fix a market $\ValDist$ and a cost $\cost_j$. When $\max_\val (\val-\cost_j)\cdot \ValDist(\val) > 0$, i.e., the maximum seller surplus for $\cost_j$ is strictly positive, then $\max \Op_j(\ValDist) \le \min \Op_{j+1}(\ValDist)$.  
\end{lemma}

\begin{lemma}\label{lem: optimal price set - scaled F}
For any $\alpha \in [0, 1]$,
define $\alpha \ValDist$ as  $(\alpha \ValDist)(\val) = \alpha \ValDist(\val)$ for all $\val \in  \ValSet$. 
For any market $\ValDist \in \Delta(\ValSet)$ and its optimal price set $\Op_j(\ValDist)$, it holds that $\Op_j(\ValDist) = \Op_j(\alpha \ValDist)$.    
\end{lemma}
 

\paragraph{Seller surplus, buyer surplus, and social welfare.}
For a specific cost $\cost_j$ and 
its optimal price set $\Op_j(\ValDist)$, the seller is indifferent between all prices in $\Op_j(\ValDist)$. Then, if an arbitrary optimal price $p_j \in \Op_j(\ValDist)$ is used for each $\cost_j$, the corresponding (optimal) overall seller surplus is
\[
    \SellerSurplus(\ValDist) =  \sum\nolimits_{j\in[\costnum]} \costpmf(\cost_j) \cdot (p_j - \cost_j) \cdot \ValDist(p_j)~.
\]
As for the buyer surplus, different optimal prices for each possible cost may lead to different buyer surplus. 
Therefore, the maximum (resp.\ minimum) buyer surplus, denoted by $\MaxBuyerSurplus(\ValDist)$ (resp.\ $\MinBuyerSurplus(\ValDist)$), is derived when we choose the minimum optimal price $\op_j$ (resp.\ maximum optimal price $\opm_j \triangleq \max \Op_j(\ValDist)$) for each cost $\cost_j$:
\begin{align*}
    \text{Maximum buyer surplus:} \quad 
    \MaxBuyerSurplus(\ValDist) 
    & = \sum\nolimits_{j\in[\costnum]} \costpmf(\cost_j) \sum\nolimits_{\val \ge \op_j} (\val - \op_j) \cdot \valpmf(\val) \\
    \text{Minimum buyer surplus:} \quad 
    \MinBuyerSurplus(\ValDist) 
    & = \sum\nolimits_{j\in[\costnum]} \costpmf(\cost_j) \sum\nolimits_{\val \ge \opm_j} (\val - \opm_j) \cdot \valpmf(\val)~.
\end{align*}
Since the social welfare is the sum of the seller surplus and the buyer surplus, we can also obtain the maximum/minimum social welfare of $\ValDist$ as follows:
\begin{align*}
    \text{Maximum social welfare:} \quad 
    \MaxSocialWelfare(\ValDist) 
    & = \SellerSurplus(\ValDist) + \MaxBuyerSurplus(\ValDist)\\ 
    \text{Minimum social welfare:} \quad 
    \MinSocialWelfare(\ValDist) 
    & = \SellerSurplus(\ValDist) + \MinBuyerSurplus(\ValDist)~.
\end{align*}
\paragraph{Aggregate market and market segmentation.}

Given a market $\ValDist$, a market segmentation, denoted by $(\segWeight, \MarketSegment)_{\segi \in [\segI]}$, is a way of expressing this market as a convex combination of different market segments, where $[\segI]$ indexes the
segments, $(\segWeight)_{\segi\in[\segI]}\in\reals_+^{\segI}$ are segment weights with $\sum\nolimits_{\segi \in [\segI]} \segWeight = 1$, and each segment $\MarketSegment \in\Delta(\ValSet)$ is a market over $\ValSet$. The probability mass functions together satisfy
\begin{align*}
    \valpmf(\val) = \sum\nolimits_{\segi \in [\segI]} \segWeight \cdot \marketSegmentPmf (\val)~, \quad 
    \val\in\ValSet~.
\end{align*}
We write $\ValDist \to (\segWeight, \MarketSegment)_{\segi \in [\segI]}$ to denote such market segmentation.

Throughout the analysis, we hold a given aggregate market as fixed and identify it by $\ValDist^* \in \Delta(\ValSet)$ with the total mass normalized, i.e., $\sum\nolimits_{i\in [\valnum]} \valpmf^{*}(\val_i) = 1$.

Given a segmentation $(\segWeight, \MarketSegment)_{\segi \in [\segI]}$, any market-level functional extends to segmentations via the weighted average. 
For example, the seller's revenue under this segmentation $(\segWeight, \MarketSegment)_{\segi \in [\segI]}$ is 
\begin{align*}
    \SellerSurplus \left((\segWeight, \MarketSegment)_{\segi \in [\segI]} \right) 
    =  \sum\nolimits_{\segi \in [\segI]} \segWeight \cdot \SellerSurplus ( \MarketSegment)~.
\end{align*}
The definitions of maximum/minimum buyer surplus and the corresponding welfare objectives follow similarly.

\section{Full Generality of Extremal Market}
\label{sec:extremal}
In this section, we present our first structural result: the full generality of ``extremal markets'', which lays the foundation of all subsequent results in the paper. The extremal market here is a generalized version of the extremal market defined in \cite{BBM-15}, where the distribution of the cost consists of a single point mass. We formally define the extremal market as follows.

\begin{definition}[Extremal Markets]
\label{def: extremal markets}
    Suppose $\op$ and $\SuppSet$ satisfy: \yledit{(1) $\SuppSet = \{s_1, s_2, \dots, s_k\}$, where $s_1 < s_2 < ... < s_k$,} (2) $\{\op_j\}_j \subseteq \SuppSet \cup \{\infty\} \subseteq \{v \in \ValSet \mid v \ge \op_1\} \cup \{\infty\}$, and (3) for each $j \in [m]$, $\cost_j < \op_j \le \op_{j + 1}$.\footnote{\label{ft:q_m+1}We assume $\op_{m + 1} = \infty$ for brevity.}
    The extremal market $\ValDist_{\op, \SuppSet}$ induced by $\op$ and $S$ is the unique market satisfying the following conditions:\footnote{Uppercase letters generally refer to CDFs, defined as the {\em tail} probability at every $\val_i$, i.e., $\ValDist_{\op, \SuppSet}(\val_i) = \Pr_{v \sim \ValDist_{\op, \SuppSet}}[v \ge \val_i]$.}
    \begin{itemize}
        \item For each $i \in [k]$ and $j \in [m]$ where $\op_j \le s_i < s_{i+1} \le \op_{j + 1}$, $\ValDist_{\op, \SuppSet}(s_i) \cdot (s_i - \cost_j) = \ValDist_{\op, \SuppSet}(s_{i + 1}) \cdot (s_{i + 1} - \cost_j)$;
        \item For each $i \in [n]$ where $\val_i \notin S$, $\ValDist_{\op, \SuppSet}(\val_i) = \ValDist_{\op, \SuppSet}(v_{i + 1})$.\footnote{We assume $\ValDist_{\op, \SuppSet}(v_{n + 1}) = 0$ for brevity.}
    \end{itemize}
\end{definition}
Given the definition of extremal markets, the natural question is how we construct it given fixed a pair of $\op$ and $\SuppSet$. 
We observe that $\ValDist_{\op, \SuppSet}$ is the unique market such that, for each $j \in [m]$, all prices in the support $S$ between $\op_j$ and $\op_{j + 1}$ are equally good for cost $\cost_j$.
Then the extremal market can be constructed naturally through the following procedure:
\begin{itemize}
    \item First, tentatively put probability $1$ at the rightmost position $\supp_k$ of the market.
    \item For $i = k - 1$ to $1$, determine the probability at $\supp_i$ by letting $\supp_i$ and $\supp_{i + 1}$ give the same seller surplus for cost $\cost_j$, where $j$ satisfies $\op_j \le \supp_i < \op_{j + 1}$.
    \item Normalize the entire market so the total probability is $1$.
\end{itemize}
The definition and construction procedure of the extremal market $\ValDist_{\op, \SuppSet}$ implies that the optimal price set $\Op_j(\ValDist_{\op, \SuppSet})$ and the minimum optimal price $\op_j(\ValDist_{\op, \SuppSet})$ with $j \in [m]$ are formed as follows.

\begin{lemma}[Optimal Prices in Extremal Markets]
\label{lem: optimal price set of extremal market} 
\label{lem:optimal_price_for_extremal_market}
    Fix $\op$, $\SuppSet$ and the corresponding extremal market $\ValDist_{\op, \SuppSet}$. For each $j \in [m]$, (1) the minimum optimal price $\op_j(\ValDist_{\op, \SuppSet}) = \op_j$, and (2) when $\op_j \ne \infty$ (i.e., $\max_\val (\val-\cost_j)\cdot \ValDist_{\op, \SuppSet}(\val) \ne 0$), $\Op_j(\ValDist_{\op, \SuppSet}) = [\op_j, \op_{j+1}] \cap \SuppSet$. 
\end{lemma}

After defining and characterizing our general version of extremal market, we are ready to present our main result in this section, which connects the extremal markets with general markets: Any market $\ValDist$ supported on $\ValSet$ can be written as a convex combination of $O(n)$ extremal markets with both the seller surplus and the buyer surplus preserved,\footnote{In this paper, the buyer surplus of a market preserved by a market segmentation indicates that the buyer surplus of the market $\ValDist$ is between the minimum buyer surplus and the maximum buyer surplus of the segmentation. In particular, the buyer surplus in $\ValDist$ can be derived by a convex combination of different tiebreaking choices in the segmentation.} where $n$ is the cardinality of value set $\ValSet$. 

\begin{theorem}\label{thm:convex_combination_of_extremal_market}
    Fix a market $\ValDist$ with seller surplus $\SellerSurplus$ and buyer surplus $\BuyerSurplus$. Then there exists a segmentation of $\ValDist$ into $O(n)$ extremal markets, denoted $\ValDist = \sum_{\op,\SuppSet} \alpha_{\op,\SuppSet}\,\ValDist_{\op,\SuppSet}$, such that the seller surplus 
    \wtedit{from this market segmentation}
    is $\SellerSurplus$, the minimum buyer surplus is at most $\BuyerSurplus$ and the maximum buyer surplus is at least $\BuyerSurplus$. In addition, there is a polynomial-time algorithm that computes such a segmentation of $\ValDist$.
\end{theorem}
We design a decomposition procedure (see \Cref{alg:extremal market decomposition}) that can preserve the seller surplus and buyer surplus given a market $\ValDist$ to show the existence of such a segmentation. The key idea is that, through the procedure, we need to guarantee that the monopoly price of each cost $\cost_j$ in the target market $\ValDist$ is still optimal in the ``remainder'' of the market after each iteration to preserve the seller surplus. By ensuring that the optimal price set weakly expands after each iteration, we can further preserve the buyer surplus. 
\begin{lemma}
\label{lem: extremal market decomposition - optimal price set expansion}
    After each iteration of the while-loop in \Cref{alg:extremal market decomposition}, the residual market satisfies that for each $j \in [m]$
    \[
        \Op_j(\ValDist) \subseteq \Op_j(\alpha_{\op, \SuppSet} \cdot \ValDist_{\op, \SuppSet}) \cap \Op_j(\ValDist - \alpha_{\op, \SuppSet} \cdot \ValDist_{\op, \SuppSet})~.
    \]
\end{lemma}
At the same time, the procedure has at most $2n$ iterations to make sure the decomposition procedure terminates properly and limit the number of extremal markets that we get.
\begin{lemma}
\label{lem: extremal market decomposition - decomposition time}
    In \Cref{alg:extremal market decomposition}, for any market $\ValDist$, it takes at most $2n$ iterations to decompose it into extremal markets. 
\end{lemma}
\begin{algorithm}[t]
\caption{\textsc{DecomposeIntoExtremalMarkets}$(\ValDist)$}
\begin{algorithmic}[1]
\label{alg:extremal market decomposition}
\REQUIRE (Possibly unnormalized) market $\ValDist$ with its mass function $\valpmf(\cdot)$
\ENSURE A decomposition
$\ExtMarketDecomp=\{(\alpha_{\op^{(k)},\SuppSet^{(k)}},\ValDist_{\op^{(k)},\SuppSet^{(k)}})\}_{k=1}^K$,
where $(\op^{(k)},\SuppSet^{(k)})$ are computed from the current $\ValDist$ in iteration $k$.
\STATE $\ExtMarketDecomp\gets \emptyset$; $k \gets 0$
\WHILE{$\ValDist \neq 0$}
    \STATE $k \gets k+1$
    \Comment{Within iteration, we omit the subscript $(k)$ on $\SuppSet$, and $\op$ for readability.  }
    \STATE $\op \gets \op(\ValDist)$; $\SuppSet \gets$ the support of $\ValDist$; \yledit{ $\SuppSet \gets \{s \in \SuppSet: s \ge \op_1\}$.}
    \STATE Consider the extremal market $\ValDist_{\op,\SuppSet}$. 
    \hfill
    \STATE $\alpha_{\text{runout}} \gets
        \min_{\val \in \ValSet:\;\valpmf_{\op,\SuppSet}(\val)>0}
        \frac{f(\val)}{\valpmf_{\op,\SuppSet}(\val)}$.
        \hfill
        \Comment{Probability runs out at some $\val$: $\alpha_{\text{runout}} \cdot \valpmf_{\op, \SuppSet}(\val) = \valpmf(\val)$}.
    \STATE $\alpha_{\text{shift}} \gets$ the minimum $\alpha$ such that the minimum optimal price shifts for some cost $\cost_j$: $\op(\ValDist - \alpha \cdot \ValDist_{\op, \SuppSet})$ becomes different from $\op(\ValDist)$.
    \STATE $\alpha_{\op,\SuppSet} \gets \min\{\alpha_{\text{runout}}, \alpha_{\text{shift}}\}$
    \STATE Append $(\alpha_{\op,\SuppSet},\ValDist_{\op,\SuppSet})$ to $\ExtMarketDecomp$.

    \STATE $\ValDist \gets \ValDist - \alpha_{\op,\SuppSet} \cdot \ValDist_{\op,\SuppSet}$.
\ENDWHILE
\RETURN $\ExtMarketDecomp$

\end{algorithmic}
\end{algorithm}

To argue that the entire procedure preserves both the seller surplus and the buyer surplus, we only need to argue that in each iteration of the while-loop, for every way to break ties between different optimal prices in $\ValDist$, there is a way to break ties in $\alpha_{\op, \SuppSet} \cdot \ValDist_{\op, \SuppSet}$ and $F - \alpha_{\op, \SuppSet} \cdot \ValDist_{\op, \SuppSet}$, such that the total seller surplus (resp.\ buyer surplus) in the latter two markets is the same as that in $\ValDist$.
This is a direct corollary of the key properties mentioned above.

Since every general market can be written as a convex combination of extremal markets, for any achievable welfare outcome (specified by the buyer surplus and the seller surplus) through a market segmentation of the aggregate market $\AggMarket$, we can find a segmentation into extremal markets to complete the same task --- hence the full generality of extremal markets.
\begin{corollary}[Full Generality of Extremal Markets] 
\label{corollary: full generality of extremal market}
\label{lem: full generality of extremal markets} 
Given the aggregate market $\AggMarket$, if there exists a segmentation $(\segWeight, \MarketSegment)_{\segi \in [\segI]} $ which induces the overall seller surplus $\SellerSurplus$ and the overall buyer surplus $\BuyerSurplus$, then we can find segmentation into extremal markets $(\alpha_{\op, \SuppSet}, \ValDist_{\op, \SuppSet})$ such that the overall seller surplus is $\SellerSurplus$, the minimum buyer surplus is at most $\BuyerSurplus$ and the maximum buyer surplus is at least $\BuyerSurplus$. Moreover, given any segmentation, one can compute in polynomial time
a segmentation into extremal markets satisfying the above properties.
\end{corollary}
\paragraph{Extremal market decompositions.}
Given the full generality of extremal markets, when characterizing or optimizing achievable outcomes, we can focus on decompositions of the aggregate market into extremal markets.
We say a collection of weights $\ba = (\alpha_{\op, \SuppSet})$ is an {\em extremal market decomposition} of \yledit{the aggregate market $\AggMarket$ (denoted by $\AggMarket \to \ba$), iff $\sum_{\op, \SuppSet} \alpha_{\op, \SuppSet} \cdot \ValDist_{\op, \SuppSet} = \AggMarket$}.
We immediately have the following algorithm for optimizing achievable outcomes: Optimize over all extremal market decompositions of the aggregate market, which can be done through linear programming, since the overall seller / buyer surplus is linear in the weight $\alpha_{\op, \SuppSet}$ of each extremal market in the decomposition.
The issue with this approach is that there are {\em exponentially many} possible extremal markets, which means the LP, formulated in the straightforward way, is of exponential size.
Obtaining a polynomial-time algorithm requires much more effort and new technical ideas, which we develop step by step in the subsequent sections.

\section{Fragment Decompositions}
\label{sec:fragment}
In the previous section, we have transformed the problem from segmentation of markets in any form into segmentation based on extremal markets. While this reduces the size of the problem space significantly, there are still exponentially many extremal markets. In this section, we discuss the structural properties of extremal markets, which eventually enable our polynomial-time algorithm. We first introduce the notation of \emph{indifference fragments} and the (unique) fragment decomposition of an extremal market decomposition in \Cref{subsec:fragments}. We show that both the seller surplus and the buyer surplus induced by an extremal market decomposition are linear in the weights in its fragment decomposition. Compared with the number of extremal markets, the size of fragments is further reduced but still exponential. 
To make it more efficient, in \Cref{sec: undominated fragments and fragment decompositions}, we further introduce the undominated fragments and reduced fragment decompositions, which is unique to each extremal market decomposition. We show that the size of the reduced fragment decomposition is $O(mn^2)$, where $n =|\ValSet|$ and $m=|\CostSet|$, and the welfare outcome of the extremal market decomposition and its reduced fragment decomposition remain the linear relationship.

\subsection{Indifference Fragments and Fragment Decompositions}
\label{subsec:fragments}
We first introduce the notion of indifference fragments.
\wtedit{
Intuitively, an extremal market $\ValDist_{\op,\SuppSet}$ is constructed by stitching together ``equal-surplus'' pieces: On the $j$-th piece (values in $[\op_j, \op_{j+1})$), let $\FragSuppSet :=  [\op_j,\op_{j+1}) \cap \SuppSet $ denote the support prices in this piece.
Then, the tail $\ValDist(\cdot)$ is chosen so that every support price $\price\in \FragSuppSet$ gives the same seller surplus $(\price - \cost_j) \ValDist(\price)$, and keeps
$\ValDist_{\op,\SuppSet}(\cdot)$ flat between consecutive support points in $\FragSuppSet$.
\Cref{defn:fragment} then extracts the $j$-th piece  and renormalizes it by $\esw_{\op,\SuppSet,j} := \ValDist_{\op,\SuppSet}(\op_j)$, i.e., the ``incoming'' tail mass when the extremal market first ``reaches'' the piece. 
Equivalently, fragments are standardized local building blocks, and an extremal market is obtained by stacking these blocks with weights $\esw_{\op,\SuppSet,j}$.}

\begin{definition}[Indifference Fragments]
\label{defn:fragment}
An \emph{indifference fragment}, denoted $\Frag_{\op,\SuppSet,j}$, is a market attained by restricting an extremal market $\ValDist_{\op, \SuppSet}$ to the interval $[\op_j,\op_{j+1})$ for some $j \in [m]$ scaled down by a factor $\esw_{\op,\SuppSet,j}$, where $\esw_{\op,\SuppSet,j} := \ValDist_{\op,\SuppSet}(\op_{j}) \in[0,1]$. 
That is, $\Frag_{\op,\SuppSet,j}$ is the market whose mass function satisfies 
\[
\frag_{\op,\SuppSet,j}(\val)
:=
\begin{cases}
\displaystyle
\frac{\valpmf_{\op,\SuppSet}(\val) }{\esw_{\op,\SuppSet,j}}, & \op_j \le \val < \op_{j+1}~,\\
0, & \text{otherwise}~.
\end{cases}
\]
If $\esw_{\op,\SuppSet,j}=0$, we set $\frag_{\op,\SuppSet,j}\equiv 0$
(e.g., the empty market).
\end{definition}
For readability, we subsequently write \emph{fragments} instead of \emph{indifference fragments}.
As a sanity check, for $\sw_{\op,\SuppSet,j}>0$, the total mass of the fragment $\Frag_{\op,\SuppSet,j}$ is equal to
$\sum\nolimits_{\val} \frag_{\op,\SuppSet,j}(\val) = 1-\ValDist_{\op,\SuppSet}(\op_{j+1})/\ValDist_{\op,\SuppSet}(\op_j)$.
Since $\op_j$ and $\op_{j+1}$ are optimal prices for a seller of type $\cost_j$ in the extremal market $\ValDist_{\op,\SuppSet}$ (\Cref{lem: optimal price set of extremal market}), we also have
$\ValDist_{\op,\SuppSet}(\op_{j+1}) \cdot (\op_{j+1}-\cost_j) = \ValDist_{\op,\SuppSet}(\op_j)\cdot(\op_j-\cost_j)$.
\wtedit{With this observation, the next lemma formalizes the following key property of the fragment: After scaling by $\esw_{\op,\SuppSet,j} =\ValDist_{\op,\SuppSet}(\op_{j})$, the $j$-th \emph{fragment} is completely pinned down by the two endpoints $\op_j,\op_{j+1}$ together with the support set $\FragSuppSet=  [\op_j,\op_{j+1}) \cap \SuppSet$.}
\begin{lemma}
\label{lem:locality of fragment}
Fix a fragment $\Frag_{\op, \SuppSet, j}$ attained by an extremal market $\ValDist_{\op, \SuppSet}$. 
Let $\FragSuppSet :=  [\op_j,\op_{j+1}) \cap \SuppSet $. Then the fragment $\Frag_{\op,\SuppSet,j}$ depends on $(\op,\SuppSet)$ only through $(\op_j, \op_{j+1}, \FragSuppSet)$. Specifically, for every $\val \in \FragSuppSet$,
\[
       \frag_{\op, \SuppSet, j}(\val)  = \frac{\op_j - \cost_j}{\val - \cost_j} - \frac{\op_j - \cost_j}{\val^+ - \cost_j}~, \quad
       \Frag_{\op, \SuppSet, j}(\val)  = \frac{\op_j - \cost_j}{\val - \cost_j} - \frac{\op_j - \cost_j}{\op_{j+1} - \cost_j}~,
\]
where $\val^+ = \min \{\val' \in \FragSuppSet \cup \{\op_{j+1}\}: \val' > \val \}$.
\end{lemma}

\wtedit{\Cref{lem:locality of fragment} implies that for any two extremal markets $\ValDist_{\op,\SuppSet}$ and $\ValDist_{\op',\SuppSet'}$, for any index $j$,
if we have $\op'_j=\op_j$, $\op'_{j+1}=\op_{j+1}$ and $\FragSuppSet'_j = \FragSuppSet_j$, then one must have $\Frag_{\op,\SuppSet,j} \;=\; \Frag_{\op',\SuppSet',j}$.} In light of the above, we give a succinct parametrization of fragments as follows.

\begin{remark}[Succinct Parametrization of Fragments]
\label{remark: Succinct fragments}
For $j\in [m]$, $\ell\in [n]$, $r\in[n+1]$, and any set
$ \FragSuppSet \subseteq \ValSet \cap [\val_\ell,\val_r)$ with $\val_\ell \in \FragSuppSet$, we write $\Frag_{j, \ell, r, \FragSuppSet}$ (resp.\ $\frag_{j, \ell, r, \FragSuppSet}$) to denote the fragment $\Frag_{\op,\SuppSet,j}$ (resp.\ its mass function $\frag_{\op,\SuppSet,j}$ ) satisfying
\[
    \op_j=\val_\ell~,\quad \op_{j+1}=\val_r~,\quad \SuppSet\cap [\val_\ell,\val_r)=\FragSuppSet~,
\]
whenever such a market exists.\footnote{This implies that $\cost_j < \val_\ell \le \val_r$ since $\cost_j < \op_{j} \le \op_{j+1}$.}
\end{remark}

From now on, we generally write $\Frag_{j, \ell, r, \FragSuppSet}$ as the notation of fragments. Note that although the succinct parametrization implies that the number of different fragments is much smaller than the number of extremal markets, there can still be {\em exponentially many different fragments} because of the dependency on the support $\FragSuppSet$.
We will handle this problem in \Cref{sec: undominated fragments and fragment decompositions}.
One main reason why we introduce fragments is that they capture the essence (e.g., the seller / buyer surplus) of an extremal market, or more generally, of an extremal market decomposition. We express any extremal market as a weighted sum of fragments in a concrete way (defined in the fragment decompositions of an extremal market) and further sum up the fragment decompositions of every extremal market, scaled by its weight in extremal market decomposition, to obtain the fragment decompositions of an extremal market decomposition. The formal definition is as follows.
\begin{definition}[Fragment Decompositions]\label{def:fragment-decomposition}
Fix the family of fragments $\{\Frag_{j, \ell, r, \FragSuppSet}\}$ (as in \Cref{remark: Succinct fragments}), where $j\in [m]$, $\ell \in [n]$, $r\in [n+1]$ and $ \FragSuppSet \subseteq \ValSet \cap [\val_\ell, \val_r) $.

\smallskip
\noindent\textbf{(i) Fragment decomposition of an extremal market.}
Given an extremal market $\ValDist_{\op,\SuppSet}$, a collection of nonnegative weights
$\bw=(\esw_{j, \ell, r, \FragSuppSet}^{\op, \SuppSet})$ is called a fragment decomposition of $\ValDist_{\op,\SuppSet}$
(denoted $\ValDist_{\op,\SuppSet}\to \bw$) iff
\begin{enumerate}
    \item for any quadruple $(j, \ell, r, \FragSuppSet)$ with $j\in [m]$, $\ell \in  [n]$, $r \in [n+1]$ and $ \FragSuppSet \subseteq \ValSet \cap [\val_\ell, \val_r)$,
    \[
       \esw_{j, \ell, r, \FragSuppSet}^{\op, \SuppSet} = \begin{cases}
           \ValDist_{\op, \SuppSet}(\op_j), & \val_\ell = \op_j,\val_r = \op_{j+1}, \FragSuppSet = \SuppSet \cap [\op_j, \op_{j+1})~;\\
           0, & \text{otherwise}~.
       \end{cases}
    \]
     Note that when $\ell = r$, which implies that $\op_j = \op_{j+1}$, we have $\esw_{j, \ell, r, \FragSuppSet}^{\op, \SuppSet}= \ValDist_{\op_j} = \ValDist_{\op_{j+1}} = \esw_{j+1, r, k, \FragSuppSet'}^{\op, \SuppSet}$,  where $\val_k = \op_{j+2}$ and $\FragSuppSet' = \SuppSet \cap [\op_{j+1}, \op_{j+2})$, while $\FragSuppSet = \emptyset$.
    \item $\ValDist_{\op,\SuppSet}=\sum_{j, \ell, r, \FragSuppSet} \esw_{j, \ell, r, \FragSuppSet}^{\op, \SuppSet} \cdot \Frag_{j, \ell, r, \FragSuppSet}~$.
\end{enumerate} 

\smallskip

\noindent\textbf{(ii) Fragment decomposition of an extremal market decomposition.}
Let $\ba=(\alpha_{\op,\SuppSet})$ be an extremal market decomposition.
We say $\bw$ is a fragment decomposition of $\ba$ (denoted $\ba\to \bw$) iff
\begin{align*}
    \sum\nolimits_{\op,\SuppSet} \alpha_{\op,\SuppSet}\cdot \ValDist_{\op,\SuppSet}
    = \sum\nolimits_{j, \ell, r, \FragSuppSet} \sw_{j, \ell, r, \FragSuppSet} \cdot \Frag_{j, \ell, r, \FragSuppSet}~,
\end{align*}
where $\sw_{j, \ell, r, \FragSuppSet} = \sum_{\op, \SuppSet} \alpha_{\op, \SuppSet} \cdot \esw_{j, \ell, r, \FragSuppSet}^{\op, \SuppSet}$, and $\esw_{j, \ell, r, \FragSuppSet}^{\op, \SuppSet}$ is the weight of fragment $\Frag_{j, \ell, r, \FragSuppSet}$ in the fragment decompositions of the extremal market $\ValDist_{\op, \SuppSet}$.

\smallskip
\noindent\textbf{(iii) Fragment decomposition of an aggregate market.}
Given the aggregate market $\AggMarket$, we say $\bw$ is a fragment decomposition of $\AggMarket$
(denoted $\AggMarket \to \bw$) iff there exists an extremal market decomposition $\ba$ of the aggregate market $\AggMarket$
(i.e., $\AggMarket \to \ba$) such that $\ba\to \bw$. 
\end{definition}
\paragraph{Towards a fragment-based linear program.}
To see why fragment decompositions are potentially helpful, let us observe that the seller / buyer surplus induced by an extremal market decomposition $\ba$ can be determined from the fragment decomposition $\bw$ of $\ba$ --- in fact, both quantities are linear in $\bw$.

Consider the seller surplus first.
If we charge the seller surplus to the specific fragment $\Frag_{j, \ell, r, \FragSuppSet}$ with a non-zero weight $\sw_{j, \ell, r, \FragSuppSet}$ in a particular way: $\sw_{j, \ell, r, \FragSuppSet} \cdot \costpmf(\cost_j) \cdot (\val_\ell - \cost_j)$, where $\costpmf(\cost_j)$ is the probability of cost $\cost_j$, we can get the same seller surplus as the extremal market decomposition.
This is because in any extremal market where the weight of fragment $\Frag_{j, \ell, r, \FragSuppSet}$ is not zero in its fragment decompositions (note that the fragment does not necessarily come from a single extremal market), it must be the case that $\val_\ell$ is one of the optimal prices for cost $\cost_j$, and seller surplus for this specific $\cost_j$ equals $\ValDist(\val_\ell) \cdot \costpmf(\cost_j) \cdot (\val_\ell - \cost_j)$, where $\ValDist(\val_\ell)$ is the tail probability of $\val_\ell$. As defined in \Cref{def:fragment-decomposition}, $\sw_{j, \ell, r, \FragSuppSet}$ is the weighted sum of $\ValDist(\val_\ell)$ in the extremal markets where the weight of the fragment $\Frag_{j, \ell, r, \FragSuppSet}$ is not 0. 
Given that, the way we charge the seller surplus of $\Frag_{j, \ell, r, \FragSuppSet}$ is equal to the seller surplus contribution of $\cost_j$ in those extremal markets. 

Similarly, the contribution to the social welfare of each fragment $\Frag_{j, \ell, r, \FragSuppSet}$ is between $\left(\sum_{j' < j} \costpmf(c_{j'})\cdot (\val-\cost_{j'})\right) \cdot \left(\sum_{v \in T} \frag_{j, \ell, r, \FragSuppSet}(\val)\cdot \sw_{j, \ell, r, \FragSuppSet}\right)$ and $\left(\sum_{j' \le j} \costpmf(c_{j'})\cdot (\val-\cost_{j'})\right) \cdot \left(\sum_{v \in T}  \frag_{j, \ell, r, \FragSuppSet}(\val) \cdot \sw_{j, \ell, r, \FragSuppSet}\right)$, the former corresponding to the case where the seller, when the type is $\cost_j$, favors $\val_r$ to $\val_\ell$ as the price, and the latter $\val_\ell$ to $\val_r$. The buyer surplus then can be obtained by taking the difference of the social welfare and the seller surplus. Therefore, we can get the conclusion below.

\begin{lemma}
\label{lem: fragment decompositions}
    The  seller surplus and (maximum / minimum) buyer surplus induced by an extremal market decomposition $\ba$ are linear in its fragment decompositions $\bw$ (i.e., $\ba \to \bw$).
\end{lemma}

The above observations hint at the possibility of optimizing achievable outcomes through an LP of reduced size, based directly on fragments, rather than extremal markets.
The high-level idea is to have weights in a fragment decomposition $\bw$ as decision variables, maximize some combination of the seller surplus and the buyer surplus (both of which can be determined directly from the fragment decompositions), and enforce the constraint that $\AggMarket \to \bw$, where $\AggMarket$ is the aggregate market.
However, there are two outstanding issues with this approach:
\begin{itemize}
    \item The size of the LP, formulated in a straightforward way, is still exponential, since we need at least one decision variable for each possible fragment, and the number of possible fragments is exponential in general.
    \item It is not immediately clear how to enforce the constraint that $\AggMarket \to \bw$ in a linear way (without explicit modeling an intermediate extremal market decomposition, which would blow up the size of the LP).
    One tempting idea is to simply require that $\sum_{j, \ell, r, \FragSuppSet} \sw_{j, \ell, r, \FragSuppSet} \cdot \Frag_{j, \ell, r, \FragSuppSet} = \AggMarket$.
    However, it is not too hard to construct examples where doing so relaxes the original constraint by too much, and optimal solutions to the LP no longer make sense.
\end{itemize}
Next, we deal with the first issue by introducing undominated fragments. The second issue is discussed in \Cref{sec:flow}.

\subsection{Undominated Fragments and Reduced Fragment Decompositions}
\label{sec: undominated fragments and fragment decompositions}
In this subsection, we focus on the issue of exponentially many fragments.
Roughly speaking, our solution is to group fragments that share the same leftmost and rightmost points together, and represent each group using an ``undominated'' fragment.

\begin{figure}[h!]
    \centering
    \includegraphics[width=\linewidth]{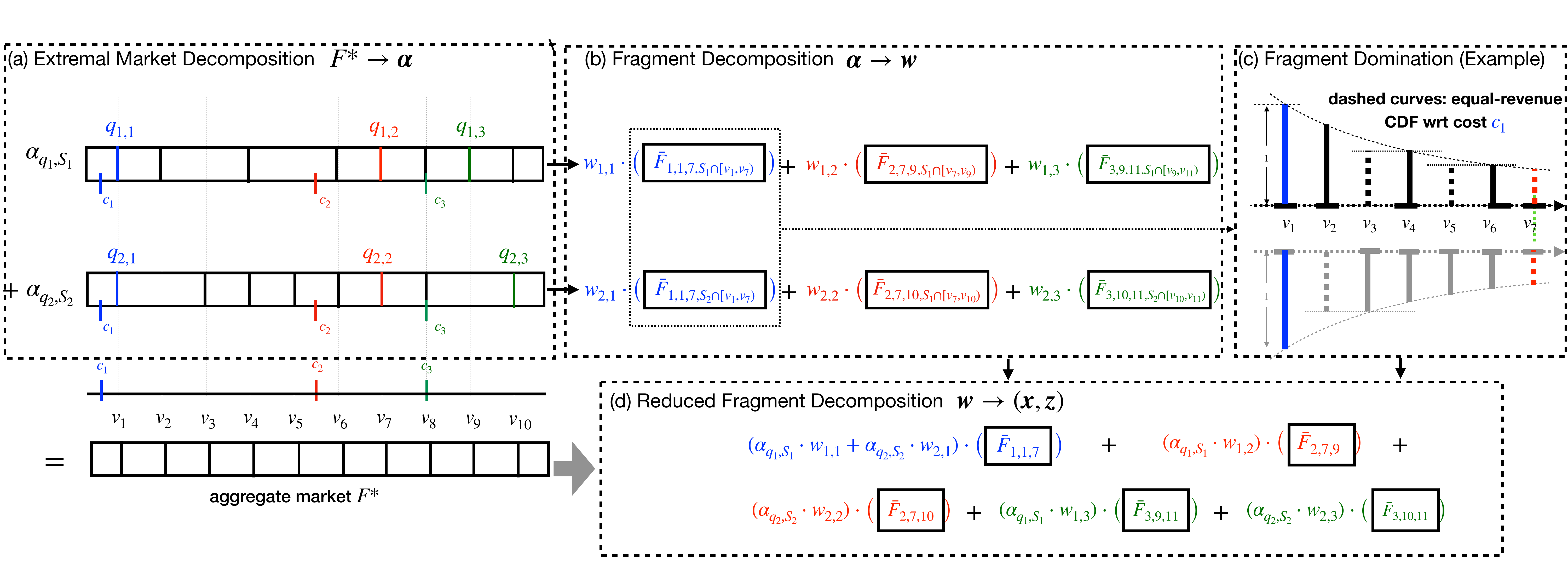}
    \caption{Decomposition Procedure (Example): Fix an aggregate market $\AggMarket$ with the value set $(\val_i)_{i\in [10]}$ ($\val_{11} = \infty$) and a cost set $(\cost_j)_{j \in [3]}$.
    Part (a) shows that the aggregate market is decomposed into $2$ extremal markets $\{\ValDist_{\op_k, \SuppSet_k}\}_{k=1,2}$ weighted by $\alpha_{\op_k, \SuppSet_k}$ respectively. For each extremal market $\ValDist_{\op_k, \SuppSet_k}$, the solid line indicates the corresponding value is in their support set and $\op_{k,j}$ stands for the minimum optimal price of the cost $\cost_j$. 
    Part (b) shows the fragment decomposition of each extremal market, where $\Frag_{j, \ell, r, \FragSuppSet}$ is the fragment with non-zero weight $\esw_{i,j}$. Here $\val_\ell= \op_{i,j}$, $\val_r = \op_{i, j+1}$ and $\FragSuppSet = \SuppSet_i \cap [\val_\ell , \val_r)$. 
    Part (c) shows that, specifically, both $\Frag_{1,1,7, \SuppSet_1 \cap[\val_1, \val_7)}$ and $\Frag_{1,1,7, \SuppSet_2 \cap[\val_1, \val_7)}$ are stochastically dominated by the undominated fragment $\Frag_{1,1,7}$,
    finally, Part (d) demonstrates the reduced fragment decomposition, where their weights are grouped together.}
    \label{fig:fragment_decomposition}
\end{figure}

\begin{definition} [Undominated Fragments]
\label{def: undominated fragments}
We define a fragment $\Frag_{j, \ell, r, \FragSuppSet}$ to be {\em an undominated fragment} iff $\FragSuppSet = \ValSet \cap [\val_\ell, \val_r)$. Thus, given the value set $\ValSet$ and the indexes $\ell$ and $r$, the support set $\FragSuppSet$ is uniquely determined. We simply denote the undominated fragment by $\Frag_{j, \ell, r}$. 
\end{definition}
This immediately implies the following upper bound on the number of possible undominated fragments and the first-order stochastic dominance between fragments and undominated fragments that share the same index $j$, $\ell$ and $r$.

\begin{lemma}
\label{lem: number of undominated fragments}
There are only $O(mn^2)$ possible undominated fragments.
\end{lemma}
\begin{lemma}
\label{lem: domination relationship}
   \yledit{Fix $j\in[m]$, $\ell \in[n]$ and $r\in[n+1]$} such that $\cost_j < \val_\ell \le \val_r$. Let $\Frag_{j, \ell, r, \FragSuppSet}$ be the fragment with the support set $\FragSuppSet \subseteq \ValSet \cap [\val_\ell, \val_r)$. Then $\Frag_{j, \ell, r, \FragSuppSet}$ is first-order stochastically dominated by the undominated fragment $\Frag_{j, \ell, r}$ (in tail order), i.e., for any $\val$, $\Frag_{j, \ell, r, \FragSuppSet}(\val) \le \Frag_{j, \ell, r}(\val)$.
\end{lemma}

An example of two fragments dominated by a common undominated fragment is shown in the right part of Figure~\ref{fig:fragment_decomposition}: The horizontal bars denote the supports of the two fragments.
As the figure shows, both fragments are dominated by the tail mass corresponding to the undominated fragment with the same $(j, \ell, r)$.
Below we argue that in a fragment decompositions, such fragments can be grouped together with some additional bookkeeping, such that the ``reduced'' fragment decompositions still preserves the essence of an extremal market decomposition.

\begin{definition}[Reduced fragment decompositions]
\label{def: reduced fragment decompositions}
We define the reduced fragment decomposition as follows:

\noindent\textbf{(i) Reduced fragment decompositions of a fragment decomposition.}
Given a fragment decomposition $\bw = (\sw_{j, \ell, r, \FragSuppSet})$, a collection of nonnegative weights \yledit{$\bx = (\rsdw_{j, \ell, r})_{j\in [m], \ell \in [n], r \in [n+1]}$ with a collection of nonnegative weights $\bz = (\rsdz_{j, \ell, r, i})_{j\in [m], \ell, i \in [n], r\in [n+1]}$,} denoted by the pair $(\bx, \bz)$, is called a reduced fragment decomposition of $\bw$
(denoted $\bw \to (\bx, \bz)$) iff
\begin{alignat*}{2}
    \rsdw_{j, \ell, r} 
    = \sum\nolimits_{\FragSuppSet \subseteq \ValSet \cap [\val_\ell, \val_r)} \sw_{j, \ell, r, \FragSuppSet}~, \; 
    \forall j, \ell, r~; \quad 
    \rsdz_{j, \ell, r, i} 
    = \sum\nolimits_{ \FragSuppSet \subseteq \ValSet \cap [\val_\ell, \val_r)} \sw_{j, \ell, r, \FragSuppSet} \cdot \frag_{j, \ell, r, \FragSuppSet}(\val_i)~, \; 
    \forall j, \ell, r, i~,
\end{alignat*}
where $\frag_{j, \ell, r, \FragSuppSet}(\val_i)$ is the mass function at $\val_i$ of the fragment $\Frag_{j, \ell, r, \FragSuppSet}$.

\smallskip
\noindent\textbf{(ii) Reduced fragment decompositions of an extremal market decomposition.}
Let $\ba=(\alpha_{\op,\SuppSet})$ be an extremal market decomposition (a convex combination of extremal markets).
We say $(\bx, \bz)$ is a reduced fragment decomposition of $\ba$ (denoted $\ba \to (\bx, \bz)$) iff there exists a $\bw$ such that $\bw$ is a fragment decomposition of $\ba$ (i.e. $\ba \to \bw$) and $(\bx, \bz)$ is a reduced fragment decomposition of $\bw$ (i.e., $\bw \to (\bx, \bz)$). 
\end{definition}
Given a reduced fragment decomposition $(\bx,\bz)$ defined above, we observe two key properties. Fixing the index $j$, $\ell$ and $r$, the sum of $(\rsdz_{j, \ell, r, i})_{i\in[n]}$ equals to $\rsdw_{j, \ell, r}$ scaled by the factor $(\val_r - \val_\ell) / (\val_r - \cost_j)$ since $\rsdw_{j, \ell, r}$ stands for the tail mass of the undominated fragment $\Frag_{j, \ell, r}$ whereas $\rsdz_{j, \ell, r, i}$ stands for the point mass in the fragment. 
Moreover, since any fragment $\Frag_{j, \ell, r, \FragSuppSet}$ with $\FragSuppSet \subseteq \ValSet \cap [\val_\ell, \val_r)$ is first-order stochastically dominated by the undominated fragment $\Frag_{j, \ell, r}$, the market constructed by $(\rsdz_{j, \ell, r, i})_{i\in[n]}$ is also dominated by the undominated fragment $\Frag_{j, \ell, r}$ scaled by $\rsdw_{j, \ell, r}$. We call these two properties together \emph{proper domination} of fragments.

\begin{lemma} [Proper Domination]\label{lem: reduced fragment decompositions - proper domination}
  Let $(\bx,\bz)$ be a reduced fragment decomposition of $\bw$, denoted $\bw \to (\bx,\bz)$,
\ylreplace{where $\bx=(\rsdw_{j,\ell,r})_{j\in[m],\,\ell,r\in[n]}$ and
$\bz=(\rsdz_{j,\ell,r,i})_{j\in[m],\,\ell,r,i\in[n]}$}{where $\bx=(\rsdw_{j,\ell,r})_{j\in[m],\ell \in [n], r\in[n+1]}$ and
$\bz=(\rsdz_{j,\ell,r,i})_{j\in[m],\,\ell,i\in[n], r\in [n+1]}$}.
Fix $j$, $\ell$ and $r$ such that $\cost_j < \val_\ell \le \val_r$. Let $\ValDist^{\rsdz}_{j,\ell,r}$ be the induced market with mass function
    $\valpmf^{\rsdz}_{j,\ell,r}(\val_i) := \rsdz_{j,\ell,r,i}, i \in [n]$. Then for all $\val$, we have $\ValDist^{\rsdz}_{j,\ell,r}(\val) \le \Frag_{j,\ell,r}(\val) \cdot \rsdw_{j,\ell,r}$,
    i.e., $\ValDist^{\rsdz}_{j,\ell,r}$ is first-order stochastically dominated by scaled fragment
    $\rsdw_{j,\ell,r}\cdot \Frag_{j,\ell,r}$ (in tail order).
    The total mass also satisfies 
    \begin{align*}
        \sum\nolimits_{i\in[n]} \rsdz_{j,\ell,r,i}
        = \frac{\val_r-\val_\ell}{\val_r-\cost_j} \cdot \rsdw_{j,\ell,r}~.
    \end{align*}
\end{lemma}

Observe that all fragments dominated by the same undominated fragment contribute the same amount to the seller surplus.
So, we can merge the contribution of all fragments dominated by the same undominated fragment, and determine the total contribution using the primary component $\bx$ of the reduced fragment decomposition.
Moreover, the maximum / minimum social welfare can be determined in a similar way. 
Instead of counting the contribution of each fragment separately, we count the total contribution of all fragments dominated by the same undominated fragment at once, using the secondary component $\bz$ of a reduced fragment decomposition. Therefore, when we charge the seller surplus and the buyer surplus in a concrete way, which is linear to $(\bx, \bz)$, we can preserve the welfare outcome induced by the extremal market decomposition $\ba$.

\begin{proposition}
\label{lem: surplus of reduced fragment decompositions}
    The seller surplus and (maximum / minimum) buyer surplus induced by an extremal market decomposition $\ba$ are linear in its reduced fragment decomposition $(\bx, \bz)$ (i.e., $\ba \to (\bx, \bz)$).
\end{proposition}
\paragraph{Chains of decompositions.}
\label{sec: fragments and fragment decompositions}
Slightly abusing notation, we denote a {\em chain} of decompositions by $\AggMarket \to \ba \to \bw \to (\bx, \bz)$.
We also write sub-chains of the above, which means there {\em exist} intermediate objects such that the complete chain of decompositions is legitimate. 


\section{Feasibility of Decomposition by Flow Conservation}
\label{sec:flow}
Now we only need to handle the final issue: enforcing $\AggMarket \to (\bx, \bz)$ by linear constraints.
In fact, we will show how to accomplish a slightly weaker goal, which is still sufficient for our purposes: enforcing that there is a way to segment the aggregate market $\AggMarket$ ({\em not necessarily an extremal market decomposition}) such that the seller / buyer surplus is the same as what is induced by $(\bx, \bz)$.
We show that this can be done by setting up a discounted flow network over undominated fragments $(\bx, \bz)$ satisfying proper domination, and enforcing flow conservation. 
In this section, we first give all constraints of the discounted flow network and introduce the flow variables $\by$. 
We show if the aggregate market can be decomposed into the reduced fragment decompositions, i.e., $\AggMarket \to (\bx, \bz)$, then there exists flow variables $\by$ such that $(\bx, \by, \bz)$ satifies the flow conservation (\Cref{sec: discounted flow network}). 
We next show the sufficiency of these constraints: If a triple $(\bx, \by, \bz)$ satisfies all constraints of the discounted flow network, then it is feasible (i.e., we can find a market segmentation of the aggregate market $\AggMarket$ such that it induces the same seller surplus and the maximum/minimum buyer surplus as $(\bx, \bz)$). We prove the existence by designing an algorithm to construct a specific market segmentation (\Cref{sec: segmenting market by reduced fragment decompositions}). 

After this, we can close the loop to get our main result of the paper: The discounted flow network formulation precisely characterizes all achievable outcomes through market segmentation.

\begin{theorem}
\label{thm: the equvilence of market segmentation and reduced fragment composition}
    Fix an aggregate market $\AggMarket$.
    The following two claims are equivalent:
    \begin{itemize}
        \item There is a way to segment $\AggMarket$ such that the overall seller surplus is $\SellerSurplus$ and the overall buyer surplus is $\BuyerSurplus$.
        \item There exists $(\bx, \by, \bz)$ satisfying the exact composition, flow conservation, and proper domination, whose induced seller surplus is $\SellerSurplus$, whose induced minimum buyer surplus is at most $\BuyerSurplus$, and whose induced maximum buyer surplus is at least $\BuyerSurplus$.
    \end{itemize}
\end{theorem}
\begin{proof}[Proof sketch]
For simplicity, suppose that we want to maximize the buyer surplus through market segmentation.
Suppose that the maximum buyer surplus possible is $B^*$.
The full generality of extremal markets implies that there is an extremal market decomposition $\ba$ whose induced maximum buyer surplus is $B^*$, which can further be turned into $(\bx, \by, \bz)$ satisfying the exact composition and flow conservation.
As a result, there exists $(\bx, \by, \bz)$ whose induced maximum buyer surplus is $B^*$.
Conversely, take any $(\bx^*, \by^*, \bz^*)$ satisfying exact composition, flow conservation, and proper domination, whose induced buyer surplus is $B^{**} \ge B^*$.
There is a procedure that transforms $(\bx^*, \by^*, \bz^*)$ into a way to segment $\AggMarket$, whose induced buyer surplus is $B^{**}$, which cannot be larger than the maximum buyer surplus $B^*$ achievable through segmentation.
As a result, $B^* = B^{**}$.
\end{proof}

We have proved the equivalence of the discounted flow network $(\bx, \by, \bz)$ constrained by proper domination and flow conservation and market segmentation in terms of seller surplus and buyer surplus, and there are $O(mn^2)$ undominated fragments. Therefore, we have an LP of polynomial size which captures all achievable welfare outcomes. This implies a polynomial-time algorithm.
\begin{corollary}\label{corollary: polynomial-time algorithm for achievable seller buyer surplus}
    Given the aggregate market $\AggMarket$, there exists a polynomial-time algorithm to optimize the achievable seller/buyer surplus. 
\end{corollary}

\subsection{Discounted Flow Network}
\label{sec: discounted flow network}
\paragraph{Connectivity and discount factors.}
So far, we have been discussing how an extremal market decomposition can be broken into pieces to form a (reduced) fragment decompositions.
From now on, we will turn to the other direction, i.e., how to recover an extremal market decomposition (or more accurately, something similar) by sticking fragments back into extremal markets.
Below we will focus on undominated fragments, but essentially all the claims generalize to arbitrary fragments.

First, observe that two undominated fragments $(j_1, \ell_1, r_1)$ and $(j_2, \ell_2, r_2)$, where without loss of generality $j_1 \le j_2$, can possibly come from consecutive pieces in a single extremal market iff $j_1 + 1 = j_2$ and $r_1 = \ell_2$.
When this happens, we say $(j_1, \ell_1, r_1)$ {\em connects to} $(j_2, \ell_2, r_2)$, denoted by $(j_1 ,\ell_1, r_1) \ct (j_2, \ell_2, r_2)$.
Repeatedly applying this property, $m$ undominated fragments $(j_1, \ell_1, r_1), \dots, (j_m, \ell_m, r_m)$ can be combined into an extremal market iff $(j_1, \ell_1, r_1) \ct (j_2, \ell_2, r_2) \ct \dots \ct (j_m, \ell_m, r_m)$.

The next natural question is regarding the proportion of probability we need to put into each fragment when combining them into a single extremal market.
Given all $m$ fragments that each connect to the next one, it is easy to determine the proportion of probability in each fragment: consider the extremal market defined by $\op = (\ell_1, \ell_2, \dots, \ell_m)$ and the entire $\ValSet$ as the support.
In order to produce $\alpha \cdot \ValDist_{\op, \ValSet}$, clearly we need $\alpha \cdot \ValDist_{\op, \ValSet}(\ell_1) $ units of $\Frag_{1, \ell_1, \ell_2}$, $\alpha \cdot \ValDist_{\op, \ValSet}(\ell_2)$ units of $\Frag_{1, \ell_2, \ell_3}$, etc.
The real question is: can we determine the relative proportions of two consecutive fragments locally, without knowing the other ingredients of the extremal market to be formed?
The answer turns out to be positive: the ratio between the probabilities in two consecutive fragments $(j, \ell, k)$ and $(j + 1, k, r)$ is always the same in all extremal markets that contain both fragments.

\begin{lemma}
\label{lem: discount factor}
    Let $\ValDist_{\op, \SuppSet}$ be an extremal market, and $\bw$ be a fragment decomposition of $\ValDist_{\op, \SuppSet}$ (i.e., $\ValDist_{\op, \SuppSet} \to \bw$).
    Consider two consecutive fragments $\ValDist_{j, \ell, k, \FragSuppSet_j}$ and $\ValDist_{j + 1, k, r, \FragSuppSet_{j+1}}$ where $\val_\ell = \op_j$, $v_k = \op_{j + 1}$, $\val_r = \op_{j + 2}$, $\FragSuppSet_j = S \cap [\op_j, \op_{j + 1})$ and $\FragSuppSet_{j+1} = S \cap [\op_{j + 1}, \op_{j + 2})$.
    Then $\esw_{j + 1, k, r, \FragSuppSet_j}^{\op, \SuppSet} / \esw_{j, \ell, k, \FragSuppSet_{j+1}}^{\op, \SuppSet}$ depends only on $j$, $\ell$, and $k$.
\end{lemma}

\begin{remark}
\label{rmk: discount factor}
We define the ratio between two consecutive fragments in the fragment decompositions of any extremal market $\ValDist_{\op, \SuppSet}$ as {\em discount factor}, denoted 
\[
    \df_{j, \ell, k} = \frac{\esw_{j + 1, k, r, \FragSuppSet_{j+1}}^{\op, \SuppSet}}{\esw_{j, \ell, k, \FragSuppSet_{j}}^{\op, \SuppSet}} = \frac{\val_\ell - \cost_j}{\val_k - \cost_j}~,
\]
where $\esw_{j, \ell, k, \FragSuppSet_j}^{\op, \SuppSet}$ and $\esw_{j + 1, k, r,  \FragSuppSet_{j+1}}^{\op, \SuppSet}$ are both positive.
\end{remark}

\begin{figure}[t]
    \centering
    \includegraphics[width=\linewidth]{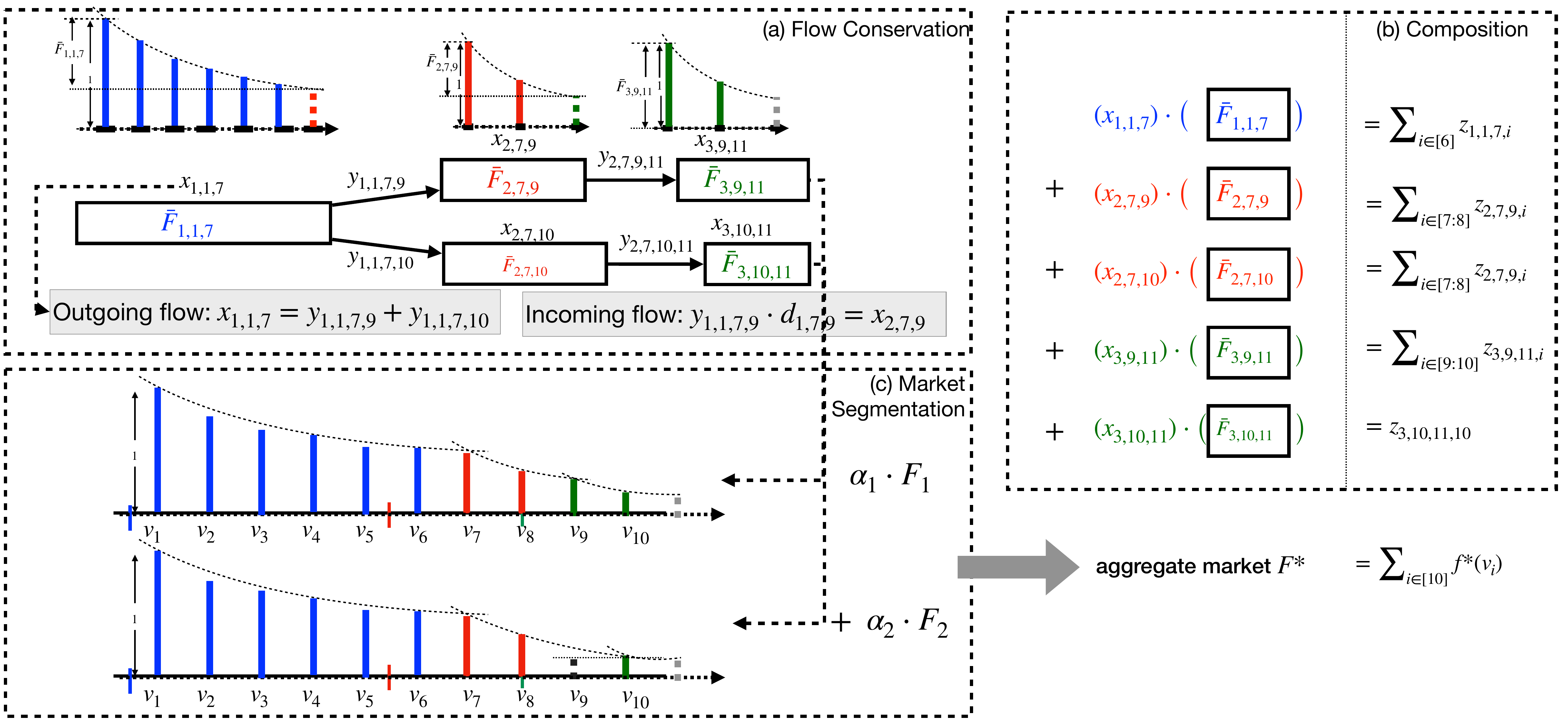}
    \caption{Feasible Discounted Flow Network (Example): 
    Fix the same aggregate market $\AggMarket$ as in \Cref{fig:fragment_decomposition}.
    Part (a) demonstrates in a feasible discounted flow network, how the fragments $\Frag_{j, \ell, r}$, where $\ell \in [n], r\in [n+1]$, connect and how the probability "flows" from one fragment to another discounted by the discount factor. 
    In Part (b), each fragment's total mass is further decomposed into the variables $(\rsdz_{j, \ell, r, i})_{i\in[n]}$ (only non-zero variables are displayed) and it satisfies the exact composition: $\sum_{j,\ell,r}\rsdz_{j,\ell, r,i} = \valpmf^*(\val_i)$. 
    Part (c) shows how the market segmentation is constructed by the flow network (\Cref{alg: market segmentation}). Note that $\ValDist_1$ and $\ValDist_2$ are not necessary to be extremal markets.}
    \label{fig:flow_network}
\end{figure}

\paragraph{Discounted flow networks.}
Now we are ready to define discounted flow networks, which capture how probabilities on different fragments in a (reduced) fragment decompositions can be combined into a way to segment the aggregate market, and define the surplus induced by the discounted flow networks.

As illustrated in Figure~\ref{fig:flow_network}, there are edges from one undominated fragment $(j, \ell, k)$ to another $(j + 1, k, r)$ iff the former fragment connects to the latter.
Moreover, there is a discount factor $\df_{j, \ell, k}$ associated with each edge.
On each edge, we further set up a variable $\rsdf_{j, \ell, k, r}$, corresponding to the probability ``flowing'' from $\Frag_{j, \ell, k}$ to $\Frag_{j + 1, k, r}$, which is ``discounted'' in the process (see further explanation below).
Given the aggregate market $\AggMarket$, we translate the constraint that $\AggMarket \to (\bx, \bz)$ into the following families of linear constraints:

\begin{itemize}
    \item Exact composition of aggregate market: for each $i \in [n]$, $\valpmf(\val_i) =\sum_{j, \ell, r} z_{j, \ell, r, i}$.
    \item Outgoing flow conservation: for each undominated fragment $\Frag_{j, \ell, k}$ \yledit{with $k\le n$}, $\rsdw_{j, \ell, k} = \sum_r \rsdf_{j, \ell, k, r}$. 
    \item Incoming flow conservation: for each undominated fragment $\Frag_{j + 1, k, r}$, $\rsdw_{j + 1, k, r} = \sum_\ell \rsdf_{j, \ell, k, r} \cdot \df_{j, \ell, k}$.
    \item Proper domination of fragments: for each undominated fragment $\Frag_{j, \ell, r}$, $\sum_{i:\, \val \le \val_i }\rsdz_{j,\ell,r,i} \le \Frag_{j,\ell,r}(\val) \cdot \rsdw_{j,\ell,r}$ for all $\val$. The total mass satisfies $\sum_{i \in [n]}\rsdz_{j, \ell, r, i} = (\val_r - \val_\ell)/(\val_r - \cost_j) \cdot \rsdw_{j, \ell, r} $.
    
\end{itemize}
 If we get a reduced fragment decomposition $(\bx, \bz)$ of $\AggMarket$, then the proper domination is already satisfied (\Cref{lem: reduced fragment decompositions - proper domination}) and we can always find the corresponding flow variable $\by$.
\begin{lemma}
\label{lem: the existance of y}
    Given the aggregate market $\AggMarket$, if $\AggMarket \to (\bx, \bz)$, then there exists $\by$ such that $(\bx, \by, \bz)$ satisfies exact composition and flow conservation.
\end{lemma}

\begin{definition}
\label{def: surplus by reduced fragment decompositions (x z)}
Fix a reduced fragment decompositions $(\bx, \bz)$, where $\bx = (\rsdw_{j, \ell, r})_{j \in [m], \ell, r \in [n]}$ and  $\bz = (\rsdz_{j, \ell, r, i})_{j \in [m], \ell, r, i \in [n]}$. We define the seller surplus induced by $(\bx, \bz)$, denoted $\SellerSurplus\left((\bx, \bz)\right)$,
\begin{equation}
    \begin{aligned}
        \SellerSurplus\left((\bx, \bz)\right) = \sum\nolimits_{j, \ell, r} \costpmf(\cost_j) \cdot (\val_\ell - \cost_j) \cdot \rsdw_{j, \ell, r}~,
    \end{aligned}
\end{equation}
the maximum (resp.\ minimum) social welfare induced by $(\bx, \bz)$, denoted $\MaxSocialWelfare\left((\bx, \bz)\right)$ (resp.\ $\MinSocialWelfare\left((\bx, \bz)\right)$),
\begin{align*}
    \text{Maximum social welfare:} \quad 
    \MaxSocialWelfare\left((\bx, \bz)\right) 
    & = \sum\nolimits_{j, \ell,r} \sum\nolimits_{j' \le j}\sum\nolimits_{i \in [\ell: r)} \costpmf(\cost_{j'}) \cdot (\val_i - \cost_{j'}) \cdot \rsdz_{j, \ell, r, i}\\ 
    \text{Minimum social welfare:} \quad 
    \MinSocialWelfare\left((\bx, \bz)\right) 
    & = \sum\nolimits_{j, \ell,r} \sum\nolimits_{j' < j}\sum\nolimits_{i \in [\ell: r)} \costpmf(\cost_{j'}) \cdot (\val_i - \cost_{j'}) \cdot \rsdz_{j, \ell, r, i}~.
\end{align*}
Moreover, the maximum buyer surplus induced by $(\bx, \bz)$ is $\MaxBuyerSurplus\left((\bx, \bz)\right) = \MaxSocialWelfare\left((\bx, \bz)\right) - \SellerSurplus\left((\bx, \bz)\right)$ and the minimum buyer surplus induced by $(\bx, \bz)$ is $\MinBuyerSurplus\left((\bx, \bz)\right) = \MinSocialWelfare\left((\bx, \bz)\right) - \SellerSurplus\left((\bx, \bz)\right)$.
\end{definition}

\begin{algorithm}[h!]
\captionsetup{skip=2pt} 
\begingroup
\footnotesize 
\caption{\textsc{SegmentMarketByReducedFragmentDecomposition}$((\bx, \by, \bz))$}
\begin{algorithmic}[1]
\setlength{\itemsep}{0pt}
\setlength{\topsep}{0pt}
\setlength{\partopsep}{0pt}
\setlength{\parskip}{0pt}
\label{alg: market segmentation}
\REQUIRE Reduced Fragment Decomposition $(\bx, \by, \bz)$ 
\ENSURE A market segmentation of $\MarketSegSet =\{(\alpha^{(t)},\MarketSegment^{(t)})\}_{t=1}^T$ of the aggregate market $\AggMarket$.

\STATE $\MarketSegSet\gets \emptyset$,  $t \gets 0$.
\WHILE{$\bx \neq \textbf{0}$}
     \STATE $t \gets t+1$
     \STATE \Comment{Within each outer iteration $t$, we omit the superscript $(t)$ on $i_j$ and $\mscdf_j$ for readability.}
    \STATE Choose any $(i_1, i_2)$ such that $x_{1, i_1, i_2} > 0$, let $p_1 \gets x_{1, i_1, i_2}$.
    \STATE $j \gets 1$.
        
        \WHILE {\ylreplace{$\val_{i_{j+1}} \ne \infty$}{$i_{j+1} \le n$}}  \label{alg-line: market segmentation - iteration t while starts} 
        \STATE \Comment{When $i_{j+1} = n+1$, then $\op_{j+1} = \val_{n+1} = \infty$ }
        \STATE Choose any $i_{j + 2}$ such that $\rsdf_{j, i_{j}, i_{j + 1}, i_{j + 2}} > 0$.
        
        \IF {$\rsdf_{j, i_{j}, i_{j} + 1, i_{j + 2}} < \mscdf_j$}
            \FOR{ $j' = 1$ to $j$}
                \STATE $\mscdf_{j'} \gets \mscdf_{j'} \cdot \rsdf_{j, i_j, i_{j + 1}, i_{j + 2}} / \mscdf_{j}$. \label{alg-line: maket segmentation - update p_j} 
            \ENDFOR
        \ENDIF
        
        \STATE $p_{j + 1} \gets \mscdf_j \cdot d_{j, i_{j}, i_{j + 1}}$. \label{alg-line: market segmentation - p_j+1} \Comment{$d_{j, i_{j}, i_{j + 1}}$ is the discount factor we defined in \Cref{rmk: discount factor}. }
        
             
        \STATE $j \gets j+1$.
    \ENDWHILE \label{alg-line: market segmentation - iteration t while ends}
    \STATE $j^{(t)} \gets j$. \label{alg-line: market segmentation - assign j^(t)} \Comment{$j^{(t)}$ is the final value of $j$ in iteration $t$, so that $i_1,\dots,i_{j^{(t)}+1}$ are chosen and $\val_{i_{j^{(t)}+1}}=\infty$.}

    \STATE Create a new market $\MarketSegment$, where for each $i \in [n]$, $\marketSegmentPmf(\val_i) = \sum_{j \in [j^{(t)}]} \frac{z_{j, i_{j}, i_{j + 1}, i} \cdot \mscdf_{j}} {x_{j, i_{j}, i_{j + 1}}}$. \label{alg-line: f(v_i)}\Comment{Since the flow variable $y>0$, the incoming flow conservation (proved in \Cref{lem: market segmentation - x y z feasible}) implies $x>0$ as well.}
    \STATE $\alpha^{(t)} \gets \mscdf_1^{(t)}$; $\MarketSegment^{(t)} \gets \MarketSegment / \alpha^{(t)}$. \label{alg-line: alpha^t and market segment F}
    \STATE $\MarketSegSet \gets \MarketSegSet\cup \{(\alpha^{(t)}, \MarketSegment^{(t)})\}$.
    \FOR {$j = 1$ to $j^{(t)}$}
        \STATE $x_{j, i_j, i_{j + 1}} \gets x_{j, i_j, i_{j + 1}} - \mscdf_{j}$. \label{alg-line: market segmentation - x}
        \STATE \algorithmicif\ $j \ne j^{(t)}$ \algorithmicthen\ $\rsdf_{j, i_j, i_{j + 1}, i_{j + 2}} \gets \rsdf_{j, i_j, i_{j + 1}, i_{j + 2}} - \mscdf_{j}$.
        \label{alg-line: market segmentation - y}
        \FOR {$i= i_j$ to $i_{j + 1} - 1$} \label{alg-line: market segmentation - z}
            \STATE $z_{j, i_j, i_{j + 1}, i} \gets z_{j, i_j, i_{j + 1}, i} - \alpha^{(t)} \cdot \marketSegmentPmf^{(t)}(\val_i)$.
        \ENDFOR
    \ENDFOR 
\ENDWHILE
\RETURN $\MarketSegSet$
\end{algorithmic}
\endgroup
\end{algorithm}

\subsection{Market Segmentation from Reduced Fragment Decompositions} \label{sec: segmenting market by reduced fragment decompositions}
We claim that any $(\bx, \bz)$ can be turned into a segmentation of the aggregate market $\AggMarket$ (which is similar to an extremal market decomposition, as explained below), if there exist a flow $\by$, such that $(\bx, \by, \bz)$ satisfies the constraints of the discounted flow network. We prove the existence by constructiion: The core idea is to repeatedly find augmenting paths and turns each path into a weighted market that is ``essentially extremal''. The procedure is described in \Cref{alg: market segmentation}.

To see why the procedure works, we first observe that 
the the triple $(\bx, \by, \bz)$ satisfies all constraints (exact composition, flow conservation, and proper domination) throughout the execution of the above procedure.

\begin{lemma}
\label{lem: market segmentation - x y z feasible}
    Let $\left(\bx^{(1)}, \by^{(1)}, \bz^{(1)}\right)$ be the input reduced fragment decomposition to \Cref{alg: market segmentation}. For each outer iteration $t$ executed by the algorithm, let $\left(\bx^{(t+1)},\by^{(t+1)},\bz^{(t+1)}\right)$ denote the tensors obtained from $\left(\bx^{(t)},\by^{(t)},\bz^{(t)}\right)$ after applying the updates in Lines~\ref{alg-line: market segmentation - x}--\ref{alg-line: market segmentation - z}. Then, for every such $t$, the triple $\left(\bx^{(t+1)}, \by^{(t+1)}, \bz^{(t+1)}\right)$ satisfies the constraints (outgoing/incoming flow conservation, fragment composition, and proper domination) as long as $\left(\bx^{(t)}, \by^{(t)}, \bz^{(t)}\right)$ satisfies the same constraints.
\end{lemma}

Since flow conservation remains valid throughout the execution of the above procedure, as long as $\bx \ne 0$, there exists a complete path (assuming in iteration $t$) from the leftmost layer of the flow network to some fragment to the right where $i_{j^{(t)}+1} = n+1$ ( $j^{(t)}$ is the final value of $j$ in iteration $t$), which implies $\val_{i_{j^{(t)}+1}} = \op_{j^{(t)}+1} = \infty$, where all the weights of undominated fragments $(\rsdw_{j, \ell, r})$ flowing through and the flow variables $(\rsdf_{j, \ell, k, r})$ involved are strictly positive.
This means we must be able to find an augmenting path. Therefore, when the procedure ends, all the weights of the reduced fragment $\rsdw_{j, \ell, r}$ and flow variables $\rsdf_{j, \ell, k, r}$ become zero.
\begin{corollary}
\label{lem: market segmentation stop condition}
    \Cref{alg: market segmentation} stops iff the remaining $\bx = \textbf{0}$.
\end{corollary}

Also, observe that the new market $\MarketSegment$ constructed in each iteration of the above procedure satisfies: $\op(\MarketSegment) = (\val_{i_1}, \val_{i_2}, \dots, \val_{i_{j^{(t)}}}, \infty, \dots, \infty)$, and in fact, for each $j \in [j^{(t)}]$, $\{v_{i_j}, v_{i_{j + 1}}\} \subseteq \Op_j(\MarketSegment) \subseteq [v_{i_j}, v_{i_{j + 1}}]$.
This implies that the seller surplus and the maximum/minimum buyer surplus in the new market $\MarketSegment$, weighted by $\alpha$, is precisely the change in the seller / buyer surplus induced by $(\bx, \bz)$.
Summing over all iterations, we see that the above procedure turns $(\bx, \by, \bz)$ into a way to segment the aggregate market $\AggMarket$ while preserving the seller / buyer surplus.
\begin{lemma}
\label{lem: market segment - exists segmentation preserves surplus}
    Given the aggregate market $\AggMarket$,
    if $(\bx, \by, \bz)$ satisfies the exact composition, flow conservation, and proper domination constraints (depicted as the discounted flow network), then there exists a way to segment $\AggMarket$ (as \Cref{alg: market segmentation}), which induces the same seller surplus and maximum/minimum buyer surplus as $(\bx, \bz)$.
\end{lemma}

\newpage
\bibliographystyle{plainnat}
\bibliography{mybib}

\newpage
\appendix

\section{Missing Proofs in \Cref{sec:prelim}}
\label{apx:prelim proof}
\begin{proof}[Proof of \Cref{lem:monotonicity of optimal price set -v2}] 
Given the market $\ValDist$ and the cost $\cost_j$, we denote the largest element in the optimal price set by $r_j$ ($r_j = \max \Op_j(\ValDist)$) and the smallest element in the optimal price set by $l_j$ ($l_j = \min \Op_j(\ValDist)$).  
Given $r_j$ and $l_{j+1}$, which are the optimal price of $\cost_j$ and $\cost_{j+1}$ separately, we have 
\begin{equation}
\label{eq:oprimal price set monotonicity - 1}
    \ValDist(r_j)\cdot(r_j - \cost_{j+1}) \le \ValDist(l_{j+1}) \cdot (l_{j+1} - \cost_{j+1}) 
\end{equation} 
and 
\begin{equation}
\label{eq:oprimal price set monotonicity - 2}
    \ValDist(r_j) \cdot (r_j - \cost_j) \ge \ValDist(l_{j+1}) \cdot (l_{j+1} - \cost_{j})~.
\end{equation}
We can further derive that
\[
   \cost_{j+1}\cdot (\ValDist(l_{j+1}) - \ValDist(r_j))  \le l_{j+1}\cdot \ValDist(l_{j+1}) - r_j \cdot \ValDist(r_j)  \le \cost_j \cdot ( \ValDist(l_{j+1}) - \ValDist(r_j))~.
\]
Since $\cost_j < \cost_{j+1}$, then $\ValDist(l_{j+1}) - \ValDist(r_j) \le 0$. Let us discuss the potential value of $\ValDist(l_{j+1})$ and $\ValDist(r_j)$ in these two cases.
    \begin{itemize}
        \item If $\ValDist(l_{j+1}) = \ValDist(r_j)$, the assumption $\max_\val (\val-\cost_j)\cdot \ValDist(\val) > 0$ indicates that $\ValDist(r_j) > 0$, then $\ValDist(l_{j+1})>0$ as well. Combining with \Cref{eq:oprimal price set monotonicity - 1} and \Cref{eq:oprimal price set monotonicity - 2}, we can derive
        $l_{j+1} = r_j$, then
        \[
            \max \Op_j(\ValDist) = \min \Op_{j+1}(\ValDist)~;
        \]
        \item if $\ValDist(l_{j+1}) - \ValDist(r_j) < 0$, we have $ r_j < l_{j+1}$, then
        \[
            \max \Op_j(\ValDist) < \min \Op_{j+1}(\ValDist)~.
        \]
     \end{itemize}
Therefore, we have proved the monotonicity of $\Op_j(\ValDist)$.
\end{proof}

\begin{proof}[Proof of \Cref{lem: optimal price set - scaled F}]
    Given a fixed $\alpha \in [0,1]$, $\Op_j(\alpha \ValDist) = \argmax_\val (\val - \cost_j) \cdot \alpha \ValDist(\val) = \argmax_\val (\val - \cost_j) \cdot \ValDist(\val) 
    = \Op_j(\ValDist)$. We conclude the lemma.
\end{proof}

\section{Missing Proofs in \Cref{sec:extremal}}
\label{apx:extremal proof}

\begin{proof}[Proof of \Cref{lem:optimal_price_for_extremal_market}] 
We separate our proof based on whether $\op_j= \infty$ or $\op_j \ne \infty$.

\textbf{When $\op_j= \infty$}, it satisfies $\val_i \notin S, \forall \val_i \in \{v \in \ValSet: \val > \cost_j\}.$ Therefore, we have $\ValDist(\val) = 0$ for all $\val \in \{\val \in \ValSet: \val > \cost_j\}$, we could derive that the maximum seller surplus of $c_j$ is 0. Therefore, the optimal price set is $\Op_j(\ValDist_{\op, \SuppSet}) =  \{\val \in \ValSet: \ValDist_{\op, \SuppSet}(\val)  = 0 \}$ and the optimal price $\op_j(\ValDist) = \op_j = \infty$.

\textbf{When $\op_j \ne \infty$}, if we prove for any $s \in [\op_j , \op_{j + 1}] \cap \SuppSet$,
\begin{align*}
    \ValDist_{\op, \SuppSet}(s) \cdot (s - \cost_j) = \ValDist_{\op, \SuppSet}(\val_i) \cdot (\val_i - \cost_j), \quad & \val \in [\op_j , \op_{j + 1}] \cap \SuppSet~, \\
    \ValDist_{\op, \SuppSet}(s) \cdot (s - \cost_j) > \ValDist_{\op, \SuppSet}(\val_i) \cdot (\val_i - \cost_j), \quad & \val \notin [\op_j , \op_{j + 1}] \cap \SuppSet~,
\end{align*}
then we prove the optimal set $\Op_j(\ValDist_{\op, \SuppSet}) = [\op_j, \op_{j+1}] \cap \SuppSet$ and the minimum optimal price $\op_j(\ValDist_{\op, \SuppSet}) = \op_j$.
The first equality can be easily observed from the definition of the extremal market.
For the inequality, fix $j \in [m]$. For any $\val \notin [\op_j , \op_{j + 1}] \cap \SuppSet$ we discuss three cases:
\begin{itemize}
    \item \textbf{if $\val \in [\op_j, \op_{j+1}]$}. Since $\val \notin  [\op_j , \op_{j + 1}] \cap \SuppSet$, $\val \notin \SuppSet$.  Therefore, $\ValDist_{\op, \SuppSet}(\val) = \ValDist_{\op, \SuppSet}(\val^{+})$, where $\val^{+} = \min \{v'\in \SuppSet: \val'> \val\}$, therefore, 
    \[
      \ValDist_{\op, \SuppSet}(\val) \cdot (\val - \cost_{j}) = \ValDist_{\op, \SuppSet}(\val^{+}) \cdot (\val - \cost_{j}) < \ValDist_{\op, \SuppSet}(\val^{+}) \cdot (\val^{+} - \cost_{j})~;
    \]
    \item \textbf{if $\val > s$}: we assume $\val \in S$ without loss.\footnote{\label{ft:v not in S}If $\val \notin \SuppSet$, either $\ValDist_{\op, \SuppSet}(\val) = 0$ or we could always find a $\val_{i+k} \in \SuppSet$ such that $\ValDist_{\op, \SuppSet}(\val) = \ValDist_{\op, \SuppSet}(\val^{+}) \ne 0$, in which case $\ValDist_{\op, \SuppSet}(\val) \cdot (\val - \cost_{j}) < \ValDist_{\op, \SuppSet}(\val^{+}) \cdot (\val^{+} - \cost_{j})$.}
    We would like to prove for all $ k \in [j+1: m]$, it has $\ValDist_{\op, \SuppSet}(s) \cdot (s - \cost_j) > \ValDist_{\op, \SuppSet}(\val) \cdot (\val - \cost_j), \forall \val \in (\op_k , \op_{k + 1}] \cap \SuppSet$ (recall that $q_{m + 1} = \infty$). We prove it by induction:
    \begin{enumerate}
        \item Given $k = j+1$ and $\val \in  (\op_{j+1} , \op_{j + 2}] \cap \SuppSet$, such equality
        $
         \ValDist_{\op, \SuppSet}(\op_{j+1}) \cdot (\op_{j+1} - \cost_{j+1}) = \ValDist_{\op, \SuppSet}(\val) \cdot (\val - \cost_{j+1})
        $ holds. 
        
        Since $\val > \op_{j+1} >\cost_{j+1} > \cost_{j} \ge 0$ and the function $h(c) = (\val -c)/(\op_{j+1} - c)$ is increasing in $c$ with $c < \op_{j+1}$, then $\frac{\val - \cost_{j+1}}{\op_{j+1} - \cost_{j+1}} > \frac{\val - \cost_{j}}{\op_{j+1} - \cost_{j}}$.  
        We can derive that 
        \[
            \ValDist_{\op, \SuppSet}(\op_{j+1}) \cdot (\op_{j+1} - \cost_{j}) > \ValDist_{\op, \SuppSet}(\val) \cdot (\val - \cost_{j})~.
        \]
        Thus, given $s \in [\op_j, \op_{j+1}] \cap \SuppSet$, we can prove that 
        \begin{align*}
            \ValDist_{\op, \SuppSet}(s) \cdot (s - \cost_{j}) > \ValDist_{\op, \SuppSet}(\val) \cdot (\val - \cost_{j})~.
        \end{align*}  
        \item Suppose given $k \ge j+1$, it has $\ValDist_{\op, \SuppSet}(s) \cdot (s - \cost_j) > \ValDist_{\op, \SuppSet}(\val) \cdot (\val - \cost_j), ~\forall \val \in (\op_k , \op_{k + 1}] \cap\SuppSet$, we prove it also holds for $k+1$. Given $\val \in (\op_{k+1}, \op_{k+2}] \cap \SuppSet$, we have
        \[
            \ValDist_{\op, \SuppSet}(\op_{k+1}) \cdot (\op_{k+1} - \cost_{k+1}) = \ValDist_{\op, \SuppSet}(\val) \cdot (\val - \cost_{k+1})~.
        \]
         Similarly, since $\frac{r - \cost_{k+1}}{\op_{k+1} - \cost_{k+1}} > 1$ and $\cost_{k+1} > \cost_{j} \ge 0$, we have $\frac{\val - \cost_{k+1}}{\op_{k+1} - \cost_{k+1}} > \frac{r - \cost_{j}}{\op_{k+1} - \cost_{j}}$. Then
         \[
             \ValDist_{\op, \SuppSet}(\op_{k+1}) \cdot (\op_{k+1} - \cost_{j}) > \ValDist_{\op, \SuppSet}(\val) \cdot (\val - \cost_{j})~.
        \]
        Since we know
        \[
             \ValDist_{\op, \SuppSet}(\supp) \cdot (\supp - \cost_{j}) > \ValDist_{\op, \SuppSet}(\op_{k+1}) \cdot (\op_{k+1} - \cost_{j})~, \quad \forall \supp \in [\op_j, \op_{j+1}] \cap \SuppSet~.
        \]
        Thus, given $s \in [\op_j, \op_{j+1}] \cap \SuppSet$, we can prove that
        \begin{align*}
             \ValDist_{\op, \SuppSet}(s) \cdot (s - \cost_{j})
             > \ValDist_{\op, \SuppSet}(\val) \cdot (\val - \cost_{j})~.
        \end{align*}
    \end{enumerate}
    Thus, we can prove that for the inequality holds for all $s \in [\op_j , \op_{j + 1}] \cap \SuppSet, \val \notin [\op_j , \op_{j + 1}] \cap \SuppSet$ and $\val >s$;
\item \textbf{if $\val < s$}:  we assume $\val > \cost_j$ and $\val \in \SuppSet$ without loss.\footref{ft:v not in S} We would like to prove for all $ k \in [1, j-1]$, it has $\ValDist_{\op, \SuppSet}(s) \cdot (s - \cost_j) > \ValDist_{\op, \SuppSet}(\val) \cdot (\val - \cost_j), ~\forall s \in [\op_j , \op_{j + 1}] \cap \SuppSet, \val \in [\op_k , \op_{k + 1}) \cap \SuppSet$. We prove it by backward induction:
    \begin{enumerate}
        \item Given  $k = j-1$, for every $\val \in  [\op_{j-1} , \op_{j}) \cap \SuppSet$ we have
        \[
            \ValDist_{\op, \SuppSet}(\op_j) \cdot (\op_j - \cost_{j-1}) = \ValDist_{\op, \SuppSet}(\val) \cdot (\val - \cost_{j-1})~,
        \]
        based on the first identity. Since $0<\frac{\val - \cost_{j-1}}{\op_j - \cost_{j-1}} < 1$ and $0\le \cost_{j-1} < \cost_{j}$, we have $\frac{\val - \cost_{j-1}}{\op_j - \cost_{j-1}} > \frac{\val - \cost_{j}}{\op_j - \cost_{j}}$. Then
        \begin{equation}
        \label{eq:extremal_market_property_proof_2}
            \ValDist_{\op, \SuppSet}(\op_j) \cdot (\op_j - \cost_{j}) > \ValDist_{\op, \SuppSet}(\val) \cdot (\val - \cost_{j})~.
        \end{equation}
        Thus, given $s \in [\op_j, \op_{j+1}] \cap \SuppSet$, we can prove that 
        \begin{align*}
              \ValDist_{\op, \SuppSet}(s) \cdot (s - \cost_{j}) > \ValDist_{\op, \SuppSet}(\val) \cdot (r - \cost_{j})~.
        \end{align*}
        \item Suppose given $k$, it has $\ValDist_{\op, \SuppSet}(s) \cdot (s - \cost_j) > \ValDist_{\op, \SuppSet}(\val) \cdot (\val - \cost_j), ~\forall \val \in [\op_k , \op_{k + 1}) \cap \SuppSet$, we prove it also holds for $k-1$. 
        Given $\val \in [\op_{k-1}, \op_{k}) \cap \SuppSet$, we have
        \[
            \ValDist_{\op, \SuppSet}(\op_{k}) \cdot (\op_{k} - \cost_{k-1}) = \ValDist_{\op, \SuppSet}(\val) \cdot (\val - \cost_{k-1})~.
        \]
         Similarly, since $0< \frac{\val - \cost_{k-1}}{\op_{k} - \cost_{k-1}} < 1$ and $0\le \cost_{k-1} < \cost_{j}$, we have $\frac{\val - \cost_{k-1}}{\op_{k} - \cost_{k-1}} > \frac{\val - \cost_{j}}{\op_{k} - \cost_{j}}$. Then
         \[
             \ValDist_{\op, \SuppSet}(\op_{k}) \cdot (\op_{k} - \cost_{j}) > \ValDist_{\op, \SuppSet}(\val) \cdot (\val - \cost_{j})~.
        \]
        We have the assumption that
        \[
             \ValDist_{\op, \SuppSet}(s) \cdot (s - \cost_{j}) > \ValDist_{\op, \SuppSet}(\op_{k}) \cdot (\op_{k} - \cost_{j}), \quad \forall \supp \in [\op_j, \op_{j+1}] \cap \SuppSet.
        \]
        Thus, given $s \in [\op_j, \op_{j+1}] \cap \SuppSet$, we can prove that 
        \begin{align*}
              \ValDist_{\op, \SuppSet}(s) \cdot (\supp - \cost_{j})
              > \ValDist_{\op, \SuppSet}(\op_{k}) \cdot (\op_{k} - \cost_{j})
              > \ValDist_{\op, \SuppSet}(\val) \cdot (\val - \cost_{j})~.
        \end{align*}
        
    \end{enumerate}
    Therefore, we conclude that the inequality holds for all $\val \notin [\op_j , \op_{j + 1}] \cap \SuppSet$ and $\val < s$.
\end{itemize}
In summary, when $\op_j \ne \infty$, both the equality and inequality hold. Therefore, the optimal price set $\Op_j(\ValDist_{\op, \SuppSet}) = [\op_j, \op_{j+1}] \cap \SuppSet$ and the minimum optimal price $\op_j(\ValDist_{\op, \SuppSet}) = \op_j$.
\end{proof}

\begin{proof}[Proof of \Cref{lem: extremal market decomposition - optimal price set expansion}] In the algorithm we assign $\op_j = \op_j(\ValDist)$ and recall that \Cref{lem: optimal price set of extremal market} proves $\op_j(\ValDist_{\op, \SuppSet}) = \op_j$, so these three notation stands for the same value. We denote them as $\op_j$ in the remaining proof to simplify the notation. We discuss $\cost_j$ in different cases and in each case we prove $\Op_j(\ValDist) \subseteq \Op_j(\alpha_{\op, \SuppSet} \cdot \ValDist_{\op, \SuppSet})$ and $\Op_j(\ValDist) \subseteq \Op_j(\ValDist - \alpha_{\op, \SuppSet} \cdot \ValDist_{\op, \SuppSet})$ separately.

\paragraph{Suppose $\max_\val(\val - \cost_j ) \cdot \ValDist(\val) = 0$.}
The optimal price set of $\cost_j$ in the market  
\[\Op_j(\ValDist) = \{\val \in \ValSet: (\val - \cost_j) \cdot \ValDist(\val)  = 0\}~.\]

For any $\val > \cost_j$, since $\ValDist(\val) = 0$, then $\val$ is not in the support set of $\ValDist$; thus, $\val \notin \SuppSet$. Then for any $\val > \cost_j$, it has $\ValDist_{\op, \SuppSet}(\val) = 0$. We can get the optimal set of $\cost_j$ in the extremal market 
\[
\Op_j(\ValDist_{\op, \SuppSet}) = \{\val \in \ValSet: (\val - \cost_j) \cdot \ValDist_{\op, \SuppSet}(\val)  = 0\}~.
\]
We just need to prove if $\ValDist(\val)  = 0$, then $\ValDist_{\op, \SuppSet}(\val) = 0$.
When $\ValDist(\val)  = 0$, it implies for any $\val' > \val$, $\ValDist(\val')  = 0$ also holds; thus, $\ValDist_{\op, \SuppSet}(\val) = 0$ also holds. 
Therefore, for every $\val \in \Op_j(\ValDist)$, we can derive $\val \in \Op_j(\ValDist_{\op, \SuppSet})$, which proves $\Op_j(\ValDist) \subseteq \Op_j(\ValDist_{\op, \SuppSet})$.

 Since for every $\val\ge \cost_j$, $0 \le (\val - \cost_j ) \cdot (\ValDist(\val) - \alpha_{\op, \SuppSet} \cdot \ValDist_{\op, \SuppSet}(\val)) \le  (\val - \cost_j ) \cdot \ValDist(\val)$,
$\max_\val(\val - \cost_j ) \cdot (\ValDist(\val) - \alpha_{\op, \SuppSet} \cdot \ValDist_{\op, \SuppSet}(\val)) = 0$ also holds. Therefore, 
\[
\Op_j(\ValDist - \alpha_{\op, \SuppSet} \cdot \ValDist_{\op, \SuppSet}) = \{\val \in \ValSet: (\val - \cost_j) \cdot \left(\ValDist(\val) - \alpha_{\op, \SuppSet} \cdot \ValDist_{\op, \SuppSet}(\val)\right)  = 0\}~.
\]
If $\ValDist(\val)  = 0$, then $\ValDist_{\op, \SuppSet}(\val) = 0$, and thus $\ValDist(\val) - \alpha_{\op, \SuppSet} \cdot \ValDist_{\op, \SuppSet}(\val) = 0$. Therefore, for every $\val \in \Op_j(\ValDist)$, we have
$\ValDist(\val_i) - \alpha_{\op, \SuppSet} \cdot \ValDist_{\op, \SuppSet}(\val_i) = 0$. We can derive that $\Op_j(\ValDist) \subseteq \Op_j(\ValDist - \alpha_{\op, \SuppSet} \cdot \ValDist_{\op, \SuppSet})$.
\paragraph{Suppose $\max_\val(\val - \cost_j ) \cdot \ValDist(\val) > 0$.}

\textbf{Part 1: we prove $\Op_j(\ValDist) \subseteq \Op_j(\alpha_{\op, \SuppSet} \cdot \ValDist_{\op, \SuppSet})$.} Since $\max_\val(\val - \cost_j ) \cdot \ValDist(\val) \ne 0$, we have $\op_j \in \SuppSet$. Due to the monotonicity of $\Op_j(\ValDist)$, we have $\max \Op_j(\ValDist) \le \min \Op_{j+1}(\ValDist) \le \op_{j+1}$ (considering $\op_{j+1}$ is either $\min \Op_{j+1}(\ValDist)$ or $\infty$). 
Therefore, $\Op_j(\ValDist)  \subseteq [\op_j, \op_{j+1}]$. 
For any value $\val \in \{\val' \in \ValSet : \val' \in [\op_j, \op_{j+1}], \val' \notin \SuppSet\}$, $\ValDist(\val) =\ValDist(\val^{+})$, where $\val^{+} = \min\{\val' \in \SuppSet: \val' > \val \}$.\footnote{If $\val^{+}$ does not exist, then $\ValDist(\val) = 0$.} 
The seller surplus of $\cost_j$ when setting price as $\val$ is $(\val-\cost_j)\cdot \ValDist(\val) = (\val-\cost_j)\cdot \ValDist(\val^{+})  <(\val^{+} -\cost_j)\cdot \ValDist(\val^{+})$, which shows $\val$ is not the optimal price of $\cost_j$. Therefore, if the value $\val \in \{\val' \in \ValSet: \val' \in [\op_j, \op_{j+1}], \val' \notin \SuppSet\}$, which is not in the support set $\SuppSet$, $\val_i \notin \Op_j(\ValDist)$ as well. 
Therefore,  $\Op_j(\ValDist)\subseteq [\op_j, \op_{j+1}] \cap \SuppSet = \Op_j(\ValDist_{\op, \SuppSet})$.
We conclude that $\Op_j(\ValDist)\subseteq \Op_j(\ValDist_{\op, \SuppSet})$.

\medskip
\textbf{Part 2: we prove $\Op_j(\ValDist) \subseteq \Op_j(\ValDist - \alpha_{\op, \SuppSet} \cdot \ValDist_{\op, \SuppSet})$.}

Recall that $\op_j \in \Op_j(\ValDist)$ when the optimal seller revenue is not $0$, for any $s \in \SuppSet$, if  $s \in  \Op_j(\ValDist)$,
\[\ValDist(\op_j)(\op_j-\cost_j) = \ValDist(s)(s-\cost_j)~,\]
if  $s \notin \Op_j(\ValDist)$
\[\ValDist(\op_j)(\op_j-\cost_j) > \ValDist(s)(s-\cost_j)~.\]
Since $\op_j \in \Op_j(\ValDist)$, we also know that given $\alpha_{\op, \SuppSet} \in [0,1]$, for all $s \in \Op_j(\alpha_{\op, \SuppSet} \cdot \ValDist_{\op, \SuppSet})$, it has $\ValDist_{\op, \SuppSet}(\op_j)(\op_j-\cost_j) = \ValDist_{\op, \SuppSet}(s)(s-\cost_j)$.
Thus, we can derive that if $s \in \Op_j(\alpha_{\op, \SuppSet} \cdot \ValDist_{\op, \SuppSet})$
\begin{equation}
\label{eq:lemma_subset_proof_1}
(\ValDist(\op_j) - \alpha_{\op, \SuppSet} \cdot \ValDist_{\op, \SuppSet}(\op_j)) (\op_j - \cost_j) \ge (\ValDist(s) - \alpha_{\op, \SuppSet} \cdot \ValDist_{\op, \SuppSet}(s)) (s-\cost_j)~, 
\end{equation}
where the equality holds if and only if $s \in \Op_j(\ValDist)$.\\
We discuss different cases of $\alpha$ in the while loop.
\begin{itemize}
    \item \textbf{Case 1: $\alpha = \alpha_{\text{runout}} < \alpha_{\text{shift}} $.} Then all the minimum optimal prices $\op_j$, where $j \in [m]$, remain them same. Thus, $\op(\ValDist - \alpha_{\op, \SuppSet} \cdot \ValDist_{\op, \SuppSet}) = \op(\ValDist)$. This indicates that $\op_j(\ValDist) \in \Op_j(\ValDist - \alpha_{\op, \SuppSet} \cdot \ValDist_{\op,\SuppSet})$ since we have proved that for all $s \in \Op_j(\ValDist)$
    \[
        (\ValDist(\op_j) - \alpha_{\op, \SuppSet} \cdot \ValDist_{\op, \SuppSet}(\op_j)) (\op_j - \cost_j) = (\ValDist(s) - \alpha_{\op, \SuppSet} \cdot \ValDist_{\op, \SuppSet}(s)) (s-\cost_j)~, 
    \]
    we conclude that $s \in \Op_j(\ValDist - \alpha_{\op, \SuppSet}\cdot \ValDist_{\op,\SuppSet}), ~\forall s \in \Op_j(\ValDist)$. Thus, $\Op_j(\ValDist) \subseteq \Op_j(\ValDist - \alpha_{\op, \SuppSet}\cdot \ValDist_{\op,\SuppSet})$.
    \item \textbf{Case 2: $\alpha =  \alpha_{\text{shift}} \le \alpha_{\text{runout}}$}. This indicates that some minimum optimal price(s) $\op_j$, where $j \in [n]$ have been shifted.
    
    We prove it by contradiction. Assume that for a specific cost $\cost_j$, $\Op_j(\ValDist) \nsubseteq \Op_j(\ValDist - \alpha_{\op, \SuppSet} \cdot \ValDist_{\op, \SuppSet})$ and $\max_\val(\val - \cost_j) \ValDist(\val) > 0$.
    \Cref{eq:lemma_subset_proof_1} shows that if there exists an element $s \in \Op_j(\ValDist)$ and $s \notin \Op_j(\ValDist - \alpha_{\op, \SuppSet} \cdot \ValDist_{\op, \SuppSet})$, all elements $ s \in \Op_j(\ValDist)$ satisfies $s \notin \Op_j(\ValDist - \alpha_{\op, \SuppSet} \cdot \ValDist_{\op, \SuppSet})$.  This indicates that there exists a $s^*$ such that $s^*\in \Op_j(\ValDist - \alpha_{\op, \SuppSet} \cdot \ValDist_{\op, \SuppSet})$ and $(\ValDist(s) - \alpha_{\op, \SuppSet} \cdot \ValDist_{\op, \SuppSet}(s)) (s - \cost_j) < (\ValDist(s^*) - \alpha_{\op, \SuppSet} \cdot \ValDist_{\op, \SuppSet}(s^*)) (s^*-\cost_j)~~ \forall s \in \Op_j(\ValDist)$. This further implies $s^* \in \SuppSet$ and  $\ValDist(s^*) - \alpha_{\op, \SuppSet} \cdot \ValDist_{\op, \SuppSet}(s^*) > 0$.
    
    We discuss the possible $s^*$ by considering different cases:
        \begin{itemize}
        \item \textbf{Suppose $s^* \in [\op_j, \op_{j+1}] \cap \SuppSet $}.  From \Cref{eq:lemma_subset_proof_1}, we know that $\forall s \in \Op_j(\alpha_{\op, \SuppSet} \cdot \ValDist_{\op, \SuppSet})$ 
        \[
            (F(\op_j) - \alpha_{\op, \SuppSet} \cdot \ValDist_{\op, \SuppSet}(\op_j)) (\op_j - \cost_j) \ge (\ValDist(s) - \alpha_{\op, \SuppSet} \cdot \ValDist_{\op, \SuppSet}(s)) (s-\cost_j)~,
        \]
        and $\Op_j(\alpha_{\op, \SuppSet} \cdot \ValDist_{\op, \SuppSet}) = [\op_j, \op_{j+1}] \cap \SuppSet$, therefore, it is impossible that $s^* \in [\op_j, \op_{j+1}]\cap \SuppSet$. 
        \item \textbf{Suppose $s^* \in (\op_{j+1}, + \infty) \cap \SuppSet$}. We would like to prove for all $k \in [j+1, m]$ it has $(F(\op_j) - \alpha_{\op, \SuppSet} \cdot \ValDist_{\op, \SuppSet}(\op_j))(\op_j - \cost_j) > (\ValDist(s) - \alpha_{\op, \SuppSet} \cdot \ValDist_{\op, \SuppSet}(s)) (s-\cost_j), \forall s \in (\op_k, \op_{k+1}] \cap \SuppSet$, therefore, $s^* \notin (\op_{j+1}, + \infty) \cap \SuppSet$. We prove it by induction. 
        \begin{enumerate}
            \item Given $s \in(\op_{j+1}, \op_{j+2}] \cap \SuppSet \subseteq \Op_{j+1}(\alpha_{\op, \SuppSet} \cdot \ValDist_{\op, \SuppSet})$ and $\ValDist(s) - \alpha_{\op, \SuppSet} \cdot \ValDist_{\op, \SuppSet}(s) > 0$,\footnote{\label{ft:assumption for si}We are only interested in $s^*$ that satisfies $F(s^*) - \alpha_{\op, \SuppSet} \cdot \ValDist_{\op, \SuppSet}(s^*) > 0$.} according to \Cref{eq:lemma_subset_proof_1}, 
            \[
                (F(\op_{j+1}) - \alpha_{\op, \SuppSet} \ValDist_{\op, \SuppSet}(\op_{j+1})) (\op_{j+1} - \cost_{j+1}) \ge (\ValDist(s) - \alpha_{\op, \SuppSet} \ValDist_{\op, \SuppSet}(s)) (s-\cost_{j+1})~,
            \]
            and 
            \[
               (F(\op_j) - \alpha_{\op, \SuppSet} \cdot \ValDist_{\op, \SuppSet}(\op_j))(\op_j - \cost_j) \ge (F(\op_{j+1}) - \alpha_{\op, \SuppSet} \cdot \ValDist_{\op, \SuppSet}(\op_{j+1}))(\op_{j+1} - \cost_j)~.
            \]
            Since $s-\cost_{j+1} > \op_{j+1} - \cost_{j+1} > 0$ and $\cost_{j+1} > \cost_{j} \ge 0$, we observe that $\frac{s-\cost_{j+1}}{\op_{j+1} - \cost_{j+1}} > \frac{s-\cost_j}{\op_{j+1} - \cost_{j}}$
                \[
                    \frac{F(\op_{j+1}) - \alpha_{\op, \SuppSet} \cdot \ValDist_{\op, \SuppSet}(\op_{j+1})}{\ValDist(s) - \alpha_{\op, \SuppSet} \cdot \ValDist_{\op, \SuppSet}(s)}  \ge  \frac{s-\cost_{j+1}}{\op_{j+1} - \cost_{j+1}} > \frac{s-\cost_j}{\op_{j+1} - \cost_{j}}~.
                \]
            
            Then we can derive that for all $s \in (\op_{j+1}, \op_{j+2}] \cap \SuppSet$
            \begin{equation}
            \label{eq:lemma_subset_proof_2}
                \begin{aligned}
                    (F(\op_j) - \alpha_{\op, \SuppSet} \cdot \ValDist_{\op, \SuppSet}(\op_j))(\op_j - \cost_j) & \ge (F(\op_{j+1}) - \alpha_{\op, \SuppSet} \cdot \ValDist_{\op, \SuppSet}(\op_{j+1}))(\op_{j+1} - \cost_j) \\
                & > (\ValDist(s) - \alpha_{\op, \SuppSet} \cdot \ValDist_{\op, \SuppSet}(s)) (s-\cost_j)~.
                \end{aligned}
            \end{equation}
            \item Suppose given $k$, it has $(F(\op_j) - \alpha_{\op, \SuppSet} \cdot \ValDist_{\op, \SuppSet}(\op_j))(\op_j - \cost_j) > (\ValDist(s) - \alpha_{\op, \SuppSet} \cdot \ValDist_{\op, \SuppSet}(s)) (s-\cost_j)~, ~~\forall s \in (\op_k, \op_{k+1}] \cap \SuppSet$, we prove it also holds for $k+1$.
            
            Given $s \in(\op_{k+1}, \op_{k+2}] \cap \SuppSet \subseteq \Op_{k+1}(\alpha_{\op, \SuppSet} \cdot \ValDist_{\op, \SuppSet})$ and $\ValDist(s) - \alpha_{\op, \SuppSet} \cdot \ValDist_{\op, \SuppSet}(s) > 0$,\footref{ft:assumption for si} then we have
        \[
        (\ValDist(\op_{k+1}) - \alpha_{\op, \SuppSet} \cdot \ValDist_{\op, \SuppSet}(\op_{k+1}))(\op_{k+1} - \cost_{k+1}) \ge (\ValDist(s) - \alpha_{\op, \SuppSet} \cdot \ValDist_{\op, \SuppSet}(s)) (s-\cost_{k+1})~.
        \]
        Similarly, since $s-\cost_{k+1} > \op_{k+1} - \cost_{k+1} $ and $\cost_{k+1} > \cost_j$, we have $\frac{s-\cost_{k+1}}{\op_{k+1} - \cost_{k+1}} > \frac{s-\cost_j}{\op_{k+1} - \cost_{j}} > 0$. Then 
        \[
        (\ValDist(\op_{k+1}) - \alpha_{\op, \SuppSet} \cdot \ValDist_{\op, \SuppSet}(\op_{k+1}))(\op_{k+1} - \cost_{j}) > (\ValDist(s) - \alpha_{\op, \SuppSet} \cdot \ValDist_{\op, \SuppSet}(s)) (s-\cost_{j})~.
        \]
        Therefore,
        \begin{align*}
            (F(\op_j) - \alpha_{\op, \SuppSet} \cdot \ValDist_{\op, \SuppSet}(\op_j))(\op_j - \cost_j) 
            & > (F(\op_{j+2}) - \alpha_{\op, \SuppSet} \cdot \ValDist_{\op, \SuppSet}(\op_{j+2}))(\op_{j+2} - \cost_j) \\
            & > (\ValDist(s) - \alpha_{\op, \SuppSet} \cdot \ValDist_{\op, \SuppSet}(s)) (s-\cost_j)~.
        \end{align*}
        \end{enumerate}
        For all $s \in (\op_j, + \infty) \cap \SuppSet$, we have 
        \begin{align*}
            (F(\op_j) - \alpha_{\op, \SuppSet} \cdot \ValDist_{\op, \SuppSet}(\op_j))(\op_j - \cost_j) > (\ValDist(s) - \alpha_{\op, \SuppSet} \cdot \ValDist_{\op, \SuppSet}(s)) (s-\cost_j)~.
        \end{align*}
        That implies  we could not find $s^* > \op_{j+1}$ such that $(F(\op_j) - \alpha_{\op, \SuppSet} \cdot \ValDist_{s, S}(s)) (\op_j - \cost_j) < (F(s^*) - \alpha_{\op, \SuppSet} \cdot \ValDist_{\op, \SuppSet}(s^*)) (s^*-\cost_j)$.
        \item \textbf{Suppose $s^* < \op_j$}. We assume $s^* > \cost_j$ without loss, then
        \[
            \ValDist_{q,S}(s)(s-\cost_{j}) <  \ValDist_{q,S}(\op_{j})(\op_{j}- \cost_{j}), ~~ \forall s \in (\cost_j,\op_j) \cap \SuppSet.
        \]
        Given $s \in (\cost_j,\op_j) \cap \SuppSet$, we define a function 
        \[
        g(\alpha) = \frac{(F(\op_{j}) - \alpha_{\op, \SuppSet} \cdot \ValDist_{\op, \SuppSet}(\op_{j})) (\op_j - \cost_j)}{(\ValDist(s) - \alpha_{\op, \SuppSet} \cdot \ValDist_{\op, \SuppSet}(s)) (s - \cost_j)}~.
        \]
        We also assume that $\ValDist(s) - \alpha_{\op, \SuppSet} \cdot \ValDist_{\op, \SuppSet}(s) > 0$,\footref{ft:assumption for si}. 
        Since $\ValDist_{\op, \SuppSet}(s) \ge \ValDist_{\op, \SuppSet}(\op_j) > 0$ the monotonicity of the function can be discussed in different cases for different $s$:
        \begin{itemize}
            \item If $\frac{\ValDist_{\op, \SuppSet}(\op_{j})}{\ValDist_{\op, \SuppSet}(s)} \le \frac{F(\op_{j})}{\ValDist(s)}$, $g(\alpha)$ is non-decreasing when $\alpha$ increases continuously from $0$. Therefore, $g(\alpha) \ge \frac{F(\op_{j}) (\op_j - \cost_j)}{\ValDist(s) (s - \cost_j)} > 1$. This shows for any $\alpha$, $ (\ValDist(\op_{j}) - \alpha_{\op, \SuppSet} \cdot \ValDist_{\op, \SuppSet}(\op_{j})) (\op_j - \cost_j)> (F(s^*) - \alpha_{\op, \SuppSet} \cdot \ValDist_{\op, \SuppSet}(s)) (s - \cost_j)$ holds.
            \item If $\frac{\ValDist_{\op, \SuppSet}(\op_{j})}{\ValDist_{\op, \SuppSet}(s)} > \frac{F(\op_{j})}{\ValDist(s)}$, $g(\alpha)$ is monotonically decreasing from $\frac{F(\op_{j}) (\op_j - \cost_j)}{\ValDist(s) (s - \cost_j)} > 1$.  When we increase $\alpha$  continuously from $0$ till we find $\alpha_{\op,\SuppSet}$ satisfying 
        \[(F(\op_{j}) - \alpha_{\op, \SuppSet} \cdot \ValDist_{\op, \SuppSet}(\op_{j})) (\op_j - \cost_j) = (\ValDist(s) - \alpha_{\op, \SuppSet} \cdot \ValDist_{\op, \SuppSet}(s)) (s - \cost_j)~,\]
        for some $s \in (\cost_j,\op_j) \cap \SuppSet$.
        When this occurs, it triggers the 2(b) since $\op_j(F-\alpha^*\cdot \ValDist_{\op,\SuppSet})$ changes to $s$ and $\alpha$ stops increasing. Therefore, it still satisfies for all $s \in (\cost_j,\op_j) \cap \SuppSet$
        \[(\ValDist(\op_{j}) - \alpha_{\op, \SuppSet} \cdot \ValDist_{\op, \SuppSet}(\op_{j})) (\op_j - \cost_j) \ge (\ValDist(s) - \alpha_{\op, \SuppSet} \cdot \ValDist_{\op, \SuppSet}(s)) (s - \cost_j)~.\]
    \end{itemize}
    Thus, we could not find a $s^* \in (\cost_j, \op_j) \cap \SuppSet$ such that $(\ValDist(s) - \alpha_{\op, \SuppSet} \cdot \ValDist_{\op, \SuppSet}(s)) (s - \cost_j) < (\ValDist(s^*) - \alpha_{\op, \SuppSet} \cdot \ValDist_{\op, \SuppSet}(s^*)) (s^*-\cost_j)$ for any $s \in \Op_j(\ValDist)$.
    
\end{itemize} 
Since we can not find $s^* \in S$ satisfying the condition we give, that indicates that $\Op_j(\ValDist) \subseteq \Op_j(\ValDist - \alpha_{\op, \SuppSet} \cdot  \ValDist_{\op, \SuppSet})$.
 \end{itemize}
Combining $\Op_j(\ValDist) \subseteq \Op_j(\alpha_{\op, \SuppSet} \cdot \ValDist_{\op, \SuppSet})$ and $\Op_j(\ValDist) \subseteq \Op_j(\ValDist - \alpha_{\op, \SuppSet} \cdot \ValDist_{\op, \SuppSet})$, we complete the proof.
\end{proof}

\begin{proof}[Proof of \Cref{lem: extremal market decomposition - decomposition time}]
    There are two conditions that could trigger $\alpha_{\op,\SuppSet}$ stopping increasing. We analyze the number of iteration for those two cases of $\alpha$ separately:
\begin{itemize}
    \item if \textbf{$\alpha =  \alpha_{\text{runout}} \le \alpha_{\text{shift}}$}: Probability runs out at some $v$: $\alpha_{\op, \SuppSet} \cdot \valpmf_{\op, \SuppSet}(\val) = \valpmf(\val)$ for some $v \in \ValSet$, it could happen at most $n = |\ValSet|$ times;
    \item if \textbf{$\alpha =  \alpha_{\text{shift}} < \alpha_{\text{runout}}$}:
    In iteration $k = 1,\dots, K$, we denote the targeted market as $\ValDist^{(k)}$, the extremal market as $\ValDist_{\op^{(k)}, \SuppSet^{(k)}}$ and its corresponding weight as $\alpha_{\op^{(k)}, \SuppSet^{(k)}}$.  Then $\ValDist^{(1)} \leftarrow \ValDist$.
    For a specific $\cost_j$, where its seller surplus in the original market $\max_\val(\val-\cost_{j})\cdot \ValDist(\val) \ne 0$, assume $\op_j $ shifts at iteration $(j_1, j_2,\dots, j_{K_j}, j_{K_{j+1}})$, where the iteration $j_{K_{j+1}}$ is when $\op_j$ becomes $\infty$.
    Since $\op_j = \infty$ also satisfies the first case ( probability at some values runs out), it is not counted in the current situation. Then $K_j$ the number of shifts of $\op_j$ before becoming $\infty$.
    This implies throughout iterations $ k \in \{j_1, j_2,\dots, j_{K_j}\}$, the seller surplus of all $\cost_{j'}$ where $j' \le j$ satisfies $\max_\val(\val-\cost_{j'})\cdot \ValDist^{(k)}(\val) > 0$.
    Due to the monotonicity of $\Op_{j-1}(\ValDist)$, 
    \[
         \max \Op_{j-1}(\ValDist^{(k)}) \le \min \Op_{j}(\ValDist^{(k)}),~~ \forall k \in \{j_1, j_2,\dots, j_{K_j}\}~.
    \]
    Recall that we prove the optimal price set weakly expands in \Cref{lem: extremal market decomposition - optimal price set expansion}, we have 
    \[
    \Op_{j-1}(\ValDist)= \Op_{j-1}(\ValDist^{(1)}) \subseteq \dots \ \Op_{j-1}(\ValDist^{(j_{K_j})})~.
    \]
    Then in all iterations $k \in \{j_1, j_2,\dots, j_{K_j}\}$, $\max \Op_{j-1}(\ValDist) \in \Op_{j-1}(\ValDist^{(k)}) $ always holds and further
    \[
        \max \Op_{j-1}(\ValDist) \le \max \Op_{j-1}(\ValDist^{(j_{K_j})}) \le \min \Op_j(\ValDist^{(j_{K_j})})~.
    \]
    Therefore, the number of iteration $K_j$ is bounded by
    \[
        K_j \le \left|\bigl[\max \Op_{j-1}(\ValDist), \min \Op_{j}(\ValDist)\bigr)\cap \ValSet\right|~.
    \]
    Summing over all $j$, we can derive that the upper bound of the total number of iterations
    \[
        \sum\nolimits_j{K_j} \le \sum\nolimits_j \left|\bigl[\max \Op_{j-1}(\ValDist), \min \Op_{j}(\ValDist)\bigr)\cap \ValSet\right| \le |\ValSet| = n~.
    \]
\end{itemize}
Combining these two situations, we can prove that it takes at most $2n$ iterations to decompose a market $\ValDist$ into extremal markets. 
\end{proof}

\begin{proof}[Proof of \Cref{thm:convex_combination_of_extremal_market}] We prove that the decomposition procedure described in \Cref{alg:extremal market decomposition} is a way to construct $F$ and it preserves both the seller surplus and the buyer surplus.
In iteration $k = 1,\dots, K$, we denote the targeted market by $\ValDist^{(k)}$, the extremal market by $\ValDist_{\op^{(k)}, \SuppSet^{(k)}}$ and its corresponding weight by $\alpha_{\op^{(k)}, \SuppSet^{(k)}}$.  $\ValDist^ 1 \leftarrow \ValDist$. Therefore, we could decompose $\ValDist$ into $\sum_{k=1}^K\alpha_{\op^{(k)}, \SuppSet^{(k)}} \cdot \ValDist_{\op^{(k)}, \SuppSet^{(k)}}$ and $K \le 2n$ since \Cref{lem: extremal market decomposition - decomposition time} have proved that there are at most $2n$ iterations.

\paragraph{The seller surplus and buyer surplus is preserved in the decomposition procedure.} 
Recall \Cref{lem: extremal market decomposition - optimal price set expansion} we have
    \[
        \Op_j(\ValDist^{(k)}) \subseteq \Op_j\left(\alpha_{\op^{(k)}, \SuppSet^{(k)}} \cdot \ValDist_{\op^{(k)}, \SuppSet^{(k)}}\right) \cap \Op_j\left(\ValDist^{(k+1)}\right)~.
    \]
Then given the market $\ValDist$, 
    \[
        \Op_j(\ValDist) \subseteq \bigcap_{k\in[K]} \Op_j\left(\alpha_{\op^{(k)}, \SuppSet^{(k)}} \cdot \ValDist_{\op^{(k)}, \SuppSet^{(k)}}\right) = \bigcap_{k\in[K]} \Op_j\left(\ValDist_{\op^{(k)}, \SuppSet^{(k)}}\right)~,
    \]
Since the procedure ends when the remaining market is zero, for $\val \in \ValSet$, it satisfies 
\[
\ValDist(\val) = \sum\nolimits_{k\in[K]} \alpha_{\op^{(k)}, \SuppSet^{(k)}} \cdot \ValDist_{\op^{(k)}, \SuppSet^{(k)}}(\val)~, 
\]
and 
\[
\valpmf(\val) = \sum\nolimits_{k\in[K]} \alpha_{\op^{(k)}, \SuppSet^{(k)}} \cdot \valpmf_{\op^{(k)}, \SuppSet^{(k)}}(\val)~. 
\]
If we choose optimal price $\mu_j$ for cost $\cost_j$ in $\ValDist$, then $\mu_j> \cost_j$ and $\mu_j\in \Op_j(\ValDist)$. Thus, 
\[
    \mu_j \in \Op_j(\ValDist_{\op^{(k)}, \SuppSet^{(k)}}), ~ k=1,2,\dots, K~.
\] 
This indicates that we can choose $\mu_j$ as the optimal price for each extremal market $\ValDist_{\op^{(k)}, \SuppSet^{(k)}}, ~ k=1,2,\dots, K$ to get the seller surplus $\SellerSurplus(\ValDist_{\op^{(k)}, \SuppSet^{(k)}}) = \sum\nolimits_{j}\costpmf(\cost_j) \cdot (\mu_j - \cost_j)\ValDist_{\op^{(k)}, \SuppSet^{(k)}}(\mu_j)$, where $\costpmf(\cost_j)$ is the probability of $\cost_j$. 
Then the seller surplus of $\ValDist$, denoted as $\SellerSurplus(\ValDist)$, is
\begin{align*}
    \SellerSurplus(\ValDist) 
    & = \sum\nolimits_{j}\costpmf(\cost_j) \cdot (\mu_j - \cost_j)\cdot 
    \ValDist(\mu_j)\\ 
    & = \sum\nolimits_{j}\costpmf(\cost_j) \cdot (\mu_j - \cost_j)\sum\nolimits_{k\in[K]} \alpha_{\op^{(k)}, \SuppSet^{(k)}} \cdot \ValDist_{\op^{(k)}, \SuppSet^{(k)}}(\mu_j)\\
    & = \sum\nolimits_{k\in[K]} \alpha_{\op^{(k)}, \SuppSet^{(k)}} \sum\nolimits_{j}\costpmf(\cost_j) \cdot (\mu_j - \cost_j) \cdot \ValDist_{\op^{(k)}, \SuppSet^{(k)}}(\mu_j) \\
    & = \sum\nolimits_{k\in[K]} \alpha_{\op^{(k)}, \SuppSet^{(k)}}\cdot \SellerSurplus(\ValDist_{\op^{(k)}, \SuppSet^{(k)}})~.
\end{align*} 

Similarly, the buyer surplus of $\ValDist$, denoted as $\BuyerSurplus(\ValDist)$, is
\begin{align*}
    \BuyerSurplus(\ValDist) 
    & = \sum\nolimits_j \costpmf(\cost_j) \sum\nolimits_{\val_i > \mu_j}(\val_i - \mu_j)\cdot \valpmf(\val_i) \\
    & = \sum\nolimits_j\costpmf(\cost_j) \sum\nolimits_{\val_i > \mu_j}(\val_i - \mu_j)\sum\nolimits_{k\in[K]} \alpha_{\op^{(k)}, \SuppSet^{(k)}} \cdot \valpmf_{\op^{(k)},\SuppSet^{(k)}}(\val_i)\\ 
    & = \sum\nolimits_{k\in[K]} \alpha_{\op^{(k)}, \SuppSet^{(k)}} \sum\nolimits_j\costpmf(\cost_j)\sum\nolimits_{\val_i > \mu_j}(\val_i - \mu_j) \cdot \valpmf_{\op^{(k)},\SuppSet^{(k)}}(\val_i) \\
    & = \sum\nolimits_{k\in[K]} \alpha_{\op^{(k)}, \SuppSet^{(k)}}\cdot \BuyerSurplus(\ValDist_{\op^{(k)}, \SuppSet^{(k)}})~.
\end{align*}
We conclude that the convex combination of extremal markets $\ValDist =\sum_{k\in[K]} \alpha_{\op^{(k)}, \SuppSet^{(k)}} \ValDist^{(k)}_{\op, \SuppSet}$ preserve both the seller surplus and the buyer surplus.
\end{proof}

\section{Missing Proofs in \Cref{sec:fragment}}
\label{apx: fragment proof}
\begin{proof}[Proof of \Cref{lem:locality of fragment}]
    Given a $\Frag_{\op, \SuppSet, j}$ obtained by an extremal market $\ValDist_{\op, \SuppSet}$ and the corresponding cost $\cost_j$, define $\esw_{\op, \SuppSet, j} := \ValDist_{\op,\SuppSet}(\op_{j})$ and $\FragSuppSet := \SuppSet \cap [\op_j, \op_{j+1})$. Assume $\esw_{\op, \SuppSet, j} > 0$.\footnote{When $\esw_{\op, \SuppSet, j} = 0$, it indicates that $\ValDist_{\op,\SuppSet}(\op_{j}) = 0$, thus, $\FragSuppSet = \emptyset$. } It satisfies $(\op_j - \cost_j)\cdot \ValDist_{\op, \SuppSet}(\op_j) = (\val - \cost_j)\cdot \ValDist_{\op, \SuppSet}(\val) \; \forall \val \in \FragSuppSet \cup \{\op_{j+1}\}$, according to \Cref{lem:optimal_price_for_extremal_market}. 

    For any $\val \in \FragSuppSet$, we find $\val^{+}$ such that $\val^{+} = \min \{\val' \in \FragSuppSet \cup \{\op_{j+1}\}: \val' > \val \}$, $\valpmf_{\op, \SuppSet}(\val) = \ValDist_{\op, \SuppSet}(\val) - \ValDist_{\op, \SuppSet}(\val^{+})$.
    \begin{align*}
        \frag_{\op, \SuppSet}(\val) 
        = \frac{\valpmf_{\op, \SuppSet}(\val)}{\esw_{\op, \SuppSet, j}} 
        = \frac{\ValDist_{\op, \SuppSet}(\val)}{\ValDist_{\op,\SuppSet}(\op_{j})} - \frac{\ValDist_{\op, \SuppSet}(\val^{+})}{\ValDist_{\op,\SuppSet}(\op_{j})} 
        = \frac{\op_j - \cost_j}{\val - \cost_j} - \frac{\op_j - \cost_j}{\val^{+} - \cost_j}~.
    \end{align*}
    Thus,
    \begin{align*}
        \Frag_{\op, \SuppSet, j}(\val) 
        = \sum\nolimits_{\val' \ge \val} \frag_{\op, \SuppSet}(\val') = \frac{\op_j - \cost_j}{\val - \cost_j} - \frac{\op_j - \cost_j}{\op_{j+1} - \cost_j}~.
    \end{align*}
    
    We observe that $\Frag_{\op, \SuppSet, j}(\val)$  and $\frag_{\op, \SuppSet}(\val)$ are uniquely determined by $\op_j$, $\op_{j+1}$ $\cost_j$, and $\val$. Therefore, the fragment $\Frag_{\op, \SuppSet, j}$ is determined by the triple $(\op_j, \op_{j+1}, \FragSuppSet)$.
\end{proof}

\begin{proof}[Proof of \Cref{lem: fragment decompositions}] \label{proof: lem - fragment decompositions}
\textbf{Seller surplus induced by $\ba$.} Given an extremal market $\ValDist_{\op, \SuppSet}$ in an extremal market decomposition $\ba = (\alpha_{\op, \SuppSet})$ and a specific cost $\cost_j$, the contribution to the seller surplus is $\alpha_{\op, \SuppSet} \cdot \costpmf(\cost_j) \cdot (\op_j - \cost_j) \cdot \ValDist(\op_j)$. According to the fragment decomposition of the extremal market defined in \Cref{def:fragment-decomposition}, $\esw_{j, \ell, r, \FragSuppSet}^{\op, \SuppSet} = \ValDist(\op_j)$ when $\val_\ell = \op_j, \val_r = \op_{j+1}$ and $T = \SuppSet \cap [\op_j, \op_{j+1})$. Otherwise, 0 Then we can charge the seller surplus induced by the pair $(\ValDist_{\op, \SuppSet}, \cost_j)$ to the specific fragment $\Frag_{j, \ell, r, \FragSuppSet}$ with the number $\alpha_{\op, \SuppSet} \cdot \costpmf(\cost_j) \cdot (\op_j - \cost_j) \cdot \esw_{j, \ell, r, \FragSuppSet}^{\op, \SuppSet}$, where $\val_\ell = \op_j, \val_r = \op_{j+1}$ and $T = \SuppSet \cap [\op_j, \op_{j+1})$. 

Therefore, we can define the seller surplus contribution of the quadruple $(j, \ell, r, \FragSuppSet)$ induced by each extremal market $\ValDist_{\op, \SuppSet}$ as $\alpha_{\op, \SuppSet} \cdot \costpmf(\cost_j) \cdot (\val_\ell - \cost_j) \cdot \esw_{j, \ell, r, \FragSuppSet}^{\op, \SuppSet}$. Note that the contribution is $0$ when either $\val_\ell \ne \op_j$ or $\val_r \ne \op_{j+1}$ or $\FragSuppSet \ne \SuppSet \cap [\op_j, \op_{j+1})$, since $\esw_{j, \ell, r, \FragSuppSet}^{\op, \SuppSet} = 0$.

The overall seller surplus of $\ba$, denoted by $\SellerSurplus(\ba)$, is 
\begin{equation}
\label{eq: seller surplus of fragment decomposition}
   \begin{aligned}
    \SellerSurplus(\ba) & = \sum\nolimits_{\op, \SuppSet}\sum\nolimits_j  \alpha_{\op, \SuppSet} \cdot  \costpmf(\cost_j) \cdot (\op_{j} - \cost_j)\cdot \ValDist_{\op, \SuppSet}(\op_j) \\
    & = \sum\nolimits_{\op, \SuppSet} \sum\nolimits_{j, \ell, r, \FragSuppSet} \alpha_{\op, \SuppSet}  \cdot \costpmf(\cost_j) \cdot (\op_j - \cost_j) \cdot \esw_{j, \ell, r, \FragSuppSet}^{\op, \SuppSet} \\
    & = \sum\nolimits_{j, \ell, r, \FragSuppSet} \costpmf(\cost_j) \cdot (\val_\ell - \cost_j)\sum\nolimits_{\op, \SuppSet} \alpha_{\op, \SuppSet} \cdot \esw_{j, \ell, r, \FragSuppSet}^{\op, \SuppSet} \\
    & = \sum\nolimits_{j, \ell, r, \FragSuppSet} \costpmf(\cost_j)\cdot (\val_\ell - \cost_j) \cdot \sw_{j, \ell, r, \FragSuppSet} ~.
    \end{aligned}
\end{equation}
We can prove that the seller surplus $\SellerSurplus(\ba)$ is linear to $\bw$. 

\textbf{Buyer surplus induced by $\ba$.}Since the buyer surplus of the extremal market decomposition $\ba$, denoted by $\BuyerSurplus(\ba)$, is the difference between social welfare $\SocialWelfare(\ba)$ and seller surplus $\SellerSurplus(\ba)$, we prove that the maximum/minimum buyer surplus is linear to $\bw$ by proving the maximum/minimum social welfare is linear to $\bw$.
For social welfare, if we choose optimal price $\mu_{\op, \SuppSet, j}$ for cost $\cost_j$ in the extremal market $\ValDist_{\op, \SuppSet}$, where $\mu_{\op, \SuppSet, j} \in \Op_j(\ValDist_{\op, \SuppSet}) =  \SuppSet \cap [\op_j, \op_{j+1}]$, then
\begin{align*}
  \SocialWelfare(\ba) 
  & = \sum\nolimits_{\op, \SuppSet} \alpha_{\op, \SuppSet}  \sum\nolimits_j \costpmf(\cost_j) \sum\nolimits_{\val \ge \mu_{\op, \SuppSet, j}} (\val - \cost_j) \cdot \valpmf_{\op, \SuppSet}(\val)~.
\end{align*}
We get maximum social welfare when $\mu_{\op, \SuppSet, j} = \op_j$, which is the minimum optimal price, for each cost in all extremal markets. The maximum social welfare is equal to
\begin{align*}
     \MaxSocialWelfare(\ba)  = \sum\nolimits_{\op, \SuppSet} \alpha_{\op, \SuppSet}  \sum\nolimits_j \costpmf(\cost_j) \sum\nolimits_{\val \ge \op_j} (\val - \cost_j) \cdot \valpmf_{\op, \SuppSet}(\val)~.
\end{align*}
Fix an extremal market $\ValDist_{\op, \SuppSet}$ in extremal market decomposition $\ba$ and a specific cost $\cost_j$, then the contribution to the maximum social welfare is 
\begin{align*}
    & \alpha_{\op, \SuppSet} \cdot \costpmf(\cost_j) \sum\nolimits_{\val \ge \op_j} (\val - \cost_j) \cdot \valpmf_{\op, \SuppSet}(\val) \\
    = ~& \sum\nolimits_{j'\ge j} \alpha_{\op, \SuppSet} \cdot \costpmf(\cost_j)  \sum\nolimits_{\val \in[\op_{j'}, \op_{j'+1})} (\val - \cost_j) \cdot \valpmf_{\op, \SuppSet}(\val)\\
    = ~ & \sum\nolimits_{j'\ge j} \alpha_{\op, \SuppSet} \cdot \costpmf(\cost_j) \sum\nolimits_{\val \in \FragSuppSet_{j'}} (\val - \cost_j) \cdot \esw_{j', \ell_{j'}, r_{j'}, \FragSuppSet_{j'}}^{\op, \SuppSet} \cdot \frag_{j', \ell_{j'}, r_{j'}, \FragSuppSet_{j'}}(\val)~,
\end{align*}
where $\val_{\ell_{j'}} = \op_{j'}$, $\val_{r_{j'}} = \op_{j'+1}$, and $\FragSuppSet_{j'} = \SuppSet \cap [\op_{j'}, \op_{j'+1})$ for each $j'$.
Then we charge the social welfare to each fragment $\Frag_{j', \ell_{j'}, r_{j'}, \FragSuppSet_{j'}}$ where $j' \ge j$, with the amount $\alpha_{\op, \SuppSet} \cdot \costpmf(\cost_j) \sum\nolimits_{\val \in \FragSuppSet_{j'}} (\val - \cost_j) \cdot \esw_{j', \ell_{j'}, r_{j'}, \FragSuppSet_{j'}}^{\op, \SuppSet} \cdot \frag_{j', \ell_{j'}, r_{j'}, \FragSuppSet_{j'}}(\val)$. This can be expanded for other $(j', \ell, r, \FragSuppSet)$ where $\val_\ell \ne \op_{j'}$ or $\val_r \ne \op_{j'+1}$ or $\FragSuppSet \ne \SuppSet \cap [\op_{j'}, \op_{j'+1})$ since $w_{j', \ell, r, \FragSuppSet}^{\op, \SuppSet} = 0$ and the social welfare charged is $0$. Thus,
\begin{align*}
    \alpha_{\op, \SuppSet} \cdot \costpmf(\cost_j) \sum\nolimits_{\val \ge \op_j} (\val - \cost_j) \cdot \valpmf_{\op, \SuppSet}(\val) = \sum\nolimits_{j'\ge j} \alpha_{\op, \SuppSet} \cdot \costpmf(\cost_j) \sum\nolimits_{\ell, r, \FragSuppSet} \sum\nolimits_{\val \in \FragSuppSet} (\val - \cost_j) \cdot \esw_{j', \ell, r, \FragSuppSet}^{\op, \SuppSet} \cdot \frag_{j', \ell, r, \FragSuppSet}(\val)~,
\end{align*}
Now we consider the social welfare charged to each fragment $\Frag_{j, \ell,r, \FragSuppSet}$. Given a quadruple $(j, \ell,r, \FragSuppSet)$, we sum over the social welfare induced by all extremal markets $\ValDist_{\op, \SuppSet}$ and $\cost_{j'}, j' \le j$, the value is $\sum_{\op, \SuppSet} \sum_{j' \le j}\alpha_{\op, \SuppSet} \cdot \costpmf(\cost_{j'}) \sum_{\val \in \FragSuppSet}(\val - \cost_{j'}) \cdot \frag_{j, \ell, r, \FragSuppSet}(\val) \cdot \esw_{j, \ell, r, \FragSuppSet}^{\op, \SuppSet}$. 
Therefore, we can rewrite the maximum social welfare as 
\begin{equation}
\label{eq: maximum social welfare of fragment decomposition}
    \begin{aligned}
    \MaxSocialWelfare(\ba) 
      & = \sum\nolimits_{j, \ell,r, \FragSuppSet} \sum\nolimits_{\op, \SuppSet} \sum\nolimits_{j' \le  j}\alpha_{\op, \SuppSet} \cdot \costpmf(\cost_{j'}) \sum\nolimits_{\val \in \FragSuppSet}(\val - \cost_{j'}) \cdot \frag_{j, \ell, r, \FragSuppSet}(\val) \cdot \esw_{j, \ell, r, \FragSuppSet}^{\op, \SuppSet} \\
      & = \sum\nolimits_{j, \ell,r, \FragSuppSet} \sum\nolimits_{j' \le  j}\sum\nolimits_{\val \in \FragSuppSet}\costpmf(\cost_{j'}) \cdot (\val - \cost_{j'}) \cdot \frag_{j, \ell, r, \FragSuppSet}(\val) \sum\nolimits_{\op, \SuppSet} \alpha_{\op, \SuppSet} \cdot \esw_{j, \ell, r, \FragSuppSet}^{\op, \SuppSet}\\
      & = \sum\nolimits_{j, \ell,r, \FragSuppSet} \sum\nolimits_{j' \le  j}\sum\nolimits_{\val \in \FragSuppSet}\costpmf(\cost_{j'}) \cdot (\val - \cost_{j'}) \cdot \frag_{j, \ell, r, \FragSuppSet}(\val) \cdot \sw_{j, \ell, r, \FragSuppSet}~.
    \end{aligned}
\end{equation}
Thus, $\MaxSocialWelfare(\ba)$ is a linear combination of $\bw$. Since $\MaxBuyerSurplus(\ba) = \MaxSocialWelfare(\ba) - \SellerSurplus(\ba)$, $\MaxBuyerSurplus(\ba)$ is also linear to $\bw$.

Similarly, we get minimum social welfare when $\mu_{\op, \SuppSet, j} = \op_{j+1}$ for each cost in all extremal markets: 
\begin{align*}
     \MinSocialWelfare(\ba)  = \sum\nolimits_{\op, \SuppSet} \alpha_{\op, \SuppSet}  \sum\nolimits_j \costpmf(\cost_j) \sum\nolimits_{\val \ge \op_{j+1}} (\val - \cost_j) \cdot \valpmf_{\op, \SuppSet}(\val)~.
\end{align*}
Fix an extremal market $\ValDist_{\op, \SuppSet}$ in extremal market decomposition $\ba = (\alpha_{\op, \SuppSet})$ and a specific cost $\cost_j$, the contribution to the minimum social welfare is 
\begin{align*}
    & \alpha_{\op, \SuppSet} \cdot \costpmf(\cost_j) \sum\nolimits_{\val \ge \op_{j+1}} (\val - \cost_j) \cdot \valpmf_{\op, \SuppSet}(\val) \\
    = ~& \sum\nolimits_{j' \ge j+1} \alpha_{\op, \SuppSet} \cdot \costpmf(\cost_j) \sum\nolimits_{\val \in[\op_{j'}, \op_{j'+1})} (\val - \cost_j) \cdot \valpmf_{\op, \SuppSet}(\val)\\
    = ~ & \sum\nolimits_{j'\ge j+1} \alpha_{\op, \SuppSet} \cdot \costpmf(\cost_j) \sum\nolimits_{\val \in \FragSuppSet_{j'}} (\val - \cost_j) \cdot \esw_{j', \ell_{j'}, r_{j'}, \FragSuppSet_{j'}}^{\op, \SuppSet} \cdot \frag_{j', \ell_{j'}, r_{j'}, \FragSuppSet_{j'}}(\val)~,
\end{align*}
where $\val_{\ell_{j'}} = \op_{j'}$, $\val_{r_{j'}} = \op_{j'+1}$, and $\FragSuppSet_{j'} = \SuppSet \cap [\op_{j'}, \op_{j'+1})$ for each $j'$, $\ell_{j'}$.
Then we charge the social welfare to each fragment $\Frag_{j', \ell_{j'}, r_{j'}, \FragSuppSet_{j'}}$ where $j' \ge j+1$, with the amount $\alpha_{\op, \SuppSet} \cdot \costpmf(\cost_j) \sum\nolimits_{\val \in \FragSuppSet_{j'}} (\val - \cost_j) \cdot \esw_{j', \ell_{j'}, r_{j'}, \FragSuppSet_{j'}}^{\op, \SuppSet} \cdot \frag_{j', \ell_{j'}, r_{j'}, \FragSuppSet_{j'}}(\val)$. This can be expanded for other $(j', \ell, r, \FragSuppSet)$ where $\ell \ne \op_{j'}$ or $r \ne \op_{j'+1}$ or $\FragSuppSet \ne \SuppSet \cap [\op_{j'}, \op_{j'+1})$ since $w_{j', \ell, r, \FragSuppSet}^{\op, \SuppSet} = 0$ and the social welfare charged is $0$. Thus,
\begin{align*}
    \alpha_{\op, \SuppSet} \cdot \costpmf(\cost_j) \sum\nolimits_{\val \ge \op_{j+1}} (\val - \cost_j) \cdot \valpmf_{\op, \SuppSet}(\val) = \sum\nolimits_{j'\ge j+1} \alpha_{\op, \SuppSet} \cdot \costpmf(\cost_j) \sum\nolimits_{\ell, r, \FragSuppSet} \sum\nolimits_{\val \in \FragSuppSet} (\val - \cost_j) \cdot \esw_{j', \ell, r, \FragSuppSet}^{\op, \SuppSet} \cdot \frag_{j', \ell, r, \FragSuppSet}(\val)~,
\end{align*}

When we consider the overall minimum social welfare charged to each fragment $\Frag_{j, \ell, r, \FragSuppSet}$, we sum over the social welfare charged from all extremal markets and $\cost_{j'}, j' < j$ (strictly less than), the value is $\sum\nolimits_{\op, \SuppSet} \sum\nolimits_{j' < j}\alpha_{\op, \SuppSet} \cdot \costpmf(\cost_{j'}) \sum\nolimits_{\val \in \FragSuppSet}(\val - \cost_{j'}) \cdot \frag_{j,\ell, r, \FragSuppSet}(\val) \cdot \esw_{j, \ell, r, \FragSuppSet}^{\op, \SuppSet}$.
Therefore, we can rewrite the minimum social welfare as 
\begin{equation}
\label{eq: minimum social welfare of fragment decomposition}
    \begin{aligned}
        \MinSocialWelfare(\ba) 
      & = \sum\nolimits_{j, \ell, r, \FragSuppSet} \sum\nolimits_{\op, \SuppSet} \sum\nolimits_{j' <j}\alpha_{\op, \SuppSet} \cdot \costpmf(\cost_{j'}) \sum\nolimits_{\val \in \FragSuppSet}(\val - \cost_{j'}) \cdot \frag_{j,\ell, r, \FragSuppSet}(\val) \cdot \esw_{j, \ell, r, \FragSuppSet}^{\op, \SuppSet} \\
      & = \sum\nolimits_{j, \ell,r, \FragSuppSet} \sum\nolimits_{j' < j}\sum\nolimits_{\val \in \FragSuppSet}\costpmf(\cost_{j'}) \cdot (\val - \cost_{j'}) \cdot \frag_{j,\ell, r, \FragSuppSet}(\val) \sum\nolimits_{\op, \SuppSet} \alpha_{\op, \SuppSet} \cdot \esw_{j, \ell, r, \FragSuppSet}^{\op, \SuppSet}\\
      & = \sum\nolimits_{j, \ell,r, \FragSuppSet} \sum\nolimits_{j' < j}\sum\nolimits_{\val \in \FragSuppSet}\costpmf(\cost_{j'}) \cdot (\val - \cost_{j'}) \cdot \frag_{j,\ell, r, \FragSuppSet}(\val) \cdot \sw_{j, \ell, r, \FragSuppSet}~.
    \end{aligned}
\end{equation}
Thus, $\MinSocialWelfare(\ba)$ is a linear combination of $\bw$. Since $\MinBuyerSurplus(\ba) = \MinSocialWelfare(\ba) - \SellerSurplus(\ba)$, $\MinBuyerSurplus(\ba)$ is also linear to $\bw$.
\end{proof}
\begin{proof}[Proof of \Cref{lem: number of undominated fragments}] According to the definition of undominated fragment $\Frag_{j, \ell, r}$, where $j \in [m]$, $\ell \in [n]$ and $r\in [n+1] $, there are at most $O(mn^2)$ possible undominated fragments. 
\end{proof}
\begin{proof}[Proof of \Cref{lem: domination relationship}]
    Given a fragment $\Frag_{j, \ell, r, \FragSuppSet}$, recall the tail function of the fragment in \Cref{lem:locality of fragment},
    \[
       \Frag_{j, \ell, r, \FragSuppSet}(\val) = \frac{\val_\ell - \cost_j}{\val - \cost_j} - \frac{\val_\ell - \cost_j}{\val_r - \cost_j}, \quad \forall \val \in \FragSuppSet~. 
    \]
    
    The undominated fragment $\Frag_{j, \ell, r}$ is also a fragment with the same $j, \ell, r$, while its support set is $\ValSet \cap [\val_l, \val_r)$. Thus,
    \[
       \Frag_{j, \ell, r}(\val) = \frac{\val_\ell - \cost_j}{\val - \cost_j}  - \frac{\val_\ell - \cost_j}{\val_r - \cost_j}, \quad \forall \val \in \ValSet \cap [\val_l, \val_r)~. 
    \]
    Since $\FragSuppSet \subseteq \ValSet \cap [\val_\ell, \val_r)$, given $\val \in \ValSet \cap [\val_l, \val_r)$ it has two cases:
    \begin{itemize}
        \item $\val \in \FragSuppSet$, then $\Frag_{j, \ell, r, \FragSuppSet}(\val) = \Frag_{j, \ell, r}(\val) $
        \item $\val \notin \FragSuppSet$, then $\Frag_{j, \ell, r, \FragSuppSet}(\val) = \Frag_{j, \ell, r, \FragSuppSet}(\val^{+}) = \Frag_{j, \ell, r}(\val^{+}) < \Frag_{j, \ell, r}(\val) $, where $\val^{+} = \min \{\val' \in \FragSuppSet \cap [\val_\ell, \val_r): \val' > \val \}$. \footnote{If $\val^{+}$ does not exist, it indicates that $\Frag_{j, \ell, r, \FragSuppSet}(\val) = 0 < \Frag_{j, \ell, r}(\val)$.} 
    \end{itemize}
Therefore, we can prove that $\Frag_{j, \ell, r}$ first-order stochastically dominates $\Frag_{j, \ell, r, \FragSuppSet}$.
\end{proof}
\begin{proof}[Proof of \Cref{lem: reduced fragment decompositions - proper domination}] Recall that the fragment $\Frag_{j,\ell,r}$ has support on $ \{\val_i:\, i\in [\ell, r)\}$ and the induced market $\ValDist^{\rsdz}_{j,\ell,r}$'s support is the subset of $ \{\val_i:\, i\in [\ell, r)\}$. Moreover, for $\val \ge \val_r$, both sides are $0$ and for $\val < \val_\ell$ both sides equal their total masses. Therefore, it suffices to prove the inequality for $\val=\val_i$ with $i \in [\ell, r)$.
    Then for any $i \in [\ell: r)$, 
    \begin{align*}
             \ValDist^{\rsdz}_{j, \ell, r}(\val_i) 
             & = \sum\nolimits_{i' \in [i: r)} \rsdz_{j, \ell, r, i'} \\
             & = \sum\nolimits_{i' \in [i: r)} \sum_{ \FragSuppSet \subseteq \ValSet \cap [\val_\ell, \val_r)} \sw_{j, \ell, r, \FragSuppSet} \cdot \frag_{j, \ell, r, \FragSuppSet}(\val_{i'})  \\
             & = \sum\nolimits_{ \FragSuppSet \subseteq \ValSet \cap [\val_\ell, \val_r)} \sw_{j, \ell, r, \FragSuppSet} \sum\nolimits_{i' \in [i: r)} \frag_{j, \ell, r, \FragSuppSet}(\val_{i'}) \\
             & = \sum\nolimits_{ \FragSuppSet \subseteq \ValSet \cap [\val_\ell, \val_r)} \sw_{j, \ell, r, \FragSuppSet} \cdot \Frag_{j, \ell, r, \FragSuppSet}(\val_{i})~.
             \end{align*}
     We invoke the domination relationship between the fragment $\Frag_{j, \ell, r, \FragSuppSet}$ and the reduced fragment $\Frag_{j, \ell, r}$ in \Cref{lem: domination relationship} to get 
    \[
    \sum\nolimits_{ \FragSuppSet \subseteq \ValSet \cap [\val_\ell, \val_r)} \sw_{j, \ell, r, \FragSuppSet} \cdot \Frag_{j, \ell, r, \FragSuppSet}(\val_{i}) \le 
    \sum\nolimits_{ \FragSuppSet \subseteq \ValSet \cap [\val_\ell, \val_r)} \sw_{j, \ell, r, \FragSuppSet} \cdot \Frag_{j, \ell, r}(\val_{i}) = \rsdw_{j, \ell, r} \cdot \Frag_{j,\ell,r}(\val_{i})~.
    \]
     Thus,  for any $i \in [\ell: r)$ , it has $\ValDist^{\rsdz}_{j, \ell, r}(\val_i) \le \rsdw_{j, \ell, r} \cdot \Frag_{j,\ell,r}(\val_{i})$.  
Therefore, the induced $\ValDist^{z}_{j, \ell, r}$ is weakly dominated by the fragment $\Frag_{j,\ell,r}$ scaled by $\rsdw_{j, \ell, r}$.

According to the definition of fragment $\Frag_{j, \ell, r, \FragSuppSet}$ (in \Cref{defn:fragment}), the point mass $\frag_{j, \ell, r, \FragSuppSet}(
\val) = 0$ when $\val \notin [\val_\ell, \val_r)$. Then
    \begin{align*}
        \sum\nolimits_{i \in [n]}\rsdz_{j, \ell, r, i} 
        & = \sum\nolimits_{i \in [n]}\sum\nolimits_{ \FragSuppSet \subseteq \ValSet \cap [\val_\ell, \val_r)} \sw_{j, \ell, r, \FragSuppSet} \cdot \frag_{j, \ell, r, \FragSuppSet}(\val_i) \\
        & = \sum\nolimits_{i \in [\ell: r)}\sum\nolimits_{ \FragSuppSet \subseteq \ValSet \cap [\val_\ell, \val_r)} \sw_{j, \ell, r, \FragSuppSet} \cdot \frag_{j, \ell, r, \FragSuppSet}(\val_i)\\
        & = \sum\nolimits_{ \FragSuppSet \subseteq \ValSet \cap [\val_\ell, \val_r)} \sw_{j, \ell, r, \FragSuppSet} \sum\nolimits_{i \in \FragSuppSet} \frag_{j, \ell, r, \FragSuppSet}(\val_i)\\
        & = \sum\nolimits_{ \FragSuppSet \subseteq \ValSet \cap [\val_\ell, \val_r)} \sw_{j, \ell, r, \FragSuppSet} \cdot (1- \frac{\val_\ell - \cost_j}{\val_r - \cost_j})\\
        & = \frac{\val_r - \val_\ell}{\val_r - \cost_j}\cdot \rsdw_{j, \ell, r}~.
    \end{align*} 
We complete the proof.
\end{proof}

\begin{proof}[Proof of \Cref{lem: surplus of reduced fragment decompositions}]
    Fix an extremal market decompoisition $\ba$. Let $\bw = (\sw_{j, \ell, r, \FragSuppSet})$ be its corresponding fragment decomposition and $(\bx, \bz)$, where $\bx = (\rsdw_{j,\ell, r})$ and $\bz=(\rsdz_{j, \ell, r, i})$, be the reduced fragment decomposition.
    According to \Cref{lem: fragment decompositions} and \Cref{eq: seller surplus of fragment decomposition}, the seller surplus of fragment decomposition induced by $\ba$ is linear to $\bw$. Specifically, the seller surplus $\SellerSurplus(\ba) = \sum_{j, \ell, r, \FragSuppSet} \costpmf(\cost_j) \cdot (\val_\ell - \cost_j) \cdot \sw_{j, \ell, r, \FragSuppSet}$. Since the fragment $\Frag_{j, \ell, r, \FragSuppSet}$ is defined only for $\FragSuppSet \subseteq \ValSet \cap [\val_\ell, \val_r)$ by \Cref{remark: Succinct fragments}; otherwise, $\Frag_{j, \ell, r, \FragSuppSet}$ is not defined and $\sw_{j, \ell, r, \FragSuppSet}$ does not exist, the seller surplus $\SellerSurplus(\ba)$ equals
    \begin{equation}
    \label{eq: seller surplus charged by x}
    \begin{aligned}
        \SellerSurplus(\ba) & = \sum\nolimits_{j, \ell, r, \FragSuppSet} \costpmf(\cost_j) \cdot (\val_\ell - \cost_j) \cdot \sw_{j, \ell, r, \FragSuppSet} \\
        & = \sum\nolimits_{j, \ell, r} \costpmf(\cost_j) \cdot (\val_\ell - \cost_j) \sum\nolimits_{{ \FragSuppSet \subseteq \ValSet \cap [\val_\ell, \val_r)}} \sw_{j, \ell, r, \FragSuppSet}\\
        & = \sum\nolimits_{j, \ell, r} \costpmf(\cost_j) \cdot (\val_\ell - \cost_j) \cdot \rsdw_{j, \ell, r}~.
    \end{aligned}    
    \end{equation}
    We can prove that $\SellerSurplus(\alpha)$ is linear to $\bx$.
    
    For the maximum social welfare, according to \Cref{lem: fragment decompositions} and \Cref{eq: maximum social welfare of fragment decomposition}, the maximum social welfare induced by $\ba$ is linear to $\bw$, and specifically,
    \begin{equation}
    \label{eq: maximum social welfare charged by z}
    \begin{aligned}
             \MaxSocialWelfare(\ba) & = \sum\nolimits_{j, \ell, r, \FragSuppSet} \sum\nolimits_{j' \le j}\sum\nolimits_{\val \in \FragSuppSet} \costpmf(\cost_{j'}) \cdot (\val - \cost_{j'}) \cdot \frag_{j, \ell, r, \FragSuppSet}(\val) \cdot \sw_{j, \ell, r, \FragSuppSet}\\
             & = \sum\nolimits_{j, \ell, r} \sum\nolimits_{\FragSuppSet \subseteq \ValSet \cap [\val_\ell, \val_r)} \sum\nolimits_{j' \le j}\sum\nolimits_{i \in [\ell: r)} \costpmf(\cost_{j'}) \cdot (\val_i - \cost_{j'})  \frag_{j,\ell, r, \FragSuppSet}(\val_i) \cdot \sw_{j, \ell, r, \FragSuppSet} \\
             & = \sum\nolimits_{j, \ell, r} \sum\nolimits_{j' \le j}\sum\nolimits_{i \in [\ell: r)} \costpmf(\cost_{j'}) \cdot  (\val_i - \cost_{j'}) \sum\nolimits_{{\FragSuppSet \subseteq \ValSet \cap [\val_\ell, \val_r)}}  \frag_{j,\ell, r, \FragSuppSet}(\val_i) \cdot \sw_{j, \ell, r, \FragSuppSet} \\
              & = \sum\nolimits_{j, \ell,r} \sum\nolimits_{j' \le j}\sum\nolimits_{i \in [\ell: r)} \costpmf(\cost_{j'}) \cdot (\val_i - \cost_{j'}) \cdot \rsdz_{j, \ell, r, i}~.
    \end{aligned}
    \end{equation}
    In the second step, fix an index triple $(j, \ell, r)$, the fragment $\Frag_{j, \ell, r, \FragSuppSet}$ is defined only for $\FragSuppSet \subseteq \ValSet \cap [\val_\ell, \val_r)$ by \Cref{remark: Succinct fragments}.  Given a support set $\FragSuppSet \subseteq \ValSet \cap [\val_\ell, \val_r)$ , when $i \in [\ell: r)$ but $\val_i \notin \FragSuppSet$, the mass at $\val_i$ is zero. Thus, we can expand $\val \in \FragSuppSet$ to $\val_i \in \{\val_i \in \ValSet: i \in [\ell, r)\} $ without changing the equality in the second step. 
    
    Therefore, the maximum social welfare is linear to $\bz$. Moreover, the maximum buyer surplus $\MaxBuyerSurplus(\ba) = \MaxSocialWelfare(\ba) - \SellerSurplus(\ba)$ is linear to $(\bx, \bz)$.

    Similarly, the minimum social welfare $\MinSocialWelfare(\ba)$ is linear to $\bw$ in \Cref{eq: minimum social welfare of fragment decomposition}, and specifically
    \begin{equation}
    \label{eq: minimum social welfare charged by z}
    \begin{aligned}
             \MinSocialWelfare(\ba) & = \sum\nolimits_{j, \ell, r, \FragSuppSet} \sum\nolimits_{j' < j}\sum\nolimits_{\val \in \FragSuppSet} \costpmf(\cost_{j'}) \cdot (\val - \cost_{j'}) \cdot \frag_{j,\ell, r, \FragSuppSet}(\val) \cdot \sw_{j, \ell, r, \FragSuppSet}\\
             & = \sum\nolimits_{j, \ell, r} \sum\nolimits_{\FragSuppSet \subseteq \ValSet \cap [\val_\ell, \val_r)} \sum\nolimits_{j' < j}\sum\nolimits_{i \in [\ell: r)} \costpmf(\cost_{j'}) \cdot (\val_i - \cost_{j'})  \frag_{j,\ell, r, \FragSuppSet}(\val_i) \cdot \sw_{j, \ell, r, \FragSuppSet} \\
             & = \sum\nolimits_{j, \ell, r} \sum\nolimits_{j' < j}\sum\nolimits_{i \in [\ell: r)} \costpmf(\cost_{j'}) \cdot (\val_i - \cost_{j'}) \sum\nolimits_{{\FragSuppSet \subseteq \ValSet \cap [\val_\ell, \val_r)}}  \frag_{j,\ell, r, \FragSuppSet}(\val_i) \cdot \sw_{j, \ell, r, \FragSuppSet} \\
              & = \sum\nolimits_{j, \ell, r} \sum\nolimits_{j' < j}\sum\nolimits_{i \in [\ell: r)} \costpmf(\cost_{j'}) \cdot (\val_i - \cost_{j'}) \cdot \rsdz_{j, \ell, r,i}~,
    \end{aligned}
    \end{equation}
    which is linear to $\bz$. Therefore, the minimum buyer surplus $\MinBuyerSurplus(\ba) = \MinSocialWelfare(\ba) - \SellerSurplus(\ba)$ is also linear to $(\bx, \bz)$.
\end{proof}

\section{Missing Proofs in \Cref{sec:flow}}
\label{apx: flow proof}
\begin{proof}[Proof of \Cref{thm: the equvilence of market segmentation and reduced fragment composition}]
    \textbf{Suppose the segmentation $(\alpha, \ValDist)$ induces the overall seller surplus $\SellerSurplus$ and the overall buyer surplus $\BuyerSurplus$.} In light of \Cref{lem: full generality of extremal markets}, there exists an extremal market segmentation $(\alpha_{\op, \SuppSet}, \ValDist_{\op, \SuppSet})$ such that the overall seller surplus $\SellerSurplus(\ba) = \SellerSurplus$, the buyer surplus satisfies $\MinBuyerSurplus(\ba) \le \BuyerSurplus \le \MaxBuyerSurplus(\ba)$. We can decompose $\ba$ into fragments with weight $\bw$ ($\ba \to \bw$ as in \Cref{def:fragment-decomposition}) and further find the reduced fragment decomposition $(\bx, \bz)$ ($\ba \to (\bx, \bz)$ as in \Cref{def: reduced fragment decompositions}) such that the proper domination (\Cref{lem: reduced fragment decompositions - proper domination}) is satisfied. In \Cref{lem: surplus of reduced fragment decompositions}, we proved the seller surplus and (maximum / minimum) buyer surplus induced by an extremal market decomposition $\ba$ are linear in its reduced fragment decomposition $(\bx, \bz)$. If we charge the seller surplus and the maximum/minimum buyer surplus as \Cref{def: surplus by reduced fragment decompositions (x z)}, then 
    \[
        \SellerSurplus((\bx, \bz)) = \SellerSurplus(\alpha) = \SellerSurplus~,
    \]
    and 
    \[
        \MinBuyerSurplus((\bx, \bz)) = \MinBuyerSurplus(\alpha) \le \BuyerSurplus \le  \MaxBuyerSurplus(\alpha) = \MaxBuyerSurplus((\bx, \bz))~.
    \]
    Moreover, \Cref{lem: the existance of y} proves the existence of $\by$ such that $(\bx, \by, \bz)$ satisfies exact composition and flow conservation. Therefore, we can prove if the segmentation exists with seller surplus $\SellerSurplus$ and buyer surplus $\BuyerSurplus$, then the reduced fragment decomposition $(\bx, \by, \bz)$ also exists satisfying $\SellerSurplus((\bx, \bz)) = \SellerSurplus$ and $\MinBuyerSurplus((\bx, \bz)) \le \BuyerSurplus \le \MaxBuyerSurplus((\bx, \bz))$.

\textbf{Suppose there exists $(\bx, \by, \bz)$ satisfying constraints with induced seller surplus $\SellerSurplus((\bx, \bz))$ and maximum/minimum buyer surplus $\MaxBuyerSurplus((\bx, \bz))$ or $\MinBuyerSurplus((\bx, \bz))$ defined in \Cref{def: surplus by reduced fragment decompositions (x z)}.} \Cref{lem: market segment - exists segmentation preserves surplus} shows that there exists a way to segment the aggregate market $\AggMarket$ which induces the same seller surplus and maximum/minimum buyer surplus as $(\bx, \bz)$. Then this proves the existence of segmentation with the overall seller surplus $\SellerSurplus$ such that $\SellerSurplus = \SellerSurplus((\bx, \bz))$ and the overall buyer surplus $\BuyerSurplus$ such that $\MinBuyerSurplus((\bx, \bz)) \le \BuyerSurplus \le \MaxBuyerSurplus((\bx, \bz))$. We complete the proof. 
\end{proof}

\begin{proof}[Proof of \Cref{lem: discount factor}]
    According to \Cref{def:fragment-decomposition}, $w_{j, \ell, k, \FragSuppSet_j} = \ValDist_{\op, \SuppSet}(\val_\ell)$ and $w_{j + 1, \ell, k, \FragSuppSet_{j+1}} = \ValDist_{\op, \SuppSet}(\val_k)$. Since both $\val_\ell$ and $\val_k$ are optimal prices for $\cost_j$ in the extremal market $\ValDist_{\op, \SuppSet}$, $(\val_\ell - \cost_j) \cdot \ValDist_{\op, \SuppSet}(\val_\ell) = (\val_k - \cost_j) \cdot \ValDist_{\op, \SuppSet}(\val_k)$. Then
    \[
     \frac{\esw_{j + 1, k, r, \FragSuppSet_{j+1}}^{\op, \SuppSet}}{\esw_{j, \ell, k, \FragSuppSet_{j}}^{\op, \SuppSet}} = \frac{\ValDist_{\op, \SuppSet}(\val_k)}{\ValDist_{\op, \SuppSet}(\val_\ell)} = \frac{\val_\ell - \cost_j}{\val_k - \cost_j}~.
    \]
    Therefore, $\sfrac{\esw_{j + 1, k, r, \FragSuppSet_{j+1}}^{\op, \SuppSet}}{\esw_{j, \ell, k, \FragSuppSet_{j}}^{\op, \SuppSet}}$ depends only on $j$, $\ell$, and $k$. 
\end{proof}

\begin{proof}[Proof of \Cref{lem: the existance of y}]
    Since $\ValDist \to (\bx, \bz)$, there exists $\ba = (\alpha_{\op, \SuppSet})$ and $\bw = (\sw_{j, \ell, r, \FragSuppSet})$ such that a {\em chain} of decompositions by $F \to \ba \to \bw \to (\bx, \bz)$ exists.

    \smallskip
    \textbf{Exact composition.} Fix $j \in [m]$ and $\ell, r \in [n]$. For each $i \in [\ell, r)$, $\rsdz_{j, \ell, r, i} = \sum_{\FragSuppSet \subseteq \ValSet \cap [\val_\ell, \val_r)} \sw_{j, \ell, r, \FragSuppSet} \cdot \frag_{j, \ell, r, \FragSuppSet}(\val_i) $, where $\frag_{j, \ell, r, \FragSuppSet}$ is the mass function of the fragment $\Frag_{j, \ell, r, \FragSuppSet}$ and $\sw_{j, \ell, r, \FragSuppSet}$ is the corresponding weight in fragment decomposition. By \Cref{remark: Succinct fragments}, the fragment $\Frag_{j, \ell, r, \FragSuppSet}$ is defined only for $\FragSuppSet \subseteq \ValSet \cap [\val_\ell, \val_r)$. Then
    \begin{align*}
        \valpmf(\val_i) & = \sum\nolimits_{j, \ell, r, \FragSuppSet} \sw_{j, \ell, r, \FragSuppSet} \cdot \frag_{j, \ell, r, \FragSuppSet}(\val_i) 
        = \sum\nolimits_{j, \ell, r}  \sum\nolimits_{\FragSuppSet \subseteq \ValSet \cap [\val_\ell, \val_r)} \sw_{j, \ell, r, \FragSuppSet} \cdot \frag_{j, \ell, r, \FragSuppSet}(\val_i) 
        = \sum\nolimits_{j, \ell, r} \rsdz_{j, \ell, r, i}~.
    \end{align*}
    The exact composition is proved. 

    \smallskip
    \textbf{Outgoing flow conservation.} Fix $j \in [m]$ and $\ell, r \in [n]$. Recall $\rsdw_{j, \ell, k}$ defined in \Cref{def: reduced fragment decompositions} and $\sw_{j, \ell, k, \FragSuppSet}$ defined in $\Cref{def:fragment-decomposition}$, 
    \begin{align*}
    \rsdw_{j, \ell, k} 
    & = \sum\nolimits_{\FragSuppSet_j \subseteq \ValSet \cap [\val_\ell, \val_k)} \sw_{j, \ell, k, \FragSuppSet_j} \\
    & =  \sum\nolimits_{\FragSuppSet_j \subseteq \ValSet \cap [\val_\ell, \val_k)} \sum\nolimits_{\op, \SuppSet} \alpha_{\op, \SuppSet} \cdot \esw_{j, \ell, k, \FragSuppSet_j}^{\op, \SuppSet} \\
    & =  \sum\nolimits_{\FragSuppSet_j \subseteq \ValSet \cap [\val_\ell, \val_k)} \sum\nolimits_{\op, \SuppSet: \op_j = \val_\ell, \op_{j+1} = \val_k, \SuppSet\cap [\val_\ell, \val_k) = \FragSuppSet_j} \alpha_{\op, \SuppSet} \cdot \esw_{j, \ell, k, \FragSuppSet_j}^{\op, \SuppSet} \\
    & =  \sum\nolimits_{\op, \SuppSet: \op_j = \val_\ell, \op_{j+1} = \val_k}\alpha_{\op, \SuppSet} \cdot \esw_{j, \ell, k, \SuppSet\cap [\val_\ell, \val_k)}^{\op, \SuppSet} \\
    & = \sum\nolimits_{\op, \SuppSet: \op_j = \val_\ell, \op_{j+1} = \val_k} \sum\nolimits_r  \sum\nolimits_{\op_{j+2} = \val_r}\alpha_{\op, \SuppSet}  \cdot \frac{\esw_{j+1, k, r, \SuppSet\cap [\val_k, \val_r)}^{\op, \SuppSet}}{\df_{j, \ell, k}} \\
     & = \sum\nolimits_r \sum\nolimits_{\op, \SuppSet: \op_j = \val_\ell, \op_{j+1} = \val_k, \op_{j+2} = \val_r} \alpha_{\op, \SuppSet} \cdot \esw_{j, \ell, k,\SuppSet\cap [\val_\ell, \val_k)}^{\op, \SuppSet}~,
    \end{align*}
     where $\sw_{j, \ell, k, \FragSuppSet_j}^{\op, \SuppSet}$ is weight of $\Frag_{j, \ell, k, \FragSuppSet_j}$ in fragment decomposition of the extremal market $\ValDist_{\op, \SuppSet}$. 
     
     Given $\FragSuppSet_j \subseteq \ValSet \cap [\val_\ell, \val_k)$, $\esw_{j, \ell, k, \FragSuppSet_{j}}^{\op, \SuppSet} = 0$ when $\op_j \ne \val_\ell$ or $\op_{j+1} \ne \val_k$ or $ \FragSuppSet_j \ne \SuppSet \cap [\val_\ell, \val_k)$. Thus, in the third step, for each $\FragSuppSet_j \subseteq \ValSet \cap [\val_\ell, \val_k)$, we can consider the pair $(\op, \SuppSet)$ in the set $\{(\op, \SuppSet): \op_j = \val_\ell, \op_{j+1} = \val_k, \FragSuppSet_{j} = \SuppSet \cap [\val_\ell, \val_k)\}$ only. In the forth step, when we sum over all $\FragSuppSet_j \subseteq \ValSet \cap [\val_\ell, \val_k)$, we can get the set $\{(\op, \SuppSet): \op_j = \val_\ell, \op_{j+1} = \val_k\}$ for the pair $(\op, \SuppSet)$ since  $\SuppSet \subseteq \ValSet $. 
     In the fifth step, we invoke the discount factor $\df_{j, \ell, k} = \sw_{j+1, k, r, \FragSuppSet_{j+1}}^{\op, \SuppSet}/ \sw_{j, \ell, k, \FragSuppSet_j}^{\op, \SuppSet}$, where $\val_r = \op_{j+1}$ and $\FragSuppSet_{j+1} = \SuppSet \cap [\val_k, \val_r)$, for each extremal market $(\op, \SuppSet)$ that satisfies $\op_j = \val_\ell$, $\op_{j+1} = \val_k$. Thus, in the last step we sum over all possible $r$ and only consider $(\op, \SuppSet)$ that satisfies $\op_j = \val_\ell$, $\op_{j+1} = \val_k$ and $\op_{j+2} = \val_r$.

    \smallskip
    \textbf{Incoming flow conservation.} Similarly, fix $j \in [m]$ and $\ell, r \in [n]$. For $\rsdw_{j+1, k, r}$, 
    \begin{align*}
    \rsdw_{j+1, k, r} 
    & = \sum\nolimits_{\FragSuppSet_{j+1} \subseteq \ValSet \cap [\val_k, \val_r)} \sw_{j+1, k, r, \FragSuppSet_{j+1}} \\
    & =  \sum\nolimits_{\op, \SuppSet: \op_{j+1} = \val_k, \op_{j+2} = \val_r} \alpha_{\op, \SuppSet} \cdot \esw_{j+1, k, r, \SuppSet \cap [\val_k, \val_r)}^{\op, \SuppSet} \\
    & = \sum\nolimits_{\op, \SuppSet: \op_{j+1} = \val_k, \op_{j+2} = \val_r} \sum\nolimits_\ell \sum\nolimits_{\op_j = \val_\ell}\alpha_{\op, \SuppSet}\cdot \df_{j, \ell, k} \cdot \esw_{j, \ell, k,\SuppSet \cap [\val_\ell, \val_k)}^{\op, \SuppSet} \\
    & = \sum\nolimits_\ell  \df_{j, \ell, k} \sum\nolimits_{\op, \SuppSet: \op_j = \val_\ell, \op_{j+1} = \val_k, \op_{j+2} = \val_r} \alpha_{\op, \SuppSet} \cdot \esw_{j, \ell, k, \SuppSet \cap [\val_\ell, \val_k)}^{\op, \SuppSet}~.
    \end{align*}
    Therefore, if there exists $\ba$ and $\bw$ such that we can find a chain of decomposition $\ValDist \to \ba \to \bw \to (\bx, \bz)$ (i.e., $\ValDist \to (\bx, \bz)$), and we define $\rsdf_{j, \ell, k, r} = \sum_{\op, \SuppSet: \op_j = \val_\ell, \op_{j+1} = \val_k, \op_{j+2} = \val_r} \alpha_{\op, \SuppSet} \cdot \esw_{j, \ell, k,\SuppSet \cap [\val_\ell, \val_k)}^{\op, \SuppSet}$, then it satisfies the outgoing flow conservation $x_{j, \ell, k} = \sum_r \rsdf_{j, \ell, k, r}$ and the ingoing flow conservation $\rsdw_{j+1, k, r} = \sum_\ell \df_{j, \ell, k} \cdot \rsdf_{j, \ell, k, r}$. The outgoing flow conservation and ingoing flow conservation are proved.
\end{proof}

\begin{lemma}
\label{lem: market segmentation - tail probability p_j}
    Fix an iteration $t$ of \Cref{alg: market segmentation}. Assume $\mscdf_1^{(t)}>0$. For each $j\in[j^{(t)}]$ (where $j^{(t)}$ is the terminal value of $j$ set in Line \ref{alg-line: market segmentation - assign j^(t)} of the algorithm), let $i_j^{(t)}$ denote the index $i_j$ selected in iteration $t$. Then
    \begin{align*}
        \MarketSegment^{(t)} \bigl(\val_{i_j^{(t)}}\bigr)=\frac{\mscdf_j^{(t)}}{\mscdf_1^{(t)}}~.
    \end{align*}
    Moreover, for every $j\in[j^{(t)}-1]$,
    \[
    \mscdf_{j+1}^{(t)}=\mscdf_{j}^{(t)}\cdot \df_{j,\, i_j^{(t)},\, i_{j+1}^{(t)}}.
    \]
\end{lemma}

\begin{proof}[Proof of \Cref{lem: market segmentation - tail probability p_j}]
 For each $i \in [n]$, $\marketSegmentPmf^{(t)}(\val_i) = \frac{1}{\alpha^{(t)}}\sum_{j \in [j^{(t)}]} \rsdz_{j, i_j^{(t)}, i_{j + 1}^{(t)}, i}^{(t)} \cdot \mscdf_j^{(t)} / \rsdw_{j, i_j^{(t)}, i_{j + 1}^{(t)}}^{(t)}$, where $\alpha^{(t)} = p_1^{(t)}$. We observe that $z_{j, i_j, i_{j + 1}, i} = 0$ if $i \notin [i_j^{(t)}, i_{j + 1}^{(t)})$. Thus, 
\begin{align*}
    \MarketSegment^{(t)}(\val_{i_j^{(t)}}) & = \sum_{i \ge {i_j^{(t)}}} \marketSegmentPmf^{(t)}(\val_i) \\
    & = \sum_{i \ge {i_j^{(t)}}} \frac{1}{\alpha}\sum_{j'\in [j^{(t)}]} \frac{\rsdz_{j', i_{j'}, i_{j' + 1}, i}^{(t)} \cdot \mscdf_{j'}^{(t)}}{\rsdw_{j', i_{j'}, i_{j' + 1}}^{(t)}}\\
    & = \frac{1}{\alpha^{(t)}}\sum_{j' \in [j: j^{(t)}]} \frac{\mscdf_{j'}^{(t)}}{x_{j', i_{j'}^{(t)}, i_{j' + 1}^{(t)}}^{(t)}}\sum_{i \in [i_{j'}^{(t)}: i_{j'+1}^{(t)})} \rsdz_{j', i_{j'}^{(t)}, i_{j' + 1}^{(t)}, i}^{(t)}\\
    & = \frac{1}{\alpha^{(t)}}\sum_{j' \in [j: j^{(t)}]} \frac{\mscdf_{j'}^{(t)}}{\rsdw_{j', i_{j'}^{(t)}, i_{j' + 1}^{(t)}}^{(t)}}\cdot 
    \frac{\val_{i_{j'+1}^{(t)}} - \val_{i_{j'}^{(t)}}}{\val_{i_{j'+1}^{(t)}} - \cost_{j'}} \cdot {\rsdw_{j', i_{j'}^{(t)}, i_{j' + 1}^{(t)}}^{(t)}} \\
    & = \frac{1}{\alpha^{(t)}} \sum_{j' \in [j: j^{(t)}]} \mscdf_{j'}^{(t)}\cdot 
    \frac{\val_{i_{j'+1}^{(t)}} - \val_{i_{j'}^{(t)}}}{\val_{i_{j'+1}^{(t)}} - \cost_{j'}}~.
\end{align*}
Now we would like to prove $\sum_{j' \in [j: j^{(t)}]} \mscdf_{j'}^{(t)}\cdot 
    (\val_{i_{j'+1}^{(t)}} - \val_{i_{j'}^{(t)}})/(\val_{i_{j'+1}^{(t)}} - \cost_{j'}) = \mscdf_j^{(t)}$. We prove it by deduction. 
    \begin{enumerate}
        \item When $j = j^{(t)}$,  since $\val_{i_{j^{(t)}+1}^{(t)}} = \infty$ and $\val_{i_{j^{(t)}}^{(t)}} < \infty$,
        \[
            \sum_{j' \in [j: j^{(t)}]} \mscdf_{j'}^{(t)}\cdot 
    \frac{\val_{i_{j'+1}^{(t)}} - \val_{i_{j'}^{(t)}}}{\val_{i_{j'+1}^{(t)}} - \cost_{j'}} = \mscdf_{j^{(t)}}^{(t)} \cdot (1- \frac{\val_{i_{j^{(t)}}^{(t)}} - \cost_{j^{(t)}}}{\val_{i_{j^{(t)}+1}^{(t)}} - \cost_{j^{(t)}}} ) = p^{t}_{j^{(t)}} ~; 
        \]
    \item assume $\sum_{j' \in [j+1: j^{(t)}]} \mscdf_{j'}^{(t)}\cdot (\val_{i_{j'+1}^{(t)}} - \val_{i_{j'}^{(t)}}) / (\val_{i_{j'+1}^{(t)}} - \cost_{j'}) =  \mscdf_{j+1}^{(t)}$.  Since $\mscdf_{j+1}^{(t)} / \mscdf_j^{(t)} = \df_{j, i_j, i_{j+1}} = \frac{\val_{i_j} - \cost_j}{\val_{i_{j+1}} - \cost_j}$,
    \begin{align*}
        \sum_{j' \in [j: j^{(t)}]} \mscdf_{j'}^{(t)} \cdot  \frac{\val_{i_{j'+1}^{(t)}} - \val_{i_{j'}^{(t)}}}{\val_{i_{j'+1}^{(t)}} - \cost_{j'}} = \mscdf_j^{(t)} \cdot \frac{\val_{i_{j+1}^{(t)}} - \val_{i_{j}^{(t)}}}{\val_{i_{j+1}^{(t)}} - \cost_{j}} + \mscdf_{j+1}^{(t)}  = \mscdf_j^{(t)}~.
    \end{align*}
    \end{enumerate}
Therefore, $\MarketSegment^{(t)}(\val_{i_j^{(t)}}) = \sfrac{\mscdf_j^{(t)}}{\mscdf_1^{(t)}}$, $j \in [j^{(t)}]$.

In \Cref{alg: market segmentation}, there are two places to initiate and update $\mscdf_{j+1}^{(t)}$,  we initiate $\mscdf_{j+1}^{(t)} := p_{j}^{(t)} \cdot \df_{j, i_j^{(t)}, i_{j+1}^{(t)}}$ in Line \ref{alg-line: market segmentation - p_j+1}, and we (might) update $\mscdf_{j+1}^{(t)}$ and $p_{j}^{(t)}$ with the same coefficient in line \ref{alg-line: maket segmentation - update p_j}. Thus, it keeps the relationship
\[
    \mscdf_{j+1}^{(t)} = \mscdf_{j}^{(t)} \cdot \df_{j, i_j^{(t)}, i_{j+1}^{(t)}}~.
\]
We complete the proof.
\end{proof}

\begin{proof}[Proof of \Cref{lem: market segmentation - x y z feasible}]
We denote the remaining fragment decomposition variables at the beginning of iteration $t$ by $\bx^{(t)} = \left(\rsdw_{j, \ell, r}^{(t)}\right)$, $\by^{(t)} = \left(\rsdf_{j, \ell, k, r}^{(t)}\right)$ and $\bz^{(t)} = \left(\rsdz_{j, \ell, r, i}^{(t)}\right)$. 
According to \Cref{alg: market segmentation}, in Line \ref{alg-line: market segmentation - x},
    \[
        \rsdw_{j, \ell, r}^{(t+1)} = \begin{cases}
            \rsdw_{j, \ell, r}^{(t)} - \mscdf_j^{(t)}, & \quad \ell = i_{j}^{(t)}, r = i_{j+1}^{(t)}~;\\
            \rsdw_{j, \ell, r}^{(t)}, & \quad \text{otherwise}~,
        \end{cases}
    \]
    in Line \ref{alg-line: market segmentation - y},
    \[
        \rsdf_{j, \ell, k, r}^{(t+1)} = \begin{cases}
            \rsdf_{j, \ell, k, r}^{(t)} - \mscdf_j^{(t)}, & \quad \ell = i_{j}^{(t)}, k = i_{j+1}^{(t)}, r = i_{j+2}^{(t)}~;\\
            \rsdf_{j, \ell, k, r}^{(t)}, & \quad \text{otherwise}~,
        \end{cases}
    \]
    and in Line \ref{alg-line: market segmentation - z},  
    \[
       \rsdz_{j, \ell, r, i}^{(t+1)} = \begin{cases}
           \rsdz_{j, \ell, r, i}^{(t)} - \alpha^{(t)} \cdot \marketSegmentPmf^{(t)}(\val_i), & \quad \ell = i^{(t)}_{j}, r=i^{(t)}_{j+1}~ i \in [i^{(t)}_{j}: i^{(t)}_{j+1})~;\\
           \rsdz_{j, \ell, r, i}^{(t)}, & \quad \text{otherwise}~,
       \end{cases}
    \]
where $i_j^{(t)}$ refers to the value index $i_j$ for the cost $j$ in iteration $t$.

\textbf{Outgoing flow conservation.} Assume that $(\bx^{(t)}, \by^{(t)}, \bz^{(t)})$ satisfies the outgoing flow conservation. Given the index $j, \ell, k$, then $\rsdw_{j, \ell, k}^{(t)} = \sum_r \rsdf_{j, \ell, k, r}^{(t)} $. We discuss the value of $\rsdw_{j, \ell, k}^{(t+1)}$ in different cases: 
\begin{itemize}
    \item when $\ell = i_{j}^{(t)}$ and $k = i_{j+1}^{(t)}$, 
        \begin{align*}
            \rsdw_{j, \ell, k}^{(t+1)} 
            & = \rsdw_{j, \ell, k}^{(t)} - \mscdf_j^{(t)} \\
            & = \sum_r \rsdf_{j, \ell, k, r}^{(t)} - \mscdf_j^{(t)} \\
            & = \sum_{r \ne i_{j+2}^{(t)}}\rsdf_{j, \ell, k, r}^{(t)} + \rsdf_{j, \ell, k, i_{j+2}^{(t)}}^{(t)}  - \mscdf_j^{(t)} \\
            & = \sum_{r \ne i_{j+2}^{(t)}}\rsdf_{j, \ell, k, r}^{(t+1)} + \rsdf_{j, \ell, k, i_{j+2}^{(t)}}^{(t+1)}  \\
            & = \sum_{r}\rsdf_{j, \ell, k, r}^{(t+1)} ~;
        \end{align*}
    \item when $\ell \ne i_{j}^{(t)}$ or $k \ne i_{j+1}^{(t)}$, for all $r \in [n]$, it satisfies $\rsdf_{j, \ell, k, r}^{(t+1)} = \rsdf_{j, \ell, k, r}^{(t)}$. Thus,
    \[ \rsdw_{j, \ell, k}^{(t+1)} =  \rsdw_{j, \ell, k}^{(t)} = \sum_{r}\rsdf_{j, \ell, k, r}^{(t)} = \sum_{r}\rsdf_{j, \ell, k, r}^{(t+1)}~. \]
\end{itemize}
Therefore, $(\bx^{(t+1)}, \by^{(t+1)}, \bz^{(t+1)})$ satisfies the outgoing flow conservation if $(\bx^{(t)}, \by^{(t)}, \bz^{(t)})$ satisfies it.

\textbf{Incoming flow conservation.} Similarly, assume that $(\bx^{(t)}, \by^{(t)}, \bz^{(t)})$ satisfies the incoming flow conservation. Given the index $j, k , r$, then $\rsdw_{j+1, k, r}^{(t)} = \sum_\ell \rsdf_{j, \ell, k, r}^{(t)} \cdot \df_{j, \ell, k} $. We consider $\rsdw_{j+1, k, r}^{(t+1)}$ in different cases: 
\begin{itemize}
    \item when $k = i_{j+1}^{(t)}$ and $r = i_{j+2}^{(t)}$, 
        \begin{align*}
            \rsdw_{j+1, k, r}^{(t+1)} 
            & = \rsdw_{j+1, \ell, k}^{(t)} - \mscdf_{j+1}^{(t)} \\
            & = \sum_\ell \rsdf_{j, \ell, k, r}^{(t)} \cdot \df_{j, \ell, k} - \mscdf_j^{(t)} \cdot \df_{j, i_j^{(t)}, i_{j+1}^{(t)}}  \\
            & = \sum_{\ell \ne i_{j}^{(t)}}\rsdf_{j, \ell, k, r}^{(t)}  \cdot \df_{j, \ell, k} + \rsdf_{j, i_{j}^{(t)}, k, r}^{(t)} \cdot \df_{j, i_j^{(t)}, k}  - \mscdf_j^{(t)} \cdot \df_{j, i_j^{(t)}, k}\\
            & = \sum_{\ell \ne i_{j}^{(t)}}\rsdf_{j, \ell, k, r}^{(t+1)} \cdot \df_{j, \ell, k} + \rsdf_{j, i_{j}^{(t)}, k, r}^{(t+1)} \cdot \df_{j, i_j^{(t)}, k} \\
            & = \sum_{r}\rsdf_{j, \ell, k, r}^{(t+1)} \cdot \df_{j, \ell, k}~;
        \end{align*}
    \item when $k \ne i_{j+1}^{(t)}$ or $r \ne i_{j+2}^{(t)}$, for all $\ell \in [n]$, it satisfies $\rsdf_{j, \ell, k, r}^{(t+1)} = \rsdf_{j, \ell, k, r}^{(t)}$. Thus,
    \[ \rsdw_{j+1, k, r}^{(t+1)} =  \rsdw_{j+1, k, r}^{(t)} = \sum_{\ell}\rsdf_{j, \ell, k, r}^{(t)} \cdot \df_{j, \ell, k} = \sum_{\ell}\rsdf_{j, \ell, k, r}^{(t+1)} \cdot \df_{j, \ell, k}~. \]
\end{itemize}
Therefore, $(\bx^{(t+1)}, \by^{(t+1)}, \bz^{(t+1)})$ satisfies the incoming flow conservation if $(\bx^{(t)}, \by^{(t)}, \bz^{(t)})$ satisfies it.

\textbf{Proper domination.} 
Assume for $t$, it has $\sum_{i\in [n]} \rsdz_{j, \ell, r, i}^{(t)} = (\val_r - \val_\ell)/(\val_r - \cost_j) \cdot \rsdw_{j, \ell, r}^{(t)} $. Given $j \in [j^{(t)}]$, 
\begin{itemize}
    \item  when $\ell \ne i^{(t)}_{j}$ or $r \ne i^{(t)}_{j+1}$, both $\rsdz_{j, \ell, r, i}^{(t+1)} = \rsdz_{j, \ell, r, i}^{(t)}$ and $\rsdw_{j, \ell, r}^{(t+1)} = \rsdw_{j, \ell, r}^{(t)}$, thus $\sum_{i\in [n]} \rsdz_{j, \ell, r, i}^{(t+1)} = (\val_r - \val_\ell)/(\val_r - \cost_j) \cdot \rsdw_{j, \ell, r}^{(t+1)}$;
    \item when $\ell = i^{(t)}_{j}$ and $r = i^{(t)}_{j+1}$,
    \begin{align*}
       \sum_{i\in [n]} \rsdz_{j, \ell, r, i}^{(t+1)} & = \sum_{i\in [\ell, r)} \rsdz_{j, \ell, r, i}^{(t+1)} \\
       & = \sum_{i\in [\ell, r)} \left(\rsdz_{j, \ell, r, i}^{(t)} - \alpha^{(t)} \cdot \marketSegmentPmf^{(t)}(\val_i)\right)\\
       & = \sum_{i\in [\ell, r)} \rsdz_{j, \ell, r, i}^{(t)} - \alpha^{(t)}  \sum_{i\in [\ell, r)}\marketSegmentPmf^{(t)}(\val_i)~.
   \end{align*}
   In \Cref{lem: market segmentation - tail probability p_j}, it is proved that $\MarketSegment^{(t)}(\val_{\ell}) = \MarketSegment^{(t)} \bigl(\val_{i_j^{(t)}}\bigr)=\mscdf_j^{(t)} / \mscdf_1^{(t)}$ and $\mscdf_{j+1}^{(t)}=\mscdf_{j}^{(t)}\cdot \df_{j, \ell, r}$. Thus, the sum of probability mass $\sum_{i\in [\ell, r)}\marketSegmentPmf^{(t)}(\val_i) = (\mscdf_j^{(t)} - \mscdf_{j+1}^{(t)})/\mscdf_1^{(t)} = \mscdf_j^{(t)} \cdot ( 1 - \df_{j, \ell, r}) / \mscdf_1^{(t)}$. Since $1-\df_{j, \ell, r} = (\val_r - \val_\ell)/(\val_r - \cost_j)$, then 
   \begin{align*}
       \sum_{i\in [n]} \rsdz_{j, \ell, r, i}^{(t+1)} 
       & = \sum_{i\in [\ell, r)} \rsdz_{j, \ell, r, i}^{(t)} - \mscdf_j^{(t)} \cdot \frac{\val_r - \val_\ell}{\val_r - \cost_j}\\
       & =\frac{\val_r - \val_\ell}{\val_r - \cost_j} \cdot (  \rsdw_{j, \ell, r}^{(t)} - \mscdf_j^{(t)}) \\
       & =\frac{\val_r - \val_\ell}{\val_r - \cost_j} \cdot \rsdw_{j, \ell, r}^{(t+1)}~.
    \end{align*}
\end{itemize}
Therefore, when it holds for $t$, it also holds for $t+1$.

Assume that given $t$, for any $i \in [n]$, $\rsdz^{(t)}_{j, \ell, r, i} \le \Frag_{j, \ell, r}(\val_i) \cdot \rsdw^{(t)}_{j, \ell, r}$. 
    For $t+1$, given $j \in [j^{(t)}]$, 
    \begin{itemize}
        \item  when $\ell \ne i^{(t)}_{j}$ or $r \ne i^{(t)}_{j+1}$, both $\rsdz_{j, \ell, r, i}^{(t+1)} = \rsdz_{j, \ell, r, i}^{(t)}$ and $\rsdw_{j, \ell, r}^{(t+1)} = \rsdw_{j, \ell, r}^{(t)}$, thus $\rsdz_{j, \ell, r, i}^{(t+1)} \le \Frag_{j, \ell, r}(\val_i) \cdot \rsdw_{j, \ell, r}^{(t+1)}$;
        \item when $\ell = i^{(t)}_{j}$ and $r = i^{(t)}_{j+1}$, according to Line \ref{alg-line: f(v_i)} and Line \ref{alg-line: alpha^t and market segment F}, $\marketSegmentPmf^{(t)}(\val_i) = \frac{1}{\mscdf_1^{(t)}} \sum_{j \in [j^{(t)}]} \rsdz_{j, i_{j}^{(t)}, i_{j + 1}^{(t)}, i}^{(t)} \cdot \mscdf_{j}^{(t)} / \rsdw_{j, i_{j}^{(t)}, i_{j + 1}^{(t)}}^{(t)}$. When $i \notin [i_{j}^{(t)}, i_{j + 1}^{(t)})$, 
        \[
        \rsdz_{j, i_{j}^{(t)}, i_{j + 1}^{(t)}, i}^{(t+1)}  = \rsdz_{j, i_{j}^{(t)}, i_{j + 1}^{(t)}, i}^{(t)}= 0 \le \Frag_{j, \ell, r}(\val_i) \cdot \rsdw_{j, \ell, r}^{(t+1)}~.
        \]
        When $i \in [i_{j}^{(t)}, i_{j + 1}^{(t)})$, 
        \[
            \marketSegmentPmf^{(t)}(\val_i) = \frac{\rsdz_{j, \ell, r, i}^{(t)} \cdot \mscdf_{j}^{(t)}}{\mscdf_1^{(t)} \cdot \rsdw_{j, \ell, r}^{(t)}}~.
        \]
        Thus, $\rsdz_{j, \ell, r, i}^{(t+1)} =
           \rsdz_{j, \ell, r, i}^{(t)} - \rsdz_{j, \ell, r, i}^{(t)} \cdot \mscdf_{j}^{(t)}/ \rsdw_{j, \ell, r}^{(t)} = (1 - \mscdf_{j}^{(t)}/ \rsdw_{j, \ell, r}^{(t)}) \cdot \rsdz_{j, \ell, r, i}^{(t)}$ and $\rsdw_{j, \ell, r}^{(t+1)} =
            \rsdw_{j, \ell, r}^{(t)} - \mscdf_j^{(t)} = (1 - \sfrac{\mscdf_{j}^{(t)}}{ \rsdw_{j, \ell, r}^{(t)}}) \cdot \rsdw_{j, \ell, r}^{(t)}$. Both are scaled by the same coefficient. Then
            \[
            \rsdz^{(t+1)}_{j, \ell, r, i} =  \left(1 - \frac{\mscdf_{j}^{(t)}}{\rsdw_{j, \ell, r}^{(t)}}\right) \cdot \rsdz_{j, \ell, r, i}^{(t)} \le  \left(1 - \frac{\mscdf_{j}^{(t)}}{\rsdw_{j, \ell, r}^{(t)}}\right) \cdot \Frag_{j, \ell, r}(\val_i) \cdot \rsdw^{(t)}_{j, \ell, r} =  \Frag_{j, \ell, r}(\val_i) \cdot \rsdw^{(t+1)}_{j, \ell, r}~.
            \]   
    Thus, when the domination relationship holds for $t$, it also holds for $t+1$.
    \end{itemize}
Therefore, if $\left(\bx^{(t)}, \by^{(t)}, \bz^{(t)}\right)$ satisfies the constraints: outgoing/incoming flow conservation, fragment composition, and proper domination, $\left(\bx^{(t+1)}, \by^{(t+1)}, \bz^{(t+1)}\right)$ also satisfies all.
\end{proof}

\begin{proof}[Proof of \Cref{lem: market segmentation stop condition}]
    It is clear that when $\bx = \textbf{0}$, \Cref{alg: market segmentation} stops. What we need to prove is when \Cref{alg: market segmentation} stops, $\bx = \textbf{0}$. 
    
    We prove it by contradiction, which means we stop the procedure either we can not find $\rsdw_{1, i_1, i_2} >0$ but for some $j>1$, it has $\rsdw_{j, i_{j}, i_{j+1}} >0$; or we can not find $\rsdf_{j, i_j, i_{j+1}, i_{j+2}} > 0 $ for a specific $j$  when $\bx \ne 0$. However, we proved that the incoming and outgoing flow conservation are always satisfied in each iteration in \Cref{lem: market segmentation - x y z feasible}, which indicates that if $\rsdw_{1, \ell, r} >0$, we can always find outgoing flow and also the corresponding $\rsdf_{j, i_j, i_{j+1}, i_{j+2}} > 0$ for all $j$; if $\rsdf_{j, i_j, i_{j+1}, i_{j+2}} > 0$, we can always find its incoming flow and also has $\rsdf_{j', i_j', i_{j'+1}, i_{j'+2}} > 0 $ for all $j' \in [j]$, thus $\rsdw_{1, i_1, i_2} >0$. These contradicts our assumption. Therefore, when \Cref{alg: market segmentation} stops, $\bx = \textbf{0}$.  
\end{proof}

\begin{proof}[Proof of \Cref{lem: market segment - exists segmentation preserves surplus}]
We denote the remaining fragment decomposition variables at the beginning of iteration $t$ by $\bx^{(t)} = (\rsdw_{j, \ell, r}^{(t)})$ and $\bz^{(t)} = (\rsdz_{j, \ell, r, i}^{(t)})$. According to Line \ref{alg-line: market segmentation - x} of \Cref{alg: market segmentation},
\[
    \rsdw_{j, \ell, k}^{(t+1)} = \begin{cases}
        \rsdw_{j, \ell, k}^{(t)} - \mscdf_j^{(t)}, & \ell = i_{j}^{(t)}, k = i_{j+1}^{(t)}~;\\
        \rsdw_{j, \ell, k}^{(t)}, & \text{otherwise}~,
    \end{cases}
\]
where $i_{j}^{(t)}$ refers to the corresponding or $i_{j}$ calculated in iteration $t$ of \Cref{alg: market segmentation}. 
According to Line \ref{alg-line: market segmentation - z} of \Cref{alg: market segmentation},
\[
\rsdz_{j, \ell, k, i}^{(t+1)} = \begin{cases}
    \rsdz_{j, \ell, k, i}^{(t)} - \alpha^{(t)} \cdot \marketSegmentPmf^{(t)}(\val_i), & \ell = i_{j}^{(t)}, k = i_{j+1}^{(t)}~;\\
    \rsdz_{j, i_j, i_{j + 1}, i}^{(t)}, & \text{otherwise}~.
    \end{cases}
\]
We first prove in each iteration $t$, the seller surplus ($\SellerSurplus(\MarketSegment^{(t)})$) and maximum/minimum buyer surplus ($\MaxBuyerSurplus(\MarketSegment^{(t)})$ / $\MinBuyerSurplus(\MarketSegment^{(t)})$) of $\MarketSegment^{(t)}$ weighted by $\alpha^{(t)}$ is equivalent to the change in seller surplus and maximum / minimum buyer surplus induced by $(\bx, \bz)$.

\textbf{The change in seller surplus and maximum / minimum buyer surplus induced by $(\bx, \bz)$.}
The seller surplus induced by $(\bx^{(t)}, \bz^{(t)})$ is $\SellerSurplus\left((\bx^{(t)}, \bz^{(t)})\right) = \sum_{j, \ell, r} \costpmf(\cost_j) \cdot (\val_\ell - \cost_j) \cdot \rsdw_{j, \ell, r}^{(t)}$ according to \Cref{def: surplus by reduced fragment decompositions (x z)}. Therefore, the change in seller surplus of iteration $t$ equals to
\begin{align*}
    \SellerSurplus\left((\bx, \bz)^{(t)}\right) - \SellerSurplus\left((\bx, \bz)^{(t+1)}\right) & = \sum_{j, \ell, r} \costpmf(\cost_j) \cdot (\val_\ell - \cost_j) \cdot \rsdw_{j, \ell, r}^{(t)} - \sum_{j, \ell, r} \costpmf(\cost_j) \cdot (\val_\ell - \cost_j) \cdot \rsdw_{j, \ell, r}^{(t+1)} \\
    & = \sum_{j \in [j^{(t)}]} \costpmf(\cost_j) \cdot (\val_{i_j^{(t)}} - \cost_j) \cdot \mscdf_j^{(t)}~,
\end{align*}
where $j^{(t)}$ refers to the corresponding $j$ when the while loop (Line \ref{alg-line: market segmentation - iteration t while starts} to Line \ref{alg-line: market segmentation - iteration t while ends}) ends in iteration $t$. 

The maximum social welfare charged by $(\bx^{(t)}, \bz^{(t)})$ is $\MaxSocialWelfare \left((\bx^{(t)}, \bz^{(t)})\right) = \sum_{j, \ell,r} \sum_{j' \in [j]}\sum_{i \in [\ell: r)} (\val_i - \cost_{j'}) \cdot \rsdz_{j, \ell, r, i}^{(t)}$ according to \Cref{eq: maximum social welfare charged by z}. Therefore, the change of maximum social welfare of iteration $t$ equals to
\begin{align*}
    &\quad \MaxSocialWelfare\left((\bx, \bz)^{(t)}\right) - \MaxSocialWelfare\left((\bx, \bz)^{(t+1)}\right) \\
    = &\quad \sum_{j, \ell,r} \sum_{j' \in [j]}\sum_{i \in [\ell: r)} \costpmf(\cost_{j'}) \cdot (\val_i - \cost_{j'}) \cdot \rsdz_{j, \ell, r, i}^{(t)} - \sum_{j, \ell,r} \sum_{j' \in [j]}\sum_{i \in [\ell: r)} \costpmf(\cost_{j'}) \cdot (\val_i - \cost_{j'}) \cdot \rsdz_{j, \ell, r, i}^{(t+1)} \\
    = &\quad \sum_{j \in [j^{(t)}]} \sum_{j' \in [j]} \sum_{i \in [i_j^{(t)}: i_{j+1}^{(t)})}  \costpmf(\cost_{j'}) \cdot (\val_{i} - \cost_{j'}) \cdot \alpha^{(t)} \cdot \marketSegmentPmf^{(t)}(\val_i)~,
\end{align*}
When we consider the social welfare for the cost $\cost_j'$, for all $j \in [j', j^{(t)}]$ and $i \in [i_j^{(t)}: i_{j+1}^{(t)})$ it has the corresponding social welfare $\costpmf(\cost_{j'}) \cdot (\val_{i} - \cost_{j'}) \cdot \alpha^{(t)} \cdot \marketSegmentPmf^{(t)}(\val_i)$ contributed.  Therefore,
\begin{align*}
    \MaxSocialWelfare\left((\bx, \bz)^{(t)}\right) - \MaxSocialWelfare\left((\bx, \bz)^{(t+1)}\right) =  \alpha^{(t)}\sum_{j' \in [j^{(t)}]} \costpmf(\cost_{j'})  \sum_{i \ge i_{j'}^{(t)}} (\val_{i} - \cost_{j'}) \cdot \marketSegmentPmf^{(t)}(\val_i)~.
\end{align*}
Similarly, the minimum social welfare charged by $(\bx^{(t)}, \bz^{(t)})$ is $\MinSocialWelfare \left((\bx^{(t)}, \bz^{(t)})\right) = \sum_{j, \ell,r} \sum_{j' \in [j-1]}\sum_{i \in [\ell: r)} (\val_i - \cost_{j'}) \cdot \rsdz_{j, \ell, r, i}^{(t)}$ according to \Cref{eq: minimum social welfare charged by z}. Therefore, the change in minimum social welfare of iteration $t$ equals to
\begin{align*}
    &\quad \MinSocialWelfare\left((\bx, \bz)^{(t)}\right) - \MinSocialWelfare\left((\bx, \bz)^{(t+1)}\right) \\
    = &\quad \sum_{j, \ell,r} \sum_{j' \in [j-1]}\sum_{i \in [\ell: r)} \costpmf(\cost_{j'}) \cdot (\val_i - \cost_{j'}) \cdot \rsdz_{j, \ell, r, i}^{(t)} - \sum_{j, \ell,r} \sum_{j' \in [j-1]}\sum_{i \in [\ell: r)} \costpmf(\cost_{j'}) \cdot (\val_i - \cost_{j'}) \cdot \rsdz_{j, \ell, r, i}^{(t+1)} \\
    = &\quad \sum_{j \in [j^{(t)}]} \sum_{j' \in [j-1]} \sum_{i \in [i_j^{(t)}: i_{j+1}^{(t)})}  \costpmf(\cost_{j'}) \cdot (\val_{i} - \cost_{j'}) \cdot \alpha^{(t)} \cdot \marketSegmentPmf^{(t)}(\val_i)~,
\end{align*}
When we consider the social welfare for the cost $\cost_j'$, for all $j \in [j'+1, j^{(t)}]$ and $i \in [i_j^{(t)}: i_{j+1}^{(t)})$ it has the corresponding social welfare $\costpmf(\cost_{j'}) \cdot (\val_{i} - \cost_{j'}) \cdot \alpha^{(t)} \cdot \marketSegmentPmf^{(t)}(\val_i)$ contributed.  Therefore,
\begin{align*}
    \MinSocialWelfare\left((\bx, \bz)^{(t)}\right) - \MinSocialWelfare\left((\bx, \bz)^{(t+1)}\right) =  \alpha^{(t)}\sum_{j' \in [j^{(t)}]} \costpmf(\cost_{j'})  \sum_{i \ge i_{j'+1}^{(t)}} (\val_{i} - \cost_{j'}) \cdot \marketSegmentPmf^{(t)}(\val_i)~.
\end{align*}
\textbf{Seller surplus and maximum/minimum buyer surplus of each constructed market $\MarketSegment^{(t)}$}. Since we have already proved $\MarketSegment^{(t)} \bigl(\val_{i_j^{(t)}}\bigr)=\mscdf_j^{(t)}/\mscdf_1^{(t)}$, we discuss the CDF of $\val_i$ for $i \in [i_j^{(t)}: i_{j+1}^{(t)})$, 
\begin{align*}
    \MarketSegment^{(t)}(\val_{i}) 
    &= \frac{\mscdf_{j+1}^{(t)}}{\mscdf_1^{(t)}}  + \sum_{i' \in [i, i_{j+1}^{(t)})} \marketSegmentPmf^{(t)}(\val_{i'})\\
    &= \frac{\mscdf_{j+1}^{(t)}}{\mscdf_1^{(t)}} + \sum_{i' \in [i, i_{j+1}^{(t)})} \frac{1}{\mscdf_1^{(t)}}\sum_{j'\in [j^{(t)}]} \frac{z_{j', i_{j'}^{(t)}, i_{j' + 1}^{(t)}, i'} \cdot \mscdf_{j'}^{(t)}}{x_{j', i_{j'}^{(t)}, i_{j' + 1}^{(t)}}}\\
    &= \frac{\mscdf_{j+1}^{(t)}}{\mscdf_1^{(t)}} + \sum_{i' \in [i, i_{j+1}^{(t)})} \frac{1}{\mscdf_1^{(t)}}\frac{\rsdz_{j, i_{j}^{(t)}, i_{j + 1}^{(t)}, i'} \cdot \mscdf_{j}^{(t)}}{\rsdw_{j, i_{j}^{(t)}, i_{j + 1}^{(t)}}} \\
    &= \frac{\mscdf_{j}^{(t)}}{\mscdf_1^{(t)}} \cdot  \left( \frac{\val_{i_j^{(t)}} - \cost_j}{\val_{i_{j+1}^{(t)}} - \cost_j} + \frac{1}{\rsdw^{(t)}_{j, i_{j}^{(t)}, i_{j + 1}^{(t)}}} \sum_{i' \in [i: i_{j+1}^{(t)})} \rsdz^{(t)}_{j, i_{j}^{(t)}, i_{j + 1}^{(t)}, i'} \right)~.
\end{align*}
In \Cref{lem: market segmentation - x y z feasible}, we have proved that the fragment domination preserves, then $\sum_{i' \in [i: i_{j+1}^{(t)})} \rsdz^{(t)}_{j, i_{j}^{(t)}, i_{j + 1}^{(t)}, i'} \le \Frag_{j, i_{j}^{(t)}, i_{j + 1}^{(t)}}(\val_i) \cdot \rsdw^{(t)}_{j, i_{j}^{(t)}, i_{j + 1}^{(t)}}$ for all $i \in [i_j^{(t)}: i_{j+1}^{(t)})$. Moreover, combining \Cref{lem:locality of fragment} and the definition of undominated fragment $\Frag_{j, i_{j}^{(t)}, i_{j + 1}^{(t)}}$, we have 
\[
    \Frag_{j, \ell, r}(\val_i)  = \frac{\val_{i_{j}^{(t)}} - \cost_j}{\val_i - \cost_j} - \frac{\val_{i_{j}^{(t)}} - \cost_j}{\val_{i_{j+1}^{(t)}} - \cost_j},~\quad i \in [i_{j}^{(t)}, i_{j+1}^{(t)})~.
\]
Then 
\begin{align*}
    \MarketSegment^{(t)}(\val_{i}) 
    &\le \frac{\mscdf_{j}^{(t)}}{\mscdf_1^{(t)}} \cdot  \left( \frac{\val_{i_j^{(t)}} - \cost_j}{\val_{i} - \cost_j}  \right),~\quad i \in [i_{j}^{(t)}, i_{j+1}^{(t)})~.
\end{align*}
Since $\MarketSegment^{(t)}(\val_{i_j^{(t)}}) = \frac{\mscdf_{j}^{(t)}}{\mscdf_1^{(t)}}$, and $\MarketSegment^{(t)}(\val_{i_{j+1}^{(t)}}) = \frac{\mscdf_{j+1}^{(t)}}{\mscdf_1^{(t)}} = \frac{\mscdf_{j}^{(t)}}{\mscdf_1^{(t)}} \cdot \frac{\val_{i_j^{(t)}} - \cost_j}{\val_{i_{j+1}^{(t)}} - \cost_j}$, for all $j \in [j^{(t)}]$, 
\[
    \MarketSegment^{(t)}(\val_{i}) \cdot (\val_{i} - \cost_j) 
    \le \MarketSegment^{(t)}(\val_{i_j^{(t)}}) \cdot ( \val_{i_j^{(t)}} - \cost_j),~\quad i \in [i_{j}^{(t)}, i_{j+1}^{(t)}]~,
\]
and the equality holds when $i = i_{j}^{(t)}$ or $i = i_{j+1}^{(t)}$.
Thus, $\MarketSegment^{(t)}$ is (weakly) dominated by $\ValDist_{\op, \SuppSet}$, where $\op = (\val_{i_j^{(t)}})_{j\in [j^{(t)}]}$ and $\SuppSet = (\val_i)_{i \in [n]}$.
Indicates that $\{\val_{i_j^{(t)}}, \val_{i_{j+1}^{(t)}}\} \subseteq \Op_j(\MarketSegment^{(t)})$. When we use $\val_{i_j^{(t)}}$ as the price for cost $\cost_j$, the seller surplus of the constructed $\MarketSegment$ equals
\begin{align*}
    \SellerSurplus(\MarketSegment^{(t)}) & = \sum_j\costpmf(\cost_j) \cdot (\val_{i_j^{(t)}} - \cost_j) \cdot \MarketSegment^{(t)}(\val_{i_j^{(t)}})
    = \sum_j\costpmf(\cost_j) \cdot (\val_{i_j^{(t)}} - \cost_j) \cdot \frac{\mscdf_j^{(t)}}{\mscdf_1^{(t)}}~.
\end{align*}
Therefore, the weighted seller surplus $\alpha^{(t)} \cdot \SellerSurplus(\MarketSegment^{(t)}) = \sum_j\costpmf(\cost_j) \cdot (\val_{i_j^{(t)}} - \cost_j) \cdot \mscdf_j^{(t)}$ equals the change of the seller surplus $\SellerSurplus\left((\bx, \bz)^{(t)}\right) - \SellerSurplus\left((\bx, \bz)^{(t+1)}\right)$ induced by $(\bx, \bz)$.

Due to the monotonicity of the optimal price set $\Op_j(\MarketSegment^{(t)})$ (\Cref{lem:monotonicity of optimal price set -v2}), we infer that $\val_{i_j^{(t)}}$ is the minimum optimal price and $\val_{i_{j+1}^{(t)}}$ is the maximum optimal price for cost $\cost_j$ in the constructed market $\MarketSegment^{(t)}$. Thus, the maximum(or minimum) social welfare is attained when we apply $\val_{i_{j}^{(t)}}$ (or $\val_{i_{j+1}^{(t)}}$) as the price for cost $\cost_j$, 
\begin{align*}
    \MaxSocialWelfare(\MarketSegment^{(t)}) 
    & = \sum_{j\in j^{(t)}}\costpmf(\cost_j) \sum_{i \ge i_{j}^{(t)}} (\val_{i} - \cost_j) \cdot \marketSegmentPmf^{(t)}(\val_{i}) ~; \\
    \MinSocialWelfare(\MarketSegment^{(t)}) 
    & = \sum_{j\in j^{(t)}}\costpmf(\cost_j) \sum_{i \ge i_{j+1}^{(t)}} (\val_{i} - \cost_j) \cdot \marketSegmentPmf^{(t)}(\val_{i})~.
\end{align*}
Therefore, the weighted maximum social welfare $\alpha^{(t)} \cdot \MaxSocialWelfare(\MarketSegment^{(t)})$ equals the change of maximum social welfare $\MaxSocialWelfare\left((\bx, \bz)^{(t)}\right) - \MaxSocialWelfare\left((\bx, \bz)^{(t+1)}\right)$ induced by $(\bx, \bz)$, and the weighted minimum social welfare $\alpha^{(t)} \cdot \MinSocialWelfare(\MarketSegment^{(t)})$ equals the change of the minimum social welfare $\MinSocialWelfare\left((\bx, \bz)^{(t)}\right) - \MinSocialWelfare\left((\bx, \bz)^{(t+1)}\right)$ induced by $(\bx, \bz)$. Then, the weighted maximum buyer surplus $\alpha^{(t)} \cdot \MaxBuyerSurplus(\MarketSegment^{(t)}) = \alpha^{(t)} \cdot (\MaxSocialWelfare(\MarketSegment^{(t)}) - \SellerSurplus(\MarketSegment^{(t)}))$ equals the change of the maximum buyer surplus $\MaxBuyerSurplus\left((\bx, \bz)^{(t)}\right) - \MaxBuyerSurplus\left((\bx, \bz)^{(t+1)}\right)$ induced by $(\bx, \bz)$, and the weighted minimum buyer surplus $\alpha^{(t)} \cdot \MinBuyerSurplus(\MarketSegment^{(t)}) =\alpha^{(t)} \cdot (\MinSocialWelfare(\MarketSegment^{(t)}) - \SellerSurplus(\MarketSegment^{(t)}))$ equals the change of the minimum buyer surplus $\MinBuyerSurplus\left((\bx, \bz)^{(t)}\right) - \MinBuyerSurplus\left((\bx, \bz)^{(t+1)}\right)$ induced by $(\bx, \bz)$.

Since \Cref{alg: market segmentation} stops iff $\bx=0$ (as in \Cref{lem: market segmentation stop condition}), we can finally conclude that if we segment the aggregate market $\AggMarket$ with \Cref{alg: market segmentation}, the seller surplus and The maximum / minimum buyer surplus of the segmentation ($\alpha^{(t)}, \MarketSegment^{(t)})$ are equivalent to the seller surplus and maximum / minimum buyer surplus induced by $(\bx, \bz)$.
\end{proof}
\end{document}